\documentclass{article}

\usepackage{amsmath}
\usepackage{amsthm}          
\usepackage{thmtools}
\usepackage{threeparttable}
\usepackage{booktabs}
\usepackage{array}
\usepackage{subcaption}
\usepackage{multirow}
\usepackage{graphicx}
\usepackage{float}
\usepackage{placeins}
\usepackage{empheq}

\newtheorem{lemma}{Lemma}[section]

\newtheorem{corollary}{Corollary}[section]

\newtheorem{definition}{Definition}
\theoremstyle{remark}
\newtheorem{remark}{Remark}[section]

\providecommand{\argmin}{}
\providecommand{\argmax}{}
\renewcommand{\argmin}{\operatorname*{arg\,min}}
\renewcommand{\argmax}{\operatorname*{arg\,max}}

\DeclareMathOperator{\erf}{\operatorname{erf}}
\DeclareMathOperator{\erfc}{\operatorname{erfc}}

\newcommand{\D}{\Delta}

\newcommand{\Th}{\mathcal{T}}

\newcommand{\Upol}{\textrm{U}^{\textrm{pol}}}
\newcommand{\Umax}{\textrm{U}^{\textrm{max}}}
\newcommand{\mupol}{\mu_{\textrm{pol}}}
\newcommand{\mumax}{\mu_{\textrm{max}}}

\newcommand{\iid}{\stackrel{\mathrm{iid}}{\sim}}

\newcommand{\expect}[1]{\mathbb{E}\!\left[#1\right]}

\newcommand{\paren}[1]{\left(#1\right)}
\newcommand{\braces}[1]{\left\{#1\right\}}
\newcommand{\brackets}[1]{\left[#1\right]}

\newcommand{\Align}{\mathsf{A}_t}
\newcommand{\Mlign}{\mathsf{A}_t^c}
\newcommand{\MlignT}[1]{\mathsf{A}_{#1}^c}

\newcommand{\PrDP}{\Pr^{\pi^*}}
\newcommand{\PrOracle}{\Pr^{\tilde{\pi}}}

\newcommand{\E}{\mathbb{E}}

\PassOptionsToPackage{numbers,compress}{natbib}

\usepackage[main, final]{neurips_2026}

\usepackage[utf8]{inputenc} 
\usepackage[T1]{fontenc}    
\usepackage{hyperref}       
\usepackage{url}            
\usepackage{booktabs}       
\usepackage{amsfonts}       
\usepackage{nicefrac}       
\usepackage{microtype}      
\usepackage{xcolor}         

\title{Budgeting Discretion: Theory and Evidence on Street-Level Decision-Making}

\author{%
  Gaurab Pokharel\\
  Institute for Advanced Computing\\
  Virginia Tech\\
  Alexandria, VA 22305 \\
  \texttt{gaurab@vt.edu} \\
  \And
  Sanmay Das\\
  Institute for Advanced Computing\\
  Virginia Tech\\
  Alexandria, VA 22305 \\
  \texttt{sanmay@vt.edu} \\
  \And
  Patrick J. Fowler\\
  Bursky School of Public Health\\
  Washington University in Saint Louis\\
  St. Louis, MO 63130 \\
  \texttt{pjfowler@wustl.edu} \\
}

\begin{document}
\maketitle

\begin{abstract}

Street-level bureaucrats, such as caseworkers, triage nurses, and school administrators, operate under a fundamental tension. The official policy on who to prioritize for services may conflict with bureaucrats' estimation of who would benefit most from those services.
The social theory of street-level bureaucracy argues that frontline agents use discretionary authority to bridge this gap, overriding policy recommendations when professional judgment identifies cases where doing so would meaningfully improve outcomes. As algorithmic recommendations become embedded in allocation settings, understanding the structure of this discretionary behavior is fundamental to designing tools that account for, rather than suppress, this value of professional judgment. We formalize discretion as a finite-horizon dynamic allocation problem. An agent implements a default policy but holds a limited budget of $K$ overrides across $T$ periods. Each period, the agent observes whether deviating from the default would yield a welfare gain and must decide whether to spend a scarce override now or conserve it for future opportunities. We show that the optimal policy is a state-dependent threshold rule and prove a behavioral invariance: for location-scale families of improvement distributions, the \emph{rate} at which an optimal agent exercises discretion is independent of the amount of welfare gains and depends only on the distribution's shape. Fat-tailed gains induce patience (conserving overrides for rare, high-stakes cases); thin-tailed gains induce routine, front-loaded spending. We also compare the conceptual implications of our model with operational data from a homelessness service system. In the model, a capacity opening makes an override cheaper to spend now, which lowers the threshold for using one. In the operational data, override rates indeed shift with short-run capacity openings in scarce housing, rising when recent exits free slots, and falling when interventions are congested. At a descriptive level, we also find that the temporal pattern of spending is front-loaded: override rates peak on Mondays and early in the fiscal year. Together, the model and the data point to a systematic structure in how discretion is used, with direct implications for the design of human-AI allocation systems that are both more efficient and equitable.

\end{abstract}

\section{Introduction}
\label{sec:introduction}

Frontline public-service workers often operate under two directives that can pull in opposite directions. Policies often mandate prioritizing the most vulnerable, whether that means the household with the highest acuity score, the patient with the most urgent need, or the child at greatest risk  \cite{hrsa2025_liverallocation, hud2017_cpd1701}. But the outcomes on which these same workers are evaluated often favor whoever would benefit most from the intervention, whether that is the household that will stay stably housed \cite{kube2023fair}, the patient whose prognosis improves the most from treatment \cite{Persad_Wertheimer_Emanuel_2009}, or the family most responsive to services. The most vulnerable person and the person who would benefit most are not always the same person. The theory of street-level bureaucracy has long recognized that caseworkers use professional discretion to bridge this gap, drawing on contextual knowledge and local judgment that no centralized rule can fully encode \cite{angelova2023algorithmic, cheng2022childwelfare, deartega2020humaninloop, lipsky1980street}.

 When caseworkers act on this judgment, they deviate from whatever standing rule governs the allocation -- a scoring instrument, a formulaic heuristic, or increasingly, an algorithmic recommendation \cite{pokharel2024discretionary, saxena2024algorithmic, starke2022fairnessperceptions, sutton2020cdss}. But discretion is not free. Prior work on discretionary behavior, whether in the street-level bureaucracy tradition \cite{chang2021formal, lipsky1980street} or in the growing literature on expert responses to algorithmic advice \cite{angelova2023algorithmic, dietvorst2015algorithm, Guo_Wu_Hartline_Hullman_2024, kesavan2020field}, has largely treated each override as a per-case choice. 

In practice, override capacity is finite, constrained by downstream resources (e.g., housing slots, ICU beds), managerial oversight, and operational bandwidth \cite{butler2021operational, green2022oversight, lipsky1980street}, and using an override today reduces the ability to do so tomorrow. The classic theory of street-level bureaucracy emphasizes precisely this point: discretion is shaped by caseload pressure and institutional scarcity \cite{butler2021operational, chang2021formal, lipsky1980street}. Yet, formal models of how finite discretion should be rationed over time are largely absent \cite{Pokharel_2025, pokharel2024discretionary}. Understanding the temporal structure of override behavior is a prerequisite for designing algorithmic tools that support professional judgment rather than suppress it.

We address this gap by formalizing the problem of \emph{Budgeted Discretion}. We ground the model in homelessness services, where caseworkers assign households to one of two interventions: Emergency Shelter (ES), a lower-intensity default, or Transitional Housing (TH), a scarce but stabilizing alternative \cite{kube2023fair, pokharel2024discretionary}. An agent has a budget of $K$ overrides over a finite horizon of $T$ periods. In each period $t$, the default algorithmic recommendation assigns an action (e.g., assign ES rather than TH), and the agent observes whether overriding that recommendation would improve welfare by an incremental amount $\Delta_t$ (the \emph{gain from discretion}). The gain $\Delta_t$ can either be positive (the agent can strictly improve the outcome by overriding) or effectively zero (the default recommendation already matches judgment).\footnote{$\Delta_t$ is induced by a baseline welfare draw and an independent ``improvement'' draw (more on this in Section~\ref{sec:the_model}).}  If the agent overrides, they spend one unit of budget and realize the improvement; if not, they conserve budget for future opportunities. In our domain of homelessness services, the caseworker can override the default to move a household from ES to TH. This framing connects sociological accounts of discretion under scarcity to a tractable finite-horizon control problem closely related to online selection with capacity constraints \cite{arlotto2020logarithmic, jiang2025tight, kleywegt1998dynamic}.

Our main theoretical contribution is a \textit{behavioral invariance}. We prove that, for location-scale families of improvement distributions \cite{casella2002location}, the rate at which an optimal agent exercises discretion at every state of the dynamic program is invariant to affine rescaling of payoffs and depends only on the distribution's shape. Intuitively, changing units (e.g., dollars versus thousands of dollars) rescales both the gains and the opportunity cost of spending by the same factor, so the probability that the agent overrides in any given state is unchanged. Once scale drops out, the key remaining driver of behavior is tail heaviness. With a fat-tailed gain distribution, most override opportunities yield modest improvements but rare cases produce very large gains, as when a housing override averts an imminent safety crisis. In this setting, optimal agents are patient, conserving overrides for outlier opportunities because the option value of waiting is high. With a thin-tailed distribution, improvements are relatively uniform, and agents override earlier and more routinely because waiting is unlikely to reveal a dramatically better case.

We compare the qualitative implications of our model with administrative data from the St.\ Louis Homeless Management Information System (HMIS), where caseworkers can deviate from an (algorithmically recovered) baseline assignment rule. We find two major effects in the data. First, when recent exits free slots in an intervention, caseworkers shift their overrides toward that intervention. In the model, this corresponds to a lower threshold for spending an override: when slots are open, each override costs less and draws less scrutiny.Second, caseworkers exercise more discretion early in the calendar, both within the week and within the fiscal year. The model predicts such front-loading under thin-tailed gains, but in these data staffing and intake schedules could produce it too, so these results should be taken as descriptive. Together, these findings identify systematic structure in observed override behavior, with direct implications for the design of human-AI allocation systems that support effective frontline decision-making.

\paragraph{Relationship to prior work:} Our model of budgeted discretion sits at the intersection of several literatures on sequential decision-making under scarcity. Here, we briefly trace the connection from budgeted allocation to threshold-based stopping, then to invariance arguments, and finally to the applied discretion literature. A detailed discussion appears in Appendix~\ref{sec:background_full}.

Structurally, our problem is online allocation with limited budget: each period reveals a stochastic opportunity (the incremental gain from overriding a default), and the agent decides whether to spend one unit of a finite budget or save it. This parallels dynamic and stochastic knapsack models \cite{kleywegt1998dynamic, dean2008benefit, arlotto2020logarithmic} and prophet inequality benchmarks \cite{hill1992survey, jiang2025tight, kleinberg2019matroid}. In contrast to this value- and regret-focused literature, we characterize \emph{the rate} at which discretionary capacity is consumed. This lens connects to optimal stopping and search, where reservation values balance current gains against the option value of waiting \cite{mccall1970economics, weitzman1979optimal}, and to multi-choice secretary problems with thresholds indexed by remaining selections \cite{gilbert1966recognizing, preater1994multiplechoice, tamaki1979double}. In our work, the system always implements a baseline policy, and the scarce action is an override whose payoff is the incremental gain $\D_t$ relative to that baseline. Our invariance result connects to a tradition of parameter-free sequential decision rules \cite{bruss2000odds, petruccelli1980best}, but where these formulations emphasize probability-of-success objectives, our objective is additive and the induced action rate is a unit-free object determined solely by distributional shape.

On the applied side, existing work on discretionary behavior, whether in the street-level bureaucracy tradition \cite{chang2021formal, lipsky1980street} or in the literature on expert responses to algorithmic advice \cite{angelova2023algorithmic, dietvorst2015algorithm, Guo_Wu_Hartline_Hullman_2024, kesavan2020field}, has largely treated each override as a static cost-benefit decision, and does not address the intertemporal rationing problem that arises when overrides are a finite resource. The same gap exists in the growing literature on human-AI decision-making, which studies how experts respond to algorithmic recommendations but not how they manage a finite capacity to deviate from them. From the caseworker's perspective, overriding a VI-SPDAT score \cite{orgcode2015vulnerability} and overriding an algorithmic recommendation pose the same intertemporal tradeoff, and the results we develop here apply to both. Our model formalizes the scarcity mechanism emphasized by classic accounts of street-level bureaucracy \cite{butler2021operational, lipsky1980street}, aligning with service-operations work where limited capacity induces threshold-like upgrading rules \cite{Shumsky_Zhang_2009} and with recent arguments that institutional context is central to the fairness of AI-driven allocation beyond automation-centric framings \cite{Pokharel_2025, Pokharel_Farabi_Fowler_Das_2025}.

\paragraph{Contributions:}
 \textbf{(1)}~We formalize discretionary overrides of algorithmic (or other scoring-based) recommendations as a finite-horizon dynamic program in which override capacity is a scarce resource, yielding an optimal threshold policy indexed by time remaining and budget remaining. \textbf{(2)}~We prove that for location-scale improvement distributions, the override rate at every state depends only on distributional shape -- not scale or location. Fat-tailed gains induce patient, outlier-targeting behavior; thin-tailed gains induce routine spending. \textbf{(3)}~Using HMIS data, we recover a baseline algorithmic policy via decision trees and show that observed override patterns shift with capacity openings, consistent with dynamic rationing of a finite discretionary budget, and are front-loaded in the calendar.

\section{Model and problem formulation}
\label{sec:the_model}

In each period $t \in \{1, \dots, T\}$, a single indivisible resource must be assigned to one of two arriving households. Each household $i \in \{1,2\}$ has a baseline welfare $w_i(t) \iid \mathcal{D}_w$ and a potential welfare improvement $I_i(t) \iid I_\theta$ if it receives the resource; all draws are mutually independent. The final welfare of a household receiving the unit is $w'_i(t) = w_i(t) + I_i(t)$. For any period, we use $m$ and $M$ to denote the indices of the households with the \textbf{m}inimum and \textbf{M}aximum baseline welfare respectively: $m = \argmin_{i \in \{1,2\}} \{w_i(t)\},$ and $M = \argmax_{i \in \{1,2\}} \{w_i(t)\}.$

\paragraph{The caseworker actions} The \textit{default policy} is a strictly needs-based rule that assigns the resource to the most vulnerable household (the one with the lower baseline welfare), reflecting standard prioritization logic in services such as homelessness or organ allocation \cite{hrsa2025_liverallocation, hud2017_cpd1701}. In contrast, a human decision-maker must balance competing priorities. For example, in homelessness, while the administrative rule is to prioritize the most vulnerable, evaluation is often in terms of success in keeping households stably housed in the future \cite{kube2023fair}, which corresponds to prioritizing the household with a higher increase in underlying welfare level. We thus model the decision-maker as being able to \textit{use discretion} in order to increase total welfare.

\subsection{The gain from discretion}
\label{subsec:discretion_gain}

The gain from discretion is the welfare difference between the two actions:
    \begin{equation}
        \label{eq:delta_def}
        \D_t = (I_M(t) - I_m(t))^+,
    \end{equation}
where $m$ and $M$ index the households with minimum and maximum baseline welfare, respectively (see Appendix~\ref{app:model_construction} for the full derivation). The gain is non-negative by construction: $\D_t > 0$ exactly when the less vulnerable household would benefit more from the resource -- an event we call \emph{misalignment}, denoted by $\Mlign$, with complement $\Align$ (alignment). Under i.i.d.\ continuous improvements, symmetry gives $p = \Pr(\Align) = 1/2$ (Remark~\ref{remark:extensions} discusses $N$ households, where $p=1/N$). Across periods, the realized gain follows a mixture:
    \begin{equation}
        \label{eq:delta_mixture}
        \Delta_t =
        \begin{cases}
            0, & \text{with probability } p \quad (\text{alignment}),  \\
            \iid f_{\D} & \text{with probability } 1-p \quad (\text{misalignment}),
        \end{cases}
    \end{equation}
where $f_\Delta$ is the density of the gain conditional on misalignment, with conditional CDF $F_{\Delta\mid \Mlign}$. We write $F_\Delta$ for the unconditional CDF of $\Delta_t$, so that $F_\Delta(x) = p + (1-p)F_{\Delta\mid\Mlign}(x)$ for $x \ge 0$. Let $\Upol_t$ and $\Umax_t$ denote the period-$t$ welfare under the default and welfare-maximizing allocations. By construction $\Umax_t = \Upol_t + \D_t$. Define $\mupol = \E[\Upol_t]$ and $\mumax = \E[\Umax_t] = \mupol + \E[\D_t]$ with $\E[\D_t] = (1-p)\E[\D_t \mid \Mlign]$ (we delegate these derivations to  Appendix~\ref{app:model_construction}).

\subsection{The dynamic program}
\label{sec:general_dp}
The goal of the agent is to find an optimal override policy $\pi^*$ that maximizes total expected welfare over $T$ periods subject to a budget of $K$ overrides:
    \begin{equation*}
        \label{eq:objective}
        \max_{(d_1,\dots,d_T)}\sum_{t=1}^{T} \E[U_t(d_t)] \quad\text{s.t.}\quad \sum_{t=1}^{T} d_t \le K,
    \end{equation*}
where $d_t \in \{0,1\}$ indicates whether an override is used at time $t$. Let the state be $(\tau,k)$, with $\tau \in \{0,\dots,T\}$ periods remaining and $k \in \{0,\dots,K\}$ budget units left, and let $V(\tau,k)$ denote the maximum expected total welfare from state $(\tau,k)$.

\paragraph{Boundary conditions:}
(i)~No time remaining: $V(0,k)=0$ for all $k$.
(ii)~No budget remaining: $V(\tau,0)=\tau \cdot \mupol$ for all $\tau \ge 1$; the agent collects the default-policy utility in every remaining period.
(iii)~Ample budget: $V(\tau,k) = \tau \cdot \mumax$ for all $k \ge \tau$; the agent can override at every misalignment.

\paragraph{The recursion:}
Conditioning on alignment or misalignment in the current period and simplifying (see Appendix~\ref{app:simplification_value_fxn} for the full algebra), the value function satisfies:
    \begin{equation}
        V(\tau,k) = \mupol + p  V(\tau-1,k) + (1-p)  \E_{\Delta|\Mlign} \Big[ \max \big\{V(\tau-1,k),   \Delta + V(\tau-1,k-1) \big\} \Big]
        \label{eq:recursion_main}
    \end{equation}
The interpretation is direct: the agent always collects the baseline utility $\mupol$. If a misalignment occurs (with probability $1-p$), the agent captures the additional gain $\D$ only if it exceeds the opportunity cost of consuming a unit of budget.

\paragraph{The optimal policy:}
The maximization in Equation~\eqref{eq:recursion_main} implies that the optimal policy is a state-dependent threshold rule:
    \begin{equation*}
        \label{eq:threshold_rule}
        d_t^*(\tau,k) =
            \begin{cases}
            1 & \text{if } \Mlign \text{ and } \D_t \ge \Th_{\tau,k}, \\
            0 & \text{otherwise},
            \end{cases}
    \end{equation*}
where the threshold
    \begin{equation}
        \Th_{\tau,k} = V(\tau-1,k) - V(\tau-1,k-1)
        \label{eq:threshold_eqn}
    \end{equation}
is the marginal continuation value of a budget unit for $\tau \ge 1$, $k \ge 1$. The agent overrides if and only if the current gain exceeds the value of carrying that unit of budget into the future. In the minimal case $T=2$, $K=1$, the threshold reduces to $\Th_{2,1} = \E[\D] = (1-p) \E[\D \mid \Mlign]$: the expected gain from having the override available in the final period (more on this in Section~\ref{subsec:policy_implications} and Appendix~\ref{app:two_period_case}).

\section{Behavioral invariance}
\label{sec:main_theorem}

What determines the \emph{rate} at which the optimal policy exercises overrides? The threshold $\Th_{\tau,k}$ depends on the units of welfare, but observable behavior is the \textit{probability} that the gain clears it.

\subsection{Assumption and base quantities}
\label{subsec:assumtion_and_base_quantities}
Assume per-household improvements follow a location-scale family \cite{casella2002location} of distributions: $I = a + sX$, where $a \in \mathbb{R}$ is a location parameter, $s > 0$ is a scale parameter, and $X \sim \mathcal{D}$ is a base shape with CDF $F_\mathcal{D}$, PDF $f_\mathcal{D}$ with non-negative support, and finite mean. Intuitively, a location-scale family is one in which every member can be obtained from a common base distribution by shifting $(a)$ and stretching $(s)$. Examples include the Exponential and Half-Normal families, and the Pareto, Gamma, Weibull, and Lognormal families when their shape parameters are held fixed.

Let $X, X' \iid \mathcal{D}$ be independent draws. Then, we define two quantities that depend only on $\mathcal{D}$ that we will use in our theorem. First, the \textbf{base threshold} ($c$) which is the expected positive difference between two draws $c \paren{\mathcal{D}} = \expect{(X' - X)^+}$, and second, the \textbf{difference CDF} ($G$) i.e., the CDF of the base difference $\Delta_0 = X' - X$:
    \begin{equation}
        G_{\mathcal{D}}(x) = \Pr(\Delta_0 \le x)
        = \int f_{\mathcal{D}}(u)  F_{\mathcal{D}}(u + x) du.
        \label{eq:G-def}
    \end{equation}

\noindent We first give a closed form for $c(\mathcal{D})$ that depends only on the CDF of the base shape.
\begin{restatable}{proposition}{PropositionOne}
\label{prop:c-closed-form}
$c(\mathcal{D})$ admits a convenient closed form. Suppose $X \sim \mathcal{D}$ has non-negative support and finite mean. Then the base threshold is:
\begin{equation}
    c(\mathcal{D})
    = \int_{0}^{\infty} F_{\mathcal{D}}(x)\big(1 - F_{\mathcal{D}}(x)\big) dx.
    \label{eq:c-closed-form}
\end{equation}
\end{restatable}

\begin{proof}[Sketch of proof]
A full derivation is in Appendix~\ref{app:proof-c-closed-form}. 
Starting from $c(\mathcal{D}) = \E[(X' - X)^+]$ with $X,X' \iid \mathcal{D}$, we condition on $X=x$ and write $\E[(X' - X)^+ \mid X=x] = \int_{x}^{\infty} (y-x)  f_{\mathcal{D}}(y) dy.$ Integration by parts rewrites this in terms of the tail CDF $1 - F_{\mathcal{D}}(y)$, and we swap the order of integration. The resulting double integral collapses to the single CDF-only expression in \eqref{eq:c-closed-form}. 
The restriction to $[0,\infty)$ uses the non-negative support of improvements.
\end{proof}

\subsection{The discretionary value and general invariance}
\label{subsec:discretionary_valu_and_general_invariance}
The recursion in Equation~\eqref{eq:recursion_main} mixes two sources of value: the baseline utility $\mupol$ that accrues every period regardless of whether discretion is exercised, and the incremental value that arises from having override budget available. The key to our main result is isolating the part of the value function attributable to having override budget. Define the \emph{discretionary value}:
    \begin{equation*}
    \label{eq:disc-value}
        W(\tau, k) = V(\tau, k) - \tau \cdot \mupol,
    \end{equation*}
which measures the expected welfare gain from holding $k$ overrides with $\tau$ periods remaining, minus the baseline which is collected at each step. By construction, $W(0, k) = 0$ and $W(\tau, 0) = 0$. Substituting into the recursion~\eqref{eq:recursion_main}, the constant $\tau \cdot \mupol$ cancels from both sides, yielding: 
    \begin{equation}
    \label{eq:W-recursion}
        W(\tau, k) = p \cdot W(\tau - 1, k) + (1 - p) \cdot \E_{\D \mid \Mlign}\brackets{\max\braces{W(\tau - 1, k), \D + W(\tau - 1, k - 1)}}.
    \end{equation}

Under the location--scale model, the gain from discretion factors as
    \begin{equation}
        \label{eq:gain-factorization-general}
        \D = (I_M - I_m)^+ = s \cdot (X_M - X_m)^+ = s \cdot \D_0^+,
    \end{equation}
where $\D_0 = X' - X$ (Appendices~\ref{app:W_recursion_and_simplification} and~\ref{app:gain-factor} give the algebra for Equations~\eqref{eq:W-recursion} and~\eqref{eq:gain-factorization-general}). The location $a$ cancels in the difference; the scale $s$ factors out because $s > 0$ preserves sign. We are now ready to state our main result.

\begin{restatable}[General behavioral invariance]{theorem}{GeneralInv}
\label{thm:general-invariance}
Let per-household improvements follow a location-scale family $I = a + sX$ with $s > 0$, where the base shape $\mathcal{D}$ has non-negative support and finite mean. Then for every state $(\tau, k)$ with $\tau \geq 1$ and $k \geq 1$:
    \begin{enumerate}\renewcommand{\labelenumi}{(\roman{enumi})}
        \item \textbf{The discretionary value scales linearly with $s$ and is invariant to $a$:} there exists a quantity $w(\tau, k; \mathcal{D})$, depending only on $\mathcal{D}$, $p$, and $(\tau, k)$, such that $W(\tau, k) = s \cdot w(\tau, k; \mathcal{D})$.
        
        \item \textbf{The optimal threshold scales linearly with $s$ and is invariant to $a$:} $\Th_{\tau,k} = s \cdot \theta_{\tau,k}(\mathcal{D})$, where $\theta_{\tau,k}(\mathcal{D}) = w(\tau - 1, k; \mathcal{D}) - w(\tau - 1, k - 1; \mathcal{D})$ is the base threshold.
        
        \item \textbf{The spending probability at every state is scale- and location-free:} $q_{\tau,k} = 1 - G_{\mathcal{D}}\paren{\theta_{\tau,k}(\mathcal{D})}$, where $G_{\mathcal{D}}$ is the CDF of the base difference $\D_0 = X' - X$ as defined in~\eqref{eq:G-def}.
    \end{enumerate}
\end{restatable}

\begin{proof}[Proof sketch (full proof in Appendix~\ref{app:general-invariance})]
\textbf{Part~(i)} proceeds by induction on $\tau$. The base case $W(0, k) = 0 = s \cdot 0$ is immediate. For the inductive step, assume $W(\tau - 1, j) = s \cdot w(\tau - 1, j; \mathcal{D})$ for all $j$. Substituting this and $\D = s \cdot \D_0^+$ into the $W$-recursion~\eqref{eq:W-recursion}, the factor $s > 0$ appears in every term inside the $\max$ and factors out (since multiplying by a positive constant preserves ordering). Pulling $s$ outside the expectation by linearity gives $W(\tau, k) = s \cdot [\cdots]$, where the bracket depends only on $\mathcal{D}$, $p$, and the base values $w(\tau - 1, \cdot; \mathcal{D})$. \textbf{Part~(ii)} follows from part~(i): $\Th_{\tau,k} = W(\tau - 1, k) - W(\tau - 1, k - 1) = s \cdot [w(\tau - 1, k; \mathcal{D}) - w(\tau - 1, k - 1; \mathcal{D})] = s \cdot \theta_{\tau,k}(\mathcal{D})$. \textbf{Part~(iii)} uses both preceding parts. The unconditional CDF of $\D = s \cdot \D_0^+$ satisfies $F_\D(x) = G_{\mathcal{D}}(x/s)$ for $x \geq 0$, since $\Pr(s \cdot \D_0^+ \leq x) = \Pr(\D_0^+ \leq x/s) = \Pr(\D_0 \leq x/s) = G_{\mathcal{D}}(x/s)$, where the second equality holds because $\{\D_0^+ \leq x/s\} = \{\D_0 \leq x/s\}$ for $x/s \geq 0$. Evaluating at the threshold: $F_\D(\Th_{\tau,k}) = G_{\mathcal{D}}(s \cdot \theta_{\tau,k} / s) = G_{\mathcal{D}}(\theta_{\tau,k})$, so $q_{\tau,k} = 1 - G_{\mathcal{D}}(\theta_{\tau,k})$.
\end{proof}

\begin{remark}[Beyond two households and a fixed budget]
\label{remark:extensions}
The proof of Theorem~\ref{thm:general-invariance} uses only two facts about the gain: it factors as $\Delta = s \cdot Z$, where the distribution of $Z$ does not depend on $a$ or $s$, and gains are independent across periods. Both facts are robust to two natural extensions. If $N$ households arrive in each period, the default serves the neediest household and an override serves the household with the largest improvement, so $\Delta = s \cdot (\max_i X_i - X_m)$, and $p = 1/N$ when improvements are i.i.d.\ and continuous. The same holds if baselines and improvements are dependent within a period, as long as their joint distribution does not depend on payoff units; $p$ is then whatever that distribution implies. If instead the budget is replenished each period by a random number of units, drawn independently of the welfare draws, the recursion~\eqref{eq:W-recursion} averages the continuation values over the replenishment, and $s$ still factors out of every term. In both cases the argument of Theorem~\ref{thm:general-invariance} carries over step by step. We leave the formal treatment of this to future work.
\end{remark}

\subsection{A closer look at a non-trivial two-period case}
\label{subsec:policy_implications}
 
Theorem~\ref{thm:general-invariance} establishes that override rates are shape-determined at every state of the general $(T, K)$ problem. Here, we look at the minimal non-trivial case $(T=2, K=1)$ where the full behavioral profile collapses to a single scalar. This case is worth developing in detail for two main reasons. First, it isolates the core spend-or-save tension in its purest form: the agent faces exactly one opportunity to override before the final period, making the tradeoff between immediate use and option value maximally transparent. Second, it admits closed-form expressions that yield concrete, interpretable predictions about when agents front-load versus back-load their discretion.

\begin{restatable}[Shape dependence of optimal policy, two-period]{theorem}{TheoremOne}
\label{thm:shape_dependence}
Let the improvement $I$ follow a location--scale family $I = a + sX$ with $s > 0$, where the base shape $\mathcal{D}$ has non-negative support and finite mean. Then, when $T=2$ and $K=1$:
\begin{enumerate}\renewcommand{\labelenumi}{(\roman{enumi})}
    \item \textbf{Optimal threshold scales linearly with $s$.} The two-period threshold $\Th_{2,1} = \E[\D]$ is invariant to location $a$ and satisfies $\Th_{2,1} = s \cdot c(\mathcal{D})$, where the base threshold $c(\mathcal{D})$ depends only on the shape and is given by~\eqref{eq:c-closed-form}.
    \item \textbf{Policy behavior is summarized by a single scalar.} All optimal spending probabilities are functions of $\psi(\mathcal{D}) = F_\D(\Th_{2,1})$, which is invariant to $a$ and $s$ and equals $\psi(\mathcal{D}) = G_\mathcal{D}(c(\mathcal{D}))$.
\end{enumerate}
\end{restatable}

This theorem can be proved in two ways: directly from the model assumptions (Appendix~\ref{app:proof-shape-dependence}) or as a special case of Theorem \ref{thm:general-invariance} (Appendix~\ref{app:two_period_invariance}). 

\paragraph{Behavioral implications of shape:} The invariance result yields a sharp behavioral dichotomy driven by distributional shape. In our numerical examples, heavier-tailed distributions tend to yield larger $\psi$. When $\psi$ is sufficiently large, the agent is patient: conservative relative to a full-information oracle in period 1 and back-loading spending to period 2. Smaller values of $\psi$ favor more frequent early overrides. In Appendix~\ref{app:behavioral_implications}, we formalize these comparisons through two lemmas that identify the exact thresholds at which the DP becomes internally patient (Lemma~\ref{lem:internal_dp}) and more conservative than the oracle (Lemma~\ref{lem:dp_v_oracle}). We also state a corollary (Corollary~\ref{cor:tails_and_patience}) connecting tail thickness to patience for families in which $\psi$ increases with tail thickness.

\paragraph{Numerical verification:}
Monte Carlo simulations ($N=10^6$ draws per configuration) confirm both claims of Theorem~\ref{thm:shape_dependence}; details and closed forms for the Exponential and Half-Normal cases are in Appendix~\ref{app:verification}. Scale invariance holds exactly: for every family tested (Exponential, Half-Normal, Gamma, Lognormal), varying the scale parameter 5- to 100-fold leaves $\hat\psi$ unchanged to three decimal places (Table~\ref{tab:scale-invariance}). Shape, by contrast, moves behavior decisively. Within the Weibull family (mean-normalized to $\E[X]=1$), moving from a lighter tail ($c=1.5$, $\hat\psi=0.668$) to a heavier tail ($c=0.8$, $\hat\psi=0.721$) reduces first-period spending by $16\%$; Lognormal exhibits an even wider range, from $\hat\psi=0.678$ ($\sigma=0.5$) to $\hat\psi=0.802$ ($\sigma=1.5$). 

In the extreme configurations of Table~\ref{tab:extreme-heavy}, estimated patience is high: Lognormal with $\sigma=5$ yields $\hat\psi=0.994$, and Pareto with $b=1.05$ yields $\hat\psi=0.962$. At the other extreme, Gamma with $k=10^4$ (near-degenerate) gives $\hat\psi\approx0.655$, the only example below the patience cutoff of $2/3$ (Corollary~\ref{cor:tails_and_patience}): this agent spends slightly more often in period 1 than in period 2. Tables~\ref{tab:scale-invariance}--~\ref{tab:extreme-heavy} in Appendix~\ref{app:numerical-verification} report the complete results.

\begin{figure}[ht]
    \centering
    \begin{subfigure}[b]{0.49\textwidth}
        \centering
        \includegraphics[width=\linewidth]{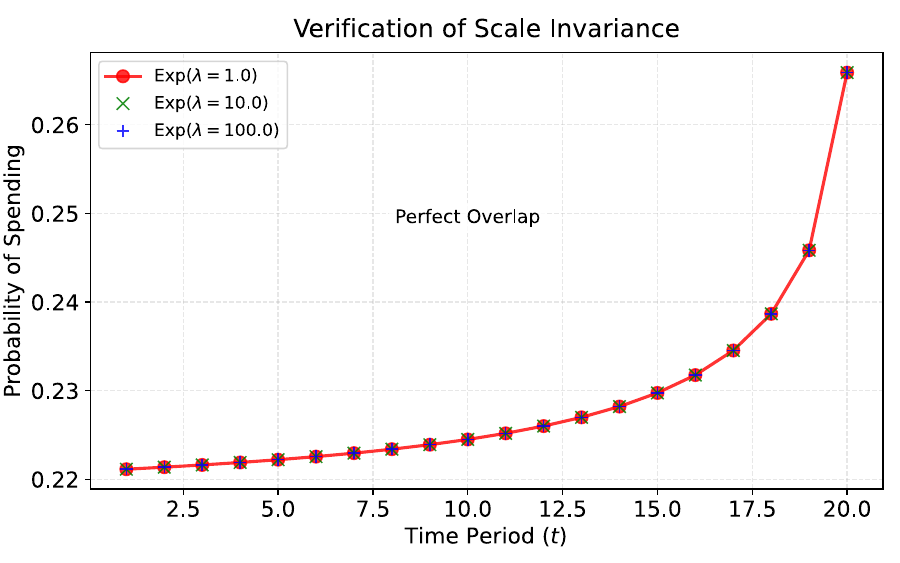}
        \caption{The optimal spending trajectory for an Exponential distribution remains identical despite varying the scale parameter $\lambda$ (the inverse of the scale) by two orders of magnitude, confirming Theorem~\ref{thm:general-invariance} (iii).}
        \label{fig:scale_invariance}
    \end{subfigure}
    \hfill 
    \begin{subfigure}[b]{0.49\textwidth}
        \centering
        \includegraphics[width=\linewidth]{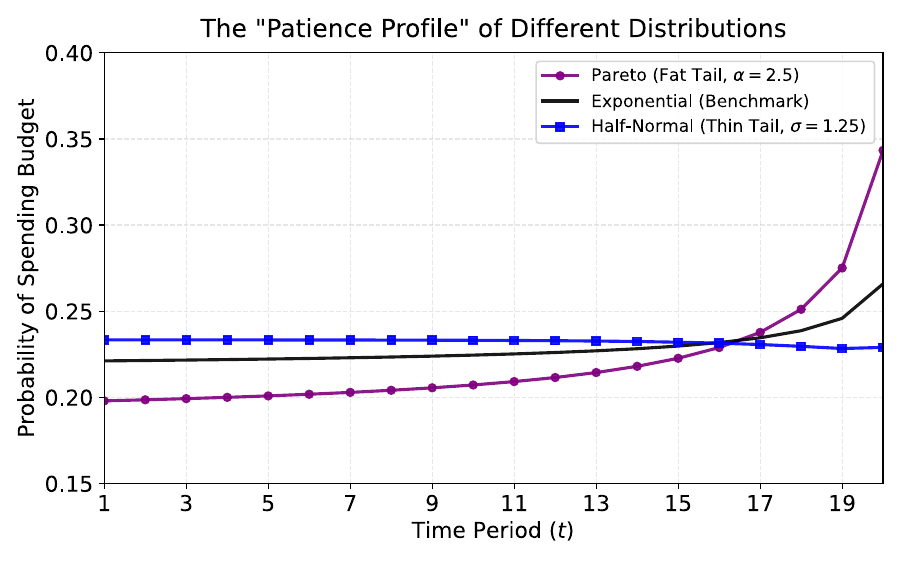}
        \caption{Spending probability over time for distributions with different tail shapes. Heavy-tailed distributions (Pareto) induce a patient policy; thin-tailed distributions (Half-Normal) spend at a nearly constant rate.}
        \label{fig:spending_profile_shape}
    \end{subfigure}
    \caption{Scale invariance and shape dependence ($T=20, K=5$): The optimal spending probability is invariant to the scale of rewards (a) and is driven by the tail of the gain distribution (b). The two panels use different vertical scales.}
    \label{fig:invariance_and_shape}
\end{figure}

\subsection{Scale invariance and shape dependence in the general problem}
\label{subsec:behavioral_implications}
Since every $q_{\tau,k}$ is scale- and location-free (Theorem~\ref{thm:general-invariance}), and the budget distribution $\pi_t(k)$ is computed from these probabilities via a deterministic forward pass (see Appendix~\ref{app:spending_profiles}), the entire observable spending profile inherits the same invariance. Figure~\ref{fig:invariance_and_shape} verifies this numerically and illustrates the behavioral consequences. This is the same dichotomy as demonstrated in the $(T=2, K=1)$ case. In our fat-tailed example, most override opportunities yield modest improvements but rare cases produce very large gains. The option value of waiting is high, so the optimal agent is patient -- conserving overrides for outlier cases and concentrating spending in the final periods as the value of saving an override declines. On the other hand, in our thin-tailed example, waiting is less likely to reveal a dramatically better opportunity, and the agent overrides earlier and more routinely. This pattern is visible in Figure~\ref{fig:spending_profile_shape}: Pareto gains (heavy tail) produce a sharply back-loaded spending curve, while Half-Normal gains (thin tail) produce a nearly flat one.

\section{Empirical evidence from homeless services}
\label{sec:empirical}

To complement our theoretical characterization of threshold-based override behavior, we examine whether discretionary overrides in a real allocation system vary with the operational constraints that our model identifies as drivers of threshold-based spending. The setting is homeless-service placement in St. Louis, where caseworkers assign households to either Emergency Shelter (ES, the lower-intensity default) or Transitional Housing (TH, the scarce, resource-intensive intervention). Because the latent gains $\D$ are not directly observed, our goal is not to estimate model primitives but to test whether the \textit{correlates} of overrides are consistent with the dynamic rationing behavior our model predicts. We ask two questions. First, how does discretionary behavior respond to shifts in the operational constraints that bind the caseworker? And second, how is the use of discretion spread over the calendar?

\textbf{Setting and discretion measurement:} 
We analyze administrative records from the HMIS covering 2008--2014, when caseworkers relied on professional judgment rather than formulaic prioritization tools. Following \citet{pokharel2024discretionary}, we define a baseline ``policy'' by training a short, interpretable decision tree (SER-DT~\cite{SERDT}, constrained to a depth of 4) on 34 household features. The tree's prediction $P_i$ represents the dominant assignment heuristic; a \emph{discretionary event} occurs when the caseworker's actual placement differs from this prediction ($A_i \neq P_i$). Appendix~\ref{appen:Operationalization_disc} discusses how errors in this recovered baseline could affect our estimates, and Appendix~\ref{apen:tree_depth_sensitvity} shows that the estimates in Table~\ref{tab:discretion_summary_daytype} stay within their 95\% confidence intervals when tree depth ranges from 2 to 6.

We distinguish \emph{upgrading} ($P_i{=}\text{ES}, A_i{=}\text{TH}$: spending the budget) from \emph{rationing} ($P_i{=}\text{TH}, A_i{=}\text{ES}$: conserving it). Whereas the earlier study asks which households receive overrides, we ask how override use varies over time. This temporal analysis is new to this paper: we aggregate overrides by day and construct measures of the system's operating conditions, including recent exits, to examine how discretion varies with capacity and the calendar. The theoretical model remains our main contribution, and these results offer an initial check of its predictions in one setting.

\textbf{Specification:}
We aggregate to the daily level and model override counts as binomial draws, with a logit link and covariates capturing the system's operational state. The key operational regressors are rolling 7-day ES and TH assignments (proxying workload and throughput burden) and rolling 7-day ES and TH exits (proxying short-run capacity \emph{openings}). In the language of the model, exits serve as a reduced-form proxy for the tightness of intervention-specific capacity: when exits are high, slots have recently been freed and the effective shadow price of that intervention falls; when exits are low, the intervention is congested and overrides into it are costlier. In the model, this shadow price is the threshold $\mathcal{T}_{\tau,k}$ in Equation~\eqref{eq:threshold_eqn}. Recent exits do not change $p$, the probability that the default already matches the caseworker's judgment, which depends on the mix of cases; they lower the threshold. Since the threshold falls as the remaining budget grows (Figure~\ref{fig:threshold_heatmap} in Appendix~\ref{appen:th_viz}), an opening has the same effect as a larger budget. We also include a lagged TH placement share, day-type indicators (Monday / Tue--Fri / Weekend), month-of-year fixed effects (October reference), and a federal-holiday indicator. All operational features use strictly lagged information to avoid mechanical dependence on same-day decisions. For directional models, denominators are restricted to the eligible trial set (e.g., upgrading is only possible when $P_i{=}\text{ES}$). The full specification and a weekly-block robustness check appear in Appendix~\ref{app:empirical-details}.

\subsection{Results}

Table~\ref{tab:discretion_summary_daytype} reports selected odds ratios from the rolling specification; Appendix~\ref{appen:weekly_glm} has the full tables.

\begin{table}[ht]
\centering
\caption{Summary of correlates of daily discretion (rolling specification; day-type baseline Tue--Fri).}
\label{tab:discretion_summary_daytype}
\resizebox{\textwidth}{!}{\begin{tabular}{l r@{\hspace{0.4em}}l r@{\hspace{0.4em}}l r@{\hspace{0.4em}}l}
\toprule
& \multicolumn{2}{c}{\textbf{Overall} $\big( p_t^{\text{all}}\big)$}
& \multicolumn{2}{c}{\textbf{Upgrading} $\big( p_t^{\uparrow} \big)$}
& \multicolumn{2}{c}{\textbf{Rationing} $\big( p_t^{\downarrow}\big)$} \\
\textbf{Predictor}
& \multicolumn{1}{c}{\textbf{OR}} & \multicolumn{1}{c}{\textbf{[95\% CI]}}
& \multicolumn{1}{c}{\textbf{OR}} & \multicolumn{1}{c}{\textbf{[95\% CI]}}
& \multicolumn{1}{c}{\textbf{OR}} & \multicolumn{1}{c}{\textbf{[95\% CI]}} \\
\midrule

\multicolumn{7}{l}{\textbf{Operational state (prior 7 days $[t-7,t-1]$)}}\\
THShare$_{t-1}$              & 1.060        & [0.859, 1.308] & 0.931        & [0.683, 1.271] & 1.150        & [0.853, 1.551]\\
ES exits ($[t-7,t-1]$)       & 0.995        & [0.982, 1.008] & 0.971**      & [0.953, 0.989] & 1.026**      & [1.007, 1.045]\\
TH exits ($[t-7,t-1]$)       & 0.995        & [0.978, 1.012] & 1.021$^\dagger$ & [0.998, 1.045] & 0.952***     & [0.928, 0.977]\\
\midrule

\multicolumn{7}{l}{\textbf{Timing (reference: Tue--Fri)}}\\
Mon                          & 1.244**      & [1.074, 1.440] & 1.140        & [0.919, 1.415] & 1.241*       & [1.006, 1.531]\\
Weekend                      & 0.762*       & [0.605, 0.961] & 0.813        & [0.574, 1.153] & 0.639**      & [0.465, 0.879]\\
\midrule

\multicolumn{7}{l}{\textbf{Seasonality (reference: October)}}\\
Nov                          & 0.705*       & [0.512, 0.970] & 0.880        & [0.544, 1.424] & 0.677$^\dagger$ & [0.427, 1.075]\\
Dec                          &0.592** &[0.428, 0.818]     & 0.730        & [0.448, 1.190] & 0.547*       & [0.345, 0.869]\\
Jan                          & 0.789 & [0.585, 1.063] & 1.125        & [0.720, 1.758] & 0.612*       & [0.398, 0.943]\\
Jun                          & 0.710*       & [0.524, 0.962] & 1.027        & [0.647, 1.630] & 0.522**      & [0.339, 0.804]\\
\midrule
\bottomrule
\end{tabular}}

\vspace{0.35em}
\begin{minipage}{0.95\linewidth}\small
\emph{Notes.} Entries are odds ratios from binomial-logit GLMs estimated on daily aggregates. Rolling operational features use the prior-week window $[t-7,t-1]$. All models include the full set of month-of-year indicators (October reference) as well as a federal-holiday indicator; the table reports selected coefficients for readability. Full coefficient tables, the weekly block specification, and an alternative day-of-week fixed-effects specification are reported in Appendices~\ref{appen:weekly_glm}--\ref{appen:full_glm}. Stars indicate significance: $^{***}p<0.001$, $^{**}p<0.01$, $^{*}p<0.05$, $^\dagger p<0.1$.
\end{minipage}
\end{table}

\paragraph{Capacity openings and substitution:} The rationing direction shows the most consistent association with capacity constraints. n Table~\ref{tab:discretion_summary_daytype}, each additional ES exit in the prior week (freeing shelter beds) is associated with more rationing (OR~$=1.026$, $p<0.01$), while each additional TH exit (freeing transitional slots) is associated with less rationing (OR~$=0.952$, $p<0.001$). This sign pattern matches a substitution logic. When TH slack increases, the need to conserve TH falls; when ES slack increases, defaulting households down to ES becomes easier. For upgrading, the pattern is complementary. Upgrades decline when ES exits are high (OR~$=0.971$, $p<0.01$) and rise (marginally) when TH exits are high (OR~$=1.021$, $p<0.1$). These odds ratios are per exit, and the number of exits in a week can vary, so the effect of a typical change is larger. Per standard deviation of the prior-week exit count (5.7 ES exits and 4.4 TH exits), ES exits are associated with 16\% higher odds of rationing and TH exits with 19\% lower odds (OR~$=1.157$ and $0.807$), and with 16\% lower and 10\% higher odds of upgrading (OR~$=0.843$ and $1.097$); the $p$-values are those of the per-exit estimates. These are reported in Tables~\ref{tab:updrading_discretion_mon} and~\ref{tab:rationing_discretion_mon} in Appendix~\ref{appen:glm_emperics}. These associations are conditional on the baseline recommendation, so they do not merely reflect mechanical availability effects. When slots are plentiful, the effective cost of diverting one to an override is lower, which in the model corresponds to a lower threshold for spending an override. \citet{pokharel2024discretionary} found that caseworkers in this system are more likely to exercise discretion on lower-acuity households, which suggests a complementary channel: both higher slot availability and lower case acuity reduce the professional scrutiny attached to deviating from the default. On its own, the capacity association does not separate our model from the idea that discretion responds to scarcity, so we treat it as a consistency check.

\paragraph{Temporal pattern of discretion:}
Discretion is front-loaded across multiple calendar cycles. Within the week, overall discretion is elevated on Mondays relative to typical weekdays (OR~$=1.244$, $p<0.01$). At the fiscal-year level, relative to October (the federal fiscal-year start), overall discretion is modestly lower in most subsequent months, with the sharpest reductions in November-December and June. This temporal pattern i.e., higher override rates at the beginning of the relevant calendar cycle is what our theory associates with thin-tailed improvement distributions, where the option value of waiting is modest and routine early spending is optimal.

Discretion is also lower on weekends, overall (OR $=0.762$, $p<0.05$) and especially for rationing (OR $=0.639$, $p<0.01$; Appendix~\ref{appen:weekly_glm}). However, we note that operational schedules could also produce this weekly pattern. Fewer staff work on weekends and intake slows, so weekend placements may stick closer to the default, while decisions that need more attention wait until Monday, when the weekend's cases are processed. The fiscal year could work the same way: October starts the federal fiscal year, when budgets and reporting cycles reset. Staffing and intake limits are themselves a form of the operational scarcity the model describes. These operational constraints could make Mondays busier, but volume by itself cannot raise our outcome, which is the share of a day's placements that override the baseline; more cases raise the denominator along with the numerator. What staffing and intake schedules can change, however, is the mix of cases. A household that waits out the weekend may arrive on Monday in worse shape than a typical weekday case, so Monday's pool may contain more cases in which judgment disagrees with the baseline, and that alone would raise the override rate without any budget rationing. Our data cannot separate these explanations from the model's prediction that thin-tailed gains produce front-loaded spending (see also Appendix~\ref{appen:remaining_slack}), so the calendar results should be read as a description of when discretion is used.

Taken together, two patterns stand out. First, capacity openings redirect discretion: when recent exits free slots in one intervention, caseworkers override toward that intervention more readily, consistent with a lower threshold for spending. Second, discretion is front-loaded in the calendar, peaking on Mondays within each week and elevated early in the fiscal year. This temporal pattern resembles the routine early spending our theory associates with thin-tailed improvement distributions, where waiting is unlikely to reveal a dramatically better case. The overall picture is one in which discretion behaves like a scarce, dynamically managed resource, and caseworkers exercise it more freely when conditions reduce the cost of doing so.

\section{Discussion and conclusion}
\label{sec:discussion-conclusion}

Frontline public-service workers have always rationed their judgment, but formal models of how that rationing should unfold over time have been largely absent. We address this by modeling discretion as a dynamic allocation problem in which overrides of an algorithmic default are a finite resource: spending one today reduces what is available tomorrow. The optimal policy is a state-dependent threshold rule, and the main theoretical result is a behavioral invariance. For location-scale families of improvement distributions, the rate at which an optimal agent exercises discretion is independent of payoff units and determined entirely by the shape of the gain distribution. Because this invariance holds at every state of the $(T,K)$ dynamic program, the entire observable spending profile is a unit-free object, determined by distributional shape, the alignment probability, and the problem dimensions. Fat-tailed gains produce patient, outlier-targeting agents; thin-tailed gains produce routine, front-loaded spenders. In homelessness services, overrides shift toward interventions whose slots have recently opened, as a lower threshold for spending would imply; we also find some evidence suggesting that discretion is front-loaded in the week and the fiscal year.

These results rest on modeling choices that isolate the spend-or-save tradeoff in its sharpest form: i.i.d. gains from a location-scale family of improvements, and a hard budget constraint. Real-world discretion can additionally be shaped by evolving caseload, soft institutional constraints, and learning. Relaxing these assumptions is a natural next step. With a noisy signal of the gain, the optimal rule could keep its threshold form, applied to the expected gain given the signal. There, we expect the invariance to hold when the noise scales with the gain, but it is unclear when the noise has a fixed size. Furthermore, our data show only the front-loaded regime; settings where a rare override can prevent a disaster, such as intensive care triage, should show the patient one, and are worth investigating. Finally, the shape of the gain distribution could also be estimated from downstream outcomes, correcting for the fact that overrides are observed only above the threshold. This paper provides the framework that such extensions would build on.

These results have practical implications for how allocation systems are designed. Whether the standing default is a scoring instrument or an algorithmic recommendation, current tools focus on flagging \textit{which} cases merit attention. Our model suggests they should also surface \textit{when} to act, making the option value of waiting legible to the caseworker managing a finite override budget. The shape-based characterization clarifies why similar organizations may exhibit different discretionary ``patience'' even under comparable resources: the answer lies not in the scale of potential gains, which washes out, but in their distributional shape. Any system designed to support effective use of discretion should take these factors into account.

\begin{ack}
We are grateful for support from NSF Award 2533162. We also thank our various community partners who helped conceptualize the challenges facing the delivery of homeless services, as well as their ongoing efforts to support local families.
\end{ack}

\bibliographystyle{abbrvnat}
\bibliography{references}


\appendix
\raggedbottom
\let\oldsection\section
\renewcommand{\section}{\FloatBarrier\oldsection}

\newpage
\section{Table of notation}
\label{app:notation}

\begin{table}[H]
\centering
\caption{Summary of notation used throughout the paper.}
\label{tab:notation}
\small
\begin{tabular}{@{}l p{0.70\textwidth} l@{}}
\toprule
\textbf{Symbol} & \textbf{Description} & \textbf{Introduced} \\
\midrule
\multicolumn{3}{@{}l}{\textit{Problem Dimensions}} \\[2pt]
$T$ & Total number of periods (finite horizon) & \S\ref{sec:introduction} \\
$K$ & Total override budget (number of discretionary actions) & \S\ref{sec:introduction} \\
$t$ & Period index, $t \in \{1, \dots, T\}$ & \S\ref{sec:the_model} \\
$\tau$ & Periods remaining (incl.\ current), $\tau = T - t + 1$ & \S\ref{sec:general_dp} \\
$k$ & Remaining budget, $k \in \{0, 1, \dots, K\}$ & \S\ref{sec:general_dp} \\
$d_t$ & Override decision indicator at time $t$ ($1$ = override, $0$ = comply) & \S\ref{sec:general_dp} \\
\midrule
\multicolumn{3}{@{}l}{\textit{Households and Welfare}} \\[2pt]
$i$ & Household index, $i \in \{1, 2\}$ & \S\ref{sec:the_model} \\
$w_i(t)$ & Baseline welfare of household $i$ in period $t$; $w_i(t) \iid \mathcal{D}_w$ & \S\ref{sec:the_model} \\
$I_i(t)$ & Welfare improvement for household $i$ if given the resource; $I_i(t) \iid I_\theta$ & \S\ref{sec:the_model} \\
$m, M$ & Indices of min and max baseline-welfare households & \S\ref{sec:the_model} \\
$\Upol$ & Total welfare under the default (needs-based) policy & \S\ref{subsec:discretion_gain} \\
$\Umax$ & Total welfare under the welfare-maximizing allocation & \S\ref{subsec:discretion_gain} \\
$\mupol$ & Expected per-period welfare under default policy, $\E[\Upol]$ & \S\ref{subsec:discretion_gain} \\
$\mumax$ & Expected per-period welfare under best allocation, $\E[\Umax]$ & \S\ref{subsec:discretion_gain} \\
\midrule
\multicolumn{3}{@{}l}{\textit{Alignment and the Gain from Discretion}} \\[2pt]
$\Align$ & Alignment event: default policy is welfare-maximizing ($I_m \geq I_M$) & \S\ref{subsec:discretion_gain} \\
$\Mlign$ & Misalignment event: override can improve welfare ($I_M > I_m$) & \S\ref{subsec:discretion_gain} \\
$p$ & Alignment probability, $p = \Pr(\Align)$; equals $1/2$ under i.i.d.\ continuous improvements & \S\ref{subsec:discretion_gain} \\
$\D_t$ & Gain from discretion at time $t$: $\D_t = (I_M - I_m)^+$ & \S\ref{subsec:discretion_gain} \\
$f_\D$ & Density of $\D_t$ conditional on $\Mlign$ & \S\ref{subsec:discretion_gain} \\
$F_\D$ & Unconditional CDF of $\D_t$, including its mass at 0 & \S\ref{subsec:discretion_gain} \\
\midrule
\multicolumn{3}{@{}l}{\textit{Location--Scale Family and Base Shape}} \\[2pt]
$I = a + sX$ & Location--scale model for improvements & \S\ref{subsec:assumtion_and_base_quantities} \\
$a$ & Location parameter & \S\ref{subsec:assumtion_and_base_quantities} \\
$s$ & Scale parameter ($s > 0$) & \S\ref{subsec:assumtion_and_base_quantities} \\
$\mathcal{D}$ & Base shape distribution; $X \sim \mathcal{D}$ & \S\ref{subsec:assumtion_and_base_quantities} \\
$F_\mathcal{D}, f_\mathcal{D}$ & CDF and PDF of the base shape $\mathcal{D}$ & \S\ref{subsec:assumtion_and_base_quantities} \\
$\D_0$ & Base difference, $\D_0 = X' - X$ for $X, X' \iid \mathcal{D}$ & \S\ref{subsec:assumtion_and_base_quantities} \\
$G_\mathcal{D}(x)$ & CDF of $\D_0$: \ $G_\mathcal{D}(x) = \int f_\mathcal{D}(u)\, F_\mathcal{D}(u+x)\,du$ & Eq.~\eqref{eq:G-def} \\
$c(\mathcal{D})$ & Base threshold, $c(\mathcal{D}) = \E[(X'-X)^+]$ & \S\ref{subsec:assumtion_and_base_quantities} \\
\midrule
\multicolumn{3}{@{}l}{\textit{Value Functions and Thresholds}} \\[2pt]
$V(\tau, k)$ & Optimal expected total welfare from state $(\tau, k)$ & \S\ref{sec:general_dp} \\
$W(\tau, k)$ & Discretionary value: $W(\tau,k) = V(\tau,k) - \tau \cdot \mupol$ & \S\ref{subsec:discretionary_valu_and_general_invariance} \\
$w(\tau,k;\mathcal{D})$ & Base discretionary value (scale-free): $W(\tau,k) = s \cdot w(\tau,k;\mathcal{D})$ & Thm.~\ref{thm:general-invariance} \\
$\Th_{\tau,k}$ & Optimal override threshold at state $(\tau,k)$ & \S\ref{sec:general_dp} \\
$\theta_{\tau,k}(\mathcal{D})$ & Base threshold (scale-free): $\Th_{\tau,k} = s \cdot \theta_{\tau,k}(\mathcal{D})$ & Thm.~\ref{thm:general-invariance} \\
\midrule
\multicolumn{3}{@{}l}{\textit{Behavioral Quantities}} \\[2pt]
$\psi$ & Patience scalar (two-period): $\psi = F_\D(\Th_{2,1}) = G_\mathcal{D}(c(\mathcal{D}))$ & Thm.~\ref{thm:shape_dependence} \\
$q_{\tau,k}$ & Conditional spending prob.\ at state $(\tau,k)$: $q_{\tau,k} = 1 - F_\D(\Th_{\tau,k})$ & \S\ref{subsec:behavioral_implications} \\
$\pi_t(k)$ & Prob.\ of having budget $k$ at the start of period $t$ & \S\ref{subsec:behavioral_implications} \\
$\pi^*$ & Optimal override policy & \S\ref{sec:general_dp} \\
$\PrDP$ & Probability under the optimal policy $\pi^*$ & App.~\ref{app:two_period_case} \\
$\PrOracle$ & Probability under the oracle (full-information) policy $\tilde{\pi}$ & App.~\ref{app:two_period_case} \\
\bottomrule
\end{tabular}
\end{table}

\section{Relationship to prior work}
\label{sec:background_full}

Our model of budgeted discretion is connected to a broad literature on sequential decision-making under scarcity. Across these settings, a decision-maker observes noisy, per-period opportunities and must decide when to convert them into a small number of irreversible actions. A recurring structural feature is a state-dependent threshold: a scarce action is taken only when the realized opportunity exceeds a continuation-value benchmark. While threshold policies are a familiar structural feature of these problems, our contribution is a behavioral invariance result that characterizes how often an optimal policy acts as the scale of payoffs changes. The related work below traces the lineage of this viewpoint, moving from budgeted online allocation benchmarks to the threshold structure of stopping and search, then to comparative statics driven by tail behavior, and finally to invariance arguments that produce scale-free behavioral predictions.

\paragraph{Budgeted online choice and clairvoyant benchmarks.}

Discretion can be viewed as an online allocation problem with a limited inventory: each period reveals a stochastic \emph{opportunity} (the incremental gain from overriding a default), and the agent decides whether to spend one unit of a finite budget or save it. This abstraction closely parallels dynamic and stochastic knapsack models, where optimal policies compare realized rewards to the shadow value of remaining capacity \cite{kleywegt1998dynamic}. Related work in stochastic packing emphasizes the effectiveness of simple adaptive policies under uncertainty \cite{dean2008benefit,arlotto2020logarithmic}. A complementary benchmark tradition compares online decisions to an offline ``prophet'' who observes all realizations in advance \cite{hill1992survey,kleinberg2019matroid,jiang2025tight}. In contrast to this value-focused literature, we use the finite-horizon optimum to characterize \emph{the rate} at which discretionary capacity is consumed. This behavioral lens naturally leads to the optimal-stopping and search tradition, where thresholds are interpreted as reservation values balancing current gains against the option value of waiting.

\paragraph{Threshold structure with limited selections.}
The tension between acting now and waiting for future opportunities is central to optimal stopping and search. In canonical job search models, a decision-maker accepts an offer only if it exceeds a reservation value that equates the marginal benefit of stopping to the option value of continued search \cite{mccall1970economics,weitzman1979optimal}. Closely related formulations appear in best-choice and secretary problems, where a sequence of realizations is observed and decisions are irrevocable; in full-information variants, optimal behavior is governed by time-varying thresholds that become less selective as the horizon shortens \cite{gilbert1966recognizing}. Multi-choice extensions allow multiple acceptances and yield thresholds indexed by remaining selections, mirroring the $(\tau,k)$ state of our model \cite{tamaki1979double,preater1994multiplechoice}. While these problems are often framed in terms of selecting the best realization or maximizing expected selected value, we use the same threshold structure to study how frequently an optimal agent exercises a scarce action when decisions are coupled through a finite budget. This framing also highlights that the shape of the opportunity distribution matters: the value of waiting depends on how much upside remains in the tail.

\paragraph{Tail geometry and selective action.}
A central comparative-statics insight from search and stopping is that greater dispersion can increase selectivity by amplifying the upside of waiting while leaving downside avoidable \cite{rothschild1970increasing}. In heavy-tailed environments, threshold-based stopping rules can effectively target rare extreme realizations, a phenomenon highlighted in asymptotic analyses of threshold-stopped processes \cite{kennedy1991limit}. Our work brings this intuition into a finite-budget setting: patience depends jointly on time and remaining discretionary capacity, and tail heaviness induces systematically more outlier-targeting use of discretion. 

\paragraph{Invariance and scale-free behavior.}
Robust sequential decision rules often arise from exploiting invariances of the underlying observation model. In best-choice problems with partial information, rank-based formulations yield stopping behavior that is insensitive to distributional parameters \cite{petruccelli1980best}. Related approaches, such as the odds algorithm, provide compact index rules with transparent comparative statics \cite{bruss2000odds}. Whereas these classic formulations often emphasize probability-of-success objectives (e.g., selecting the overall best realization or the last success), our objective is additive: after accounting for the baseline policy, the agent maximizes the expected sum of incremental gains $\Delta_t$ obtained from at most $K$ overrides. Building on this perspective, our main theoretical contribution is a behavioral invariance: for location--scale families of improvement distributions, the \emph{probability} that an optimal agent spends an override at a given state is invariant to affine rescaling of payoffs. While numerical thresholds necessarily depend on units, the induced \emph{action rate} is a unit-free object determined solely by distributional shape, yielding predictions about discretionary behavior across environments. These scale-free predictions are particularly useful in empirical settings, where payoffs are measured imperfectly and discretion is shaped by institutional constraints.

\paragraph{Discretion under organizational scarcity.}
Empirically, discretionary overrides are pervasive in high-stakes services. Behavioral studies document under-use of algorithmic recommendations following observed errors \cite{dietvorst2015algorithm} and strategic overrides to repair perceived biases \cite{saxena2024algorithmic,cheng2022childwelfare, Guo_Wu_Hartline_Hullman_2024, Guo_Wu_Hartline_Hullman_2025}. Economic field evidence similarly treats discretion as a channel for private information and local judgment \cite{angelova2023algorithmic,kesavan2020field}, but often models overriding as a static choice. Complementary work in public administration has proposed formal models of street-level discretion in which frontline agents choose a level of engagement by trading off public-service benefits against personal costs such as stress, time, and information acquisition, under constraints on resources and authority \cite{chang2021formal}. These models capture why discretion is costly and institutionally bounded, but do not address the intertemporal rationing problem that arises when overrides are treated as a finite resource. Classic accounts of street-level bureaucracy instead emphasize that discretion is fundamentally shaped by caseload pressure and downstream scarcity \cite{lipsky1980street,butler2021operational}. Our model formalizes this scarcity mechanism and aligns with service-operations work in which limited capacity induces protection levels and threshold-like upgrading rules \cite{Shumsky_Zhang_2009}. We use homeless-service allocation as an empirical setting to connect these ideas, showing that observed override behavior tracks capacity in a manner consistent with optimal management of discretion under scarcity \cite{kube2023fair}, and aligning with recent arguments that discretion and institutional context are central to the fairness and ethics of AI-driven societal decisions beyond purely automation-centric framings \cite{Pokharel_2025, Pokharel_Farabi_Fowler_Das_2025}.

Taken together, this literature scaffolds our analysis. We bring these threads together in a setting motivated by discretionary overrides, where the primitive is not choosing the best item but choosing when to deviate from a standing default under a hard budget. Accordingly, our threshold is not calibrated to maximize a success probability, but emerges as the shadow price of discretionary capacity in a welfare-maximizing DP. By focusing on the induced spending probabilities -- a behaviorally observable object -- we obtain a scale-free characterization of discretion that links the geometry of the gain distribution to systematic differences in how aggressively or conservatively agents consume an override budget.

\section{Model construction and derivations from Section~\ref{sec:the_model}}
\label{app:model_construction}

This appendix provides the complete construction of the gain from discretion $\D_t$ and the associated distributional definitions summarized in Section~\ref{sec:the_model}.

\subsection{Allocation setup}
In each period $t \in \{1, \dots, T\}$, two households arrive. Each household $i \in \{1,2\}$ has a \textbf{baseline welfare} $w_i(t)$, drawn i.i.d.\ from a distribution $\mathcal{D}_w$, and a potential \textbf{welfare improvement} $I_i(t)$, drawn i.i.d.\ from a distribution $I_\theta$, representing the gain in welfare if that household receives the resource. Baselines and improvements are mutually independent. The final welfare of a household that receives the resource is $w_i'(t) = w_i(t) + I_i(t)$. For any period, let $m$ and $M$ denote the indices of the households with minimum and maximum baseline welfare:
    \begin{equation*}
        m = \argmin_{i \in \{1,2\}} \{w_i(t)\}, \qquad M = \argmax_{i \in \{1,2\}} \{w_i(t)\}.
    \end{equation*}

\subsection{Default policy and welfare-maximizing allocation}
The \emph{default policy} is a strictly needs-based rule: it always assigns the resource to the most vulnerable household ($m$) i.e., the household with the lower beginning welfare. This reflects standard prioritization logic in services such as homelessness or organ allocation \cite{hrsa2025_liverallocation, hud2017_cpd1701}. The total welfare under this policy is:
    \begin{equation*}
        \label{eq:policy_utility_app}
        \Upol = (w_m + I_m) + w_M = w_m + w_M + I_m.
    \end{equation*}
In contrast, a welfare-maximizing assignment gives the resource to whichever household maximizes the total final welfare. For example, in homelessness services, while the administrative rule is to prioritize the most vulnerable, evaluation is often in terms of success in keeping households stably housed \cite{kube2023fair}, which corresponds to prioritizing the household with the higher improvement. The maximum possible welfare under this is:
    \begin{equation*}
        \label{eq:discretion_utility_app}
        \Umax = w_m + w_M + \max\{I_m, I_M\}.
    \end{equation*}

\subsection{Derivation of the gain from discretion}
The realized \textbf{gain from discretion} is the welfare difference between the two allocations:
    \begin{equation*}
        \label{eq:delta_full_derivation}
        \D = \Umax - \Upol = \big(\max\{I_m, I_M\} - I_m\big) = (I_M - I_m)^+.
    \end{equation*}
Note that the baseline welfare terms $(w_m, w_M)$ cancel exactly, so the gain depends only on the improvements. The gain is non-negative by construction: $\D > 0$ if and only if $I_M > I_m$, i.e., when the less vulnerable household would benefit more from the resource than the neediest one. We call this event \textit{misalignment}, denoted $\Mlign$, with complement $\Align$ (alignment). Under i.i.d.\ continuous improvements, $\Pr(\Mlign) = \Pr(\Align) = 1/2$ by symmetry. We write $p = \Pr(\Align)$.

\subsection{Distributional definitions}
\label{app:distributional_defs}
The gain from discretion follows a mixture distribution across periods. In any period $t$:
\begin{equation*}
    \Delta_t =
    \begin{cases}
        0, & \text{with probability } p \quad (\text{alignment}: I_m \geq I_M),  \\
        \iid f_{\D} \text{ on } (0,\infty), & \text{with probability } 1-p \quad (\text{misalignment}: I_M > I_m),
    \end{cases}
\end{equation*}
where $f_\Delta$ is the density of the gain conditional on misalignment, with conditional CDF $F_{\Delta\mid A_t^c}$. We write $F_\Delta$ for the unconditional CDF, so that 
\[F_\Delta(x)=p+(1-p)F_{\Delta\mid A_t^c}(x),\qquad x\ge0.\]

Let $\Upol_t$ and $\Umax_t$ respectively denote the period-$t$ welfare under the policy-compliant and the welfare-maximizing allocation. By construction:
    \begin{equation*}
        \Umax_t = \Upol_t + \D_t.
    \end{equation*}
Taking expectations and defining the per-period expected utilities:
    \begin{equation*}
        \mupol = \E[\Upol_t], \qquad \mumax = \E[\Umax_t] = \mupol + \E[\D_t].
    \end{equation*}
Because $\D_t = 0$ with probability $p$ and is strictly positive with probability $1-p$, the unconditional expected gain decomposes as (mentioned in Section~\ref{sec:the_model}):
    \begin{equation*}
        \E[\D_t] = (1-p)\E[\D_t \mid \Mlign].
    \end{equation*}

\subsection{Simplification of the recursion from Section ~\ref{sec:general_dp}}
\label{app:simplification_value_fxn}
 
We start from the conditional recursion, splitting on whether we have alignment ($\Align$ with probability $p$) or not ($\Mlign$ with probability $(1-p)$). 

\begin{align*}
    V(\tau, k) &= \underbrace{p \cdot \expect{\paren{\Umax + V(\tau -1, k)} \mid \Align }}_{\text{alignment } \paren{\D = 0}} \\
    &\quad + \underbrace{\paren{1-p} \cdot \expect{ \max \braces{\paren{\Upol + V(\tau -1, k)}, \paren{\Umax + V(\tau -1, k -1)}}  \mid \Mlign}}_{\text{misalignment } \paren{\D > 0}}
\end{align*}

\noindent Let us analyze the two portions in more detail: 
\begin{itemize}
    \item \textbf{On Alignment, } $\Umax = \Upol$, so we can re-write the portion conditioned on $\Align$ as: 
    \begin{equation*}
       \expect{\paren{\Umax + V(\tau -1, k)} \mid \Align} = \expect{\Upol \mid \Align} + V(\tau-1,k)
    \end{equation*}

    \item \textbf{On Misalignment,} on the other hand, we can write $\Umax = \Upol + \D$ with $\D > 0$. Then, looking at the $\max$ expression within the second half, we can re-write it as: 
    \begin{align*}
        &\max \braces{\Upol + V(\tau-1,k), \Umax +V(\tau-1,k-1)} \\
        &\qquad = \Upol + \max \braces{V(\tau-1,k), \D + V(\tau-1,k-1)}
    \end{align*}
    Taking the conditional on $\mathsf A_t^c$ then gives us the second half as: 
     \begin{align*}
         &= \expect{\Upol + \max \braces{V(\tau-1,k),  \D + V(\tau-1,k-1)} \mid \Mlign} \\ 
         &= \expect{\Upol \mid \Mlign} + \mathbb{E}_{\D \mid \Mlign } \brackets{\max \braces{V(\tau-1,k), \Delta + V(\tau-1,k-1)}}
     \end{align*}
\end{itemize}

Note that, by law of total expectation: 
    \begin{equation}
        p \cdot \expect{\Upol \mid \Align} + (1-p) \cdot \expect{\Upol \mid \Mlign} = \expect{\Upol}.
        \label{equ:simplification}
    \end{equation}

Going back to the original recursion: 
{\small \begin{align*}
    V(\tau,k) & = p \cdot \expect{\Umax + V(\tau-1,k) \mid \Align} + \\
               & (1-p) \cdot \expect{ \max \braces{\paren{\Upol + V(\tau-1,k)}, \paren{\Umax+V(\tau-1,k-1)}} \mid \Mlign }\\[6pt]
              &= p \cdot \expect{\Upol \mid \Align} + p \cdot V(\tau-1,k) + \\ 
              &(1-p) \cdot \expect{\Upol \mid \Mlign} + (1-p) \cdot \mathbb{E}_{\D \mid \Mlign} \left[ \max \braces{V(\tau-1,k), \Delta + V(\tau-1,k-1)} \right]
\end{align*}}

Now using the result from Equation \eqref{equ:simplification}, we get our final recursion as: 
\begin{equation*}
    V(\tau,k) = \mupol + p  V(\tau-1,k) + (1-p)  \E_{\Delta|\Mlign} \Big[ \max \big\{V(\tau-1,k),   \Delta + V(\tau-1,k-1) \big\} \Big]. 
    \tag{\ref{eq:recursion_main}}
\end{equation*}

\section{Proofs}
\label{app:proofs}

\subsection{Proof of Proposition~\ref{prop:c-closed-form}}
\label{app:proof-c-closed-form}
\PropositionOne*
\begin{proof}
    We can derive a compact expression for the base threshold $c(\mathcal{D})$. Let $X, X' \iid \mathcal{D}$. By the law of total expectation, $c(\mathcal{D}) = \expect{(X'-X)^+} = \expect{\expect{(X'-x)^+ \mid X=x}}$. For a fixed $x$, the inner conditional expectation is:
        \begin{equation*}
            \expect{(X'-x)^+ \mid X=x} = \int_{x}^{\infty} (y-x) f_{\mathcal{D}}(y)dy.
        \end{equation*}
    We use integration by parts. Let $g(y) = 1 - F_{\mathcal{D}}(y)$, so $g'(y) = -f_{\mathcal{D}}(y)$. Then:
        \begin{align*}
            \int_{x}^{\infty} (y-x) f_{\mathcal{D}}(y)dy 
            &= -\int_{x}^{\infty} (y-x) g'(y)dy \\
            &= -\Big[(y-x)g(y)\Big]_{x}^{\infty} + \int_{x}^{\infty} g(y)dy \\
            &= 0 + \int_{x}^{\infty} (1 - F_{\mathcal{D}}(y))dy. \tag{since $\lim_{y\to\infty} y(1-F_{\mathcal{D}}(y))=0$}
        \end{align*}
    This equality holds assuming the mean is finite. Now, we take the outer expectation over $X=x$:
        \begin{equation*}
            c(\mathcal{D}) = \expect{\int_{x}^{\infty} (1 - F_{\mathcal{D}}(y))dy} = \int_{-\infty}^{\infty} f_{\mathcal{D}}(x) \left( \int_{x}^{\infty} (1 - F_{\mathcal{D}}(y))dy \right) dx.
        \end{equation*}
    The domain of integration is $\braces{(x,y): -\infty < x \le y < \infty}$. We swap the order of integration:
        \begin{align*}
            c(\mathcal{D}) &= \int_{-\infty}^{\infty} (1 - F_{\mathcal{D}}(y)) \left( \int_{-\infty}^{y} f_{\mathcal{D}}(x)dx \right) dy \\
            &= \int_{-\infty}^{\infty} (1 - F_{\mathcal{D}}(y)) \cdot F_{\mathcal{D}}(y) dy.
        \end{align*}
    This gives the compact expression for $c$ in terms of the CDF alone:
        \begin{equation*}
            c(\mathcal{D}) = \int_{-\infty}^\infty F_{\mathcal{D}}(x)(1 - F_{\mathcal{D}}(x)) dx.
        \end{equation*}
    Because $X$ has non-negative support, the integral runs from $0$ to $\infty$.
        \begin{equation*}
            c(\mathcal{D}) = \int_{0}^\infty F_{\mathcal{D}}(x)(1 - F_{\mathcal{D}}(x)) dx.
        \end{equation*}
\end{proof}

\subsection{Proof of Theorem~\ref{thm:general-invariance} (general behavioral invariance)}
\label{app:general-invariance}

This section provides the complete proof of Theorem~\ref{thm:general-invariance}. We first derive the $W$-recursion, then prove each part.
\GeneralInv*
The proof proceeds by induction on $\tau$ through the recursion as detailed in Equation~\eqref{eq:recursion_main}. For this argument to work, we need two ingredients:
\begin{enumerate}
    \item A recursion for the ``discretionary value'' -- the part of the value function that comes from having override budget, stripped of the baseline utility that accrues regardless. This recursion is where the induction operates.
    \item The factorization $\D = s \cdot Z$, which shows that the gain from discretion separates into a scale factor $s$ and a shape-only random variable $Z$. This is the engine of invariance: it is what allows $s$ to be pulled out of the recursion at every step.
\end{enumerate}

\subsubsection{Derivation of the \texorpdfstring{$W$}{W}-recursion from \texorpdfstring{Equation~\eqref{eq:W-recursion}}{Equation (7)} in Section~\ref{subsec:discretionary_valu_and_general_invariance}}

\label{app:W_recursion_and_simplification}

The recursion in Equation~\eqref{eq:recursion_main} mixes two sources of value: the baseline utility $\mupol$ that is accrued every period regardless of whether discretion is exercised, and the incremental value that arises from having override budget available. To isolate the latter, we subtract out the baseline accumulation.

\begin{definition}[Discretionary value]\label{def:disc-value}
The \emph{discretionary value} at state $(\tau, k)$ is $W(\tau, k) = V(\tau, k) - \tau \cdot \mupol$. This measures the expected welfare gain from holding $k$ override units with $\tau$ periods remaining, relative to the baseline of following the default policy in every remaining period.
\end{definition}

With this definition, we state and prove the following lemma that gives us Equation~\eqref{eq:W-recursion}. 

\begin{lemma}[Recursion for $W$]\label{lem:W-recursion}
    The discretionary value satisfies
        \begin{equation}
        \tag{\ref{eq:W-recursion}}
            W(\tau, k) = p \cdot W(\tau - 1, k) + (1 - p) \cdot \E\brackets{\max\braces{W(\tau - 1, k),  \D + W(\tau - 1, k - 1)}  \Big|  \Mlign}
        \end{equation}
    with boundary conditions $W(0, k) = 0$ for all $k$ and $W(\tau, 0) = 0$ for all $\tau \geq 0$.
\end{lemma}

\begin{proof}
    We substitute $V(\tau, k) = W(\tau, k) + \tau \cdot \mupol$ into every occurrence of $V$ in the recursion from Equation~\eqref{eq:recursion_main}. The left-hand side becomes $W(\tau, k) + \tau \cdot \mupol$. On the right-hand side, every $V(\tau - 1, \cdot)$ term gains an additive offset of $(\tau - 1) \cdot \mupol$. Inside the $\max$ operator, this offset appears in both arguments and therefore factors out. Collecting the constant terms on the right gives $\mupol + p \cdot (\tau - 1)\mupol + (1 - p) \cdot (\tau - 1)\mupol = \tau \cdot \mupol$. Canceling $\tau \cdot \mupol$ from both sides yields~\eqref{eq:W-recursion}.
    
    For the boundary conditions: $V(0, k) = 0$ gives $W(0, k) = 0$. When the budget is exhausted, $V(\tau, 0) = \tau \cdot \mupol$, giving $W(\tau, 0) = 0$.
\end{proof}
The recursion~\eqref{eq:W-recursion} further admits a useful simplification. Write $d_k = W(\tau - 1, k) - W(\tau - 1, k - 1)$ for the marginal value of the $k$-th budget unit at time $\tau - 1$. Since $W(\tau - 1, k) = W(\tau - 1, k - 1) + d_k$, the $\max$ inside~\eqref{eq:W-recursion} becomes $W(\tau - 1, k - 1) + \max\braces{d_k, \D}$. Substituting and simplifying:
\begin{equation}\label{eq:W-compact}
W(\tau, k) = W(\tau - 1, k) + (1 - p) \cdot \E\brackets{(\D - d_k)^+  \Big|  \Mlign}.
\end{equation}
The interpretation is direct: the discretionary value at $(\tau, k)$ equals the value at $(\tau - 1, k)$ -- carrying the same budget into one fewer period -- plus the expected ``bonus'' from the current period, which is collected only when a misalignment occurs and the gain exceeds the opportunity cost of spending.

\subsubsection{The gain factorization from Section \ref{subsec:discretionary_valu_and_general_invariance}}
\label{app:gain-factor}
The second ingredient is a factorization of the gain $\D$ under the location-scale assumption. Recall that the gain from discretion is $\D = (I_M - I_m)^+$, where $I_M = a + sX_M$ and $I_m = a + sX_m$ are the improvement draws for the two households, with $X_M, X_m \iid \mathcal{D}$. Therefore:
    \begin{equation*}
        \label{eq:delta-ls}
            \D = \paren{(a + sX_M) - (a + sX_m)}^+ = \paren{s(X_M - X_m)}^+ = s \cdot (X_M - X_m)^+.
    \end{equation*}
Two things happen. First, the location parameter $a$ cancels: it appears with the same sign in both improvements, so it drops out of the difference. Second, the scale parameter $s$ factors out of the positive-part operator because $s > 0$ preserves sign. Writing $Z = (X_M - X_m)^+$, whose distribution depends only on the base shape $\mathcal{D}$, we have
    \begin{equation*}
        \label{eq:delta-scaled}
        \D = s \cdot Z
    \end{equation*}
This is the fundamental factorization: the gain in any period is the scale $s$ times a shape-only random variable $Z$.

\subsubsection{Putting the pieces together for the proof}

Now we are ready to prove the Theorem~\ref{thm:general-invariance}. The argument for part~(i) is the core of the proof and proceeds by induction on the number of periods remaining $\tau$. It combines the ingredients developed above: the $W$-recursion (Lemma~\ref{lem:W-recursion}), the gain factorization, and the ability to factor $s$ out of the resulting expression. Parts~(ii) and~(iii) then follow as short algebraic consequences.

\noindent\textbf{Proof of Part (i): $W(\tau, k) = s \cdot w(\tau, k; \mathcal{D})$}. We prove by induction on $\tau$ that, for all $k \in \{0, 1, \ldots, K\}$ (for $k = 0$ this is immediate from $W(\tau, 0) = 0$, so the inductive step below takes $k \geq 1$), the discretionary value $W(\tau, k)$ can be written as $s$ times a quantity $w(\tau, k; \mathcal{D})$ that does not depend on $a$ or $s$.

\medskip
\noindent\emph{Step 1: Base case ($\tau = 0$).} By the boundary condition in Lemma~\ref{lem:W-recursion}, $W(0, k) = 0$ for all $k$. We can write this as $W(0, k) = s \cdot 0$, so taking $w(0, k; \mathcal{D}) = 0$ satisfies the claim trivially: both sides are zero regardless of $s$ and $a$.

\medskip
\noindent\emph{Step 2: Inductive hypothesis.} Assume the claim holds at time $\tau - 1$. That is, assume there exist quantities $w(\tau - 1, j; \mathcal{D})$, depending only on $\mathcal{D}$, $p$, and $j$, such that
    \begin{equation}
    \label{eq:inductive-hyp}
        W(\tau - 1, j) = s \cdot w(\tau - 1, j; \mathcal{D}) \qquad \text{for all } j \in \{0, 1, \ldots, K\}.
    \end{equation}

\medskip
\noindent\emph{Step 3: Substituting into the $W$-recursion.} Our goal is to show that $W(\tau, k)$ also has this form. We start from the recursion established in Lemma~\ref{lem:W-recursion}:
\begin{equation*}
    W(\tau, k) = p \cdot W(\tau - 1, k) + (1 - p) \cdot \E_{\D \mid \Mlign}\brackets{\max\braces{W(\tau - 1, k),  \D + W(\tau - 1, k - 1)}}.
\end{equation*}

We now replace every quantity on the right-hand side using the inductive hypothesis~\eqref{eq:inductive-hyp} for the $W$ terms, and the gain factorization $\D = s \cdot Z$ for the gain. Making both substitutions simultaneously, we get:
    \begin{equation}
    \label{eq:sub1}
    \begin{split}
        W(\tau, k) &= p \cdot s \cdot w(\tau - 1, k; \mathcal{D}) \\
        &\quad + (1 - p) \cdot \E_{\D \mid \Mlign}\brackets{\max\braces{s \cdot w(\tau - 1, k; \mathcal{D}),  s \cdot Z + s \cdot w(\tau - 1, k - 1; \mathcal{D})}}.
    \end{split}
    \end{equation}
At this stage, $s$ appears in every term inside the $\max$ as well as in the first summand. The rest of the argument consists of extracting it.

\medskip
\noindent\emph{Step 4: Factoring $s$ out of the $\max$.} Look at the two arguments inside the $\max$ in~\eqref{eq:sub1}. The first argument is $s \cdot w(\tau - 1, k; \mathcal{D})$. The second argument is $s \cdot Z + s \cdot w(\tau - 1, k - 1; \mathcal{D}) = s \cdot \brackets{Z + w(\tau - 1, k - 1; \mathcal{D})}$. Both arguments have $s$ as a common factor. Because $s > 0$, multiplying by $s$ preserves the ordering of real numbers, so the identity $\max\{s \cdot u,  s \cdot v\} = s \cdot \max\{u, v\}$ holds. Applying this:
    \begin{align}
        &\max\braces{s \cdot w(\tau - 1, k; \mathcal{D}),  s \cdot Z + s \cdot w(\tau - 1, k - 1; \mathcal{D})} \notag \\
        &\qquad = s \cdot \max\braces{w(\tau - 1, k; \mathcal{D}),  Z + w(\tau - 1, k - 1; \mathcal{D})}. \label{eq:factor-s}
    \end{align}

\medskip
\noindent\emph{Step 5: Pulling $s$ outside the expectation.} Substituting~\eqref{eq:factor-s} back into~\eqref{eq:sub1}, the expectation becomes
    \begin{equation*}
        \E_{\D \mid \Mlign}\brackets{s \cdot \max\braces{w(\tau - 1, k; \mathcal{D}),  Z + w(\tau - 1, k - 1; \mathcal{D})}}.
    \end{equation*}
Since $s$ is a positive constant (it does not depend on the random variable $Z$), the linearity of expectation allows us to pull it outside:
    \begin{equation*}
        = s \cdot \E_{\D \mid \Mlign}\brackets{\max\braces{w(\tau - 1, k; \mathcal{D}),  Z + w(\tau - 1, k - 1; \mathcal{D})}}.
    \end{equation*}
Substituting this back, the full expression~\eqref{eq:sub1} becomes:
    \begin{equation}
    \label{eq:s-outside}
        \begin{split}
        W(\tau, k) &= p \cdot s \cdot w(\tau - 1, k; \mathcal{D}) \\
        &\quad + (1 - p) \cdot s \cdot \E_{\D \mid \Mlign}\brackets{\max\braces{w(\tau - 1, k; \mathcal{D}),  Z + w(\tau - 1, k - 1; \mathcal{D})}}.
    \end{split}
    \end{equation}

\medskip
\noindent\emph{Step 6: Factoring $s$ from the entire expression.} Both terms on the right-hand side of~\eqref{eq:s-outside} have $s$ as a common factor. Extracting it:

    \begin{equation}
    \label{eq:W-factored}
    \begin{split}
        W(\tau, k) &= s \cdot w(\tau, k; \mathcal{D}), \qquad \text{where} \\
        w(\tau, k; \mathcal{D}) &= p \cdot w(\tau - 1, k; \mathcal{D}) \\
        &\quad + (1 - p) \cdot \E_{\D \mid \Mlign}\brackets{\max\braces{w(\tau - 1, k; \mathcal{D}),  Z + w(\tau - 1, k - 1; \mathcal{D})}}.
    \end{split}
    \end{equation}
The last two lines define $w(\tau, k; \mathcal{D})$. 

\medskip
\noindent\emph{Step 7: Verifying that $w(\tau, k; \mathcal{D})$ is free of $a$ and $s$.} For the induction to close, we must confirm that $w(\tau, k; \mathcal{D})$ as defined in~\eqref{eq:W-factored} depends only on the base shape $\mathcal{D}$, the alignment probability $p$, and the state $(\tau, k)$. We check each component:
\begin{itemize}
    \item The quantities $w(\tau - 1, k; \mathcal{D})$ and $w(\tau - 1, k - 1; \mathcal{D})$ are free of $a$ and $s$ by the inductive hypothesis.
    \item The random variable $Z = (X_M - X_m)^+$ is determined entirely by two independent draws from the base shape $\mathcal{D}$. It does not involve $a$ (which canceled in the difference) or $s$ (which was factored out).
    \item The alignment probability $p$ is a model primitive.
    \item No other quantities appear.
\end{itemize}
Therefore $w(\tau, k; \mathcal{D})$ is indeed free of $a$ and $s$. This completes the inductive step, and by induction, part~(i) holds for all $\tau \in \{0, 1, \ldots, T\}$ and all $k \in \{0, 1, \ldots, K\}$. \hfill $\square$

\bigskip
\noindent\textbf{Proof of Part (ii): $\Th_{\tau,k} = s \cdot \theta_{\tau,k}(\mathcal{D})$}. This part follows directly from part~(i) by taking differences. The optimal threshold is defined as the marginal value of carrying one additional unit of budget into the future:
    \begin{equation*}
        \Th_{\tau,k} = V(\tau - 1, k) - V(\tau - 1, k - 1).
    \end{equation*}
Since $V(\tau - 1, j) = W(\tau - 1, j) + (\tau - 1)\mupol$ (by Definition~\ref{def:disc-value}), the baseline term $(\tau - 1)\mupol$ appears in both $V(\tau - 1, k)$ and $V(\tau - 1, k - 1)$ and cancels in the difference:
\begin{equation*}
    \Th_{\tau,k} = W(\tau - 1, k) - W(\tau - 1, k - 1).
\end{equation*}
Now apply part~(i) to both $W$ terms:
    \begin{align}
        \Th_{\tau,k} &= s \cdot w(\tau - 1, k; \mathcal{D}) - s \cdot w(\tau - 1, k - 1; \mathcal{D}) \notag \\
                    &= s \cdot \brackets{w(\tau - 1, k; \mathcal{D}) - w(\tau - 1, k - 1; \mathcal{D})}. \label{eq:threshold-scaled}
    \end{align}
Define the \emph{base threshold} $\theta_{\tau,k}(\mathcal{D}) = w(\tau - 1, k; \mathcal{D}) - w(\tau - 1, k - 1; \mathcal{D})$. This is a difference of two shape-only quantities, so it inherits their invariance. We conclude $\Th_{\tau,k} = s \cdot \theta_{\tau,k}(\mathcal{D})$. \hfill $\square$

\bigskip
\textbf{Proof of Part (iii): $q_{\tau,k} = 1 - G_\mathcal{D}\paren{\theta_{\tau,k}(\mathcal{D})}$}. This part shows that the scale $s$ cancels from the spending probability. The argument has three sub-steps: we express the CDF of $\D$ in terms of the base-shape CDF, substitute the scaled threshold from part~(ii), and observe that $s$ cancels.

\medskip
\noindent\emph{Step 1: Relating $F_\D$ to $G_\mathcal{D}$.} Recall from Section~\ref{app:gain-factor} that the unconditional gain is $\D = s \cdot Z'$, where $Z' = \D_0^+ = (X_M - X_m)^+$ is the positive part of the base difference. (Note that $Z'$ takes value $0$ under alignment and is strictly positive under misalignment, mirroring the mixture structure of $\D$.) For any $x \geq 0$:
\begin{equation}\label{eq:cdf-relation}
F_\D(x) = \Pr(\D \leq x) = \Pr(s \cdot Z' \leq x) = \Pr \paren{Z' \leq \frac{x}{s}} = G_\mathcal{D} \paren{\frac{x}{s}},
\end{equation}
where $G_{\mathcal{D}}$ is the CDF of the signed difference $X_M-X_m$. For nonnegative arguments, this agrees with the CDF of its positive part $Z'$.

\medskip
\noindent\emph{Step 2: Evaluating at the threshold.} We substitute $x = \Th_{\tau,k} = s \cdot \theta_{\tau,k}(\mathcal{D})$ from part~(ii) into~\eqref{eq:cdf-relation}:
\begin{equation}\label{eq:cancel-s}
F_\D(\Th_{\tau,k}) = G_\mathcal{D} \paren{\frac{s \cdot \theta_{\tau,k}(\mathcal{D})}{s}} = G_\mathcal{D} \paren{\theta_{\tau,k}(\mathcal{D})}.
\end{equation}
The scale parameter $s$ appears in the numerator (from the threshold scaling in part~(ii)) and in the denominator (from the CDF rescaling in~\eqref{eq:cdf-relation}), and cancels exactly.

\medskip
\noindent\emph{Step 3: Concluding.} The spending probability is therefore
\begin{equation}\label{eq:q-invariant}
q_{\tau,k} = 1 - F_\D(\Th_{\tau,k}) = 1 - G_\mathcal{D} \paren{\theta_{\tau,k}(\mathcal{D})}.
\end{equation}
The right-hand side involves only the base-shape CDF $G_\mathcal{D}$ evaluated at the base threshold $\theta_{\tau,k}(\mathcal{D})$. Both of these are functions of $\mathcal{D}$, $p$, and $(\tau, k)$ alone. The location parameter $a$ never entered the gain $\D$ (it canceled in the difference in Section~\ref{app:gain-factor}), and the scale parameter $s$ has canceled in~\eqref{eq:cancel-s}. \hfill $\square$

\bigskip
This completes the proof of Theorem~\ref{thm:general-invariance}. \qed

\section{Two-period special case \texorpdfstring{$(T=2, K=1)$}{(T=2, K=1)}}
\label{app:two_period_case}

This appendix works out the budgeted discretion problem in the minimal non-trivial case -- a single override ($K=1$) over two periods ($T=2$) -- as defined in Section~\ref{subsec:policy_implications}, providing closed-form expressions that are recovered as special cases of the general results in Section~\ref{sec:main_theorem}.

\subsection{State space and terminal values}

The state is $(\tau, k)$ with $\tau \in \{1,2\}$ periods remaining and $k \in \{0,1\}$ budget units. The four states and their values are:

\begin{center}
\renewcommand{\arraystretch}{1.3}
    \begin{tabular}{ccl}
    \toprule
    $\tau$ & $k$ & $V(\tau, k)$ \\
    \midrule
    $1$ & $0$ & $\mupol$ \quad (no budget: collect default welfare) \\
    $1$ & $1$ & $\mumax$ \quad (budget available: override any misalignment) \\
    $2$ & $0$ & $2\mupol$ \quad (no budget for either period) \\
    $2$ & $1$ & solved below \\
    \bottomrule
    \end{tabular}
\end{center}

\noindent The non-trivial case is $(\tau, k) = (2, 1)$: the agent has two periods remaining and one override to allocate.

\subsection{Derivation of the optimal threshold}
\label{app:threshold_derivation}

Applying the recursion in Equation~\eqref{eq:recursion_main} at state $(2, 1)$:
    \begin{equation*}
        V(2,1) = \mupol + p V(1,1) + (1-p) \E_{\D \mid \Mlign}\brackets{\max\braces{V(1,1),  \D + V(1,0)}}.
    \end{equation*}
Substituting $V(1,1) = \mumax$ and $V(1,0) = \mupol$:
    \begin{equation*}
        V(2,1) = \mupol + p \mumax + (1-p) \E\brackets{\max\braces{\mumax,  \D + \mupol}  \Big|  \Mlign}.
    \end{equation*}
The agent overrides in the first period if and only if the immediate gain exceeds the value of saving the budget for the second period. By the general threshold rule~\eqref{eq:threshold_eqn}, this requires $\D_1 \geq \Th_{2,1}$, where:
\begin{equation*}
    \Th_{2,1} = V(1,1) - V(1,0) = \mumax - \mupol = \E[\D].
\end{equation*}
Since $\E[\D] = (1-p) \E[\D \mid \Mlign]$, the threshold decomposes as:
    \begin{equation*}
        \label{eq:optimal_threshold_app}
        \Th_{2,1} = (1-p) \cdot \E[\D \mid \Mlign].
    \end{equation*}
The threshold is the product of the \emph{frequency} of override opportunities ($1-p$) and their expected \emph{magnitude} ($\E[\D \mid \Mlign]$). Intuitively, if misalignments are rare ($p$ large), the threshold is low because the option value of holding the budget is small; if misalignments are frequent and large, the threshold is high because saving the override for the second period is more valuable.

In the final period ($\tau = 1$), if budget remains the agent overrides whenever a misalignment occurs (i.e., whenever $\D_2 > 0$), since there is no future period to save for.

\subsection{The patience scalar \texorpdfstring{$\psi$}{Ψ}}
\label{app:psi_derivation}

While the threshold $\Th_{2,1}$ depends on the units of welfare, the agent's observable behavior--how often they override--is unit-free. Let $\psi$ denote the probability that the agent does \emph{not} spend the budget in the first period. The agent saves in two scenarios:
\begin{enumerate}
    \item \textbf{Alignment} (probability $p$): the default policy is already optimal ($\D_1 = 0$), so there is no gain to capture.
    \item \textbf{Insufficient misalignment} (probability $(1-p) \cdot F_{\D|\Mlign}(\Th_{2,1})$): a misalignment occurs but the realized gain falls below the threshold.
\end{enumerate}
Combining these disjoint events:
\begin{equation*}
    \psi = p + (1-p) F_{\D|\Mlign}(\Th_{2,1}).
\end{equation*}
By the law of total probability, this is exactly the unconditional CDF of $\D$ evaluated at the threshold:
    \begin{equation*}
        \label{eq:psi_def_app}
        \psi = F_\D(\Th_{2,1}) = \Pr(\D_1 \leq \Th_{2,1}).
    \end{equation*}
The probability that the optimal agent exercises discretion in the first period is $1 - \psi$.

\subsection{Conditional spending probabilities}

We first establish the conditional spending probabilities used in the behavioral comparisons of Appendix~\ref{app:behavioral_implications}.

\begin{itemize}
    \item \textbf{Oracle Policy ($\tilde{\pi}$):} An oracle sees both $\D_1$ and $\D_2$ in advance. It spends at time $t$ if either (i) only period $t$ misaligns, or (ii) both periods misalign and $\D_t > \D_s$ (where $s=3-t$). Conditional on both periods misaligning, the positive gains are independent and identically distributed with a continuous distribution. Hence $\Pr(\Delta_t>\Delta_s\mid A_t^c,A_s^c)=1/2$. 
    \begin{equation}
        \PrOracle \paren{\text{spend at }t \mid \MlignT{t}} = p \cdot 1 + (1-p) \cdot \frac{1}{2} = \frac{1+p}{2}.
        \label{eq:oracle_prob}
    \end{equation}
    
    \item \textbf{DP Policy ($\pi^*$):} The DP only sees the current state.
    \begin{itemize}
        \item At $t=1$, it spends if $\D_1$ clears the threshold:
        \begin{equation}
            \PrDP \paren{\text{spend at }1 \mid \MlignT{1}} = \Pr(\D_1 \ge \Th_{2,1} \mid \MlignT{1}) = 1 - F_{\D \mid \MlignT{1}}(\Th_{2,1}) = \frac{1-\psi}{1-p}.
            \label{eq:dp_prob1}
        \end{equation}
        \item At $t=2$, it spends upon any misalignment \textit{if the budget is available}. The budget is available if $t=1$ was an alignment (prob. $p$) or if it was a misalignment but the gain was too low to spend (probability $(1-p)F_{\Delta\mid \MlignT{1}}(\Th_{2,1})$)
        \begin{equation}
            \PrDP \paren{\text{spend at }2 \mid \MlignT{2}} = p + (1-p)F_{\D \mid \MlignT{1}}(\Th_{2,1}) = \psi.
            \label{eq:dp_prob2}
        \end{equation}
    \end{itemize}
\end{itemize}

\subsection{Proof of Theorem~\ref{thm:shape_dependence} (two-period invariance)}
\label{app:proof-shape-dependence}

\TheoremOne*

\begin{proof}
We prove the theorem by showing that for any continuous base distribution $\mathcal{D}$ with finite mean, and for any location-scale transform defined by parameters $(a, s)$ with $s>0$, the key probability $\psi$ is invariant. Let $I_1, I_2$ be two independent draws from the location-scale family $I = a + sX$, where $X_1, X_2 \iid \mathcal{D}$ are drawn from the base shape.
We proceed in four steps.

\noindent \emph{Step 1: Differences cancel location and scale linearly.}
The signed difference between two realizations under the location-scale transform is:
    \begin{equation*}
        I_2 - I_1 = (a+sX_2)-(a+sX_1) = s(X_2-X_1) = s\Delta_0.
    \end{equation*}
The location parameter $a$ cancels out entirely, and the scale $s$ factors out. Consequently, the realized gain $\Delta = (I_M - I_m)^+$ scales as $\Delta = s \Delta_0^+$.

\noindent \emph{Step 2: The optimal threshold scales by $s$.}
The optimal threshold $\mathcal{T}_{2,1}$ is defined as the expectation of the realized gain $\Delta$. Since $s>0$:
\begin{equation*}
    \mathcal{T}_{2,1} = \expect{\Delta} = \expect{s \Delta_0^+}
    = s \cdot \expect{\Delta_0^+} = s \cdot c(\mathcal{D}).
\end{equation*}
This establishes Part 1 of the theorem: the threshold is the base threshold $c(\mathcal{D})$ scaled linearly by $s$.

\noindent \emph{Step 3: The CDF of the difference rescales its argument.}
Let $G_{I}$ be the CDF of the difference $I_2 - I_1$. Using the result from Step 1:
\begin{equation*}
    G_{I}(x) = \Pr(I_2 - I_1 \le x) = \Pr(s\Delta_0 \le x)
    = \Pr\Big(\Delta_0 \le \frac{x}{s}\Big)
    = G_{\mathcal{D}}\Big(\frac{x}{s}\Big).
\end{equation*}

\noindent \emph{Step 4: Evaluate at the threshold to obtain invariance.}
We now combine these results to find the patience scalar $\psi$, the probability of not spending in period 1. Since $F_\D(x) = G_I(x)$ for $x \geq 0$ (the pair $(I_M, I_m)$ has the same distribution as $(I_2, I_1)$, and $\{(I_M - I_m)^+ \le x\} = \{I_M - I_m \le x\}$ for $x \ge 0$):
\begin{align*}
    \psi &= F_\D(\mathcal{T}_{2,1}) = G_{I}(\mathcal{T}_{2,1}) & \text{(definition of $\psi$)} \\
    &= G_{I}\big(s \cdot c(\mathcal{D})\big) & \text{(from Step 2)} \\
    &= G_{\mathcal{D}}\Big(\frac{s \cdot c(\mathcal{D})}{s}\Big) & \text{(from Step 3)} \\
    &= G_{\mathcal{D}}\big(c(\mathcal{D})\big). &
\end{align*}
This establishes Part 2 of the theorem: $\psi$ depends strictly on the base shape objects $G_{\mathcal{D}}$ and $c(\mathcal{D})$, and is invariant to the location $a$ and scale $s$. 

\end{proof}

\subsection{Recovery of \texorpdfstring{$(T=2, K=1)$}{(T=2, K=1)} invariance from the general Theorem~\ref{thm:general-invariance}}
\label{app:two_period_invariance}

As mentioned at the end of Section~\ref{subsec:policy_implications}, Theorem~\ref{thm:shape_dependence} can be proved either directly, as in Appendix~\ref{app:proof-shape-dependence}, or as a corollary of the more general Theorem~\ref{thm:general-invariance}. Under the location--scale model $I = a + sX$ with $X \sim \mathcal{D}$, the general theorem gives $W(\tau, k) = s \cdot w(\tau, k; \mathcal{D})$. At $\tau = 1$:
\begin{align*}
    W(1, 0) &= 0, \\
    W(1, 1) &= (1-p) \E[\D \mid \Mlign] = s (1-p) \E[Z \mid Z > 0],
\end{align*}
where $Z = (X' - X)^+$ for $X, X' \iid \mathcal{D}$. Since $Z = 0$ with probability $p$ and $Z > 0$ with probability $1-p$, the unconditional expectation satisfies $\E[Z] = (1-p) \E[Z \mid Z > 0]$. Therefore:
\begin{equation*}
    w(1, 1; \mathcal{D}) = (1-p) \E[Z \mid Z > 0] = \E[Z] = \E[(X'-X)^+] = c(\mathcal{D}).
\end{equation*}
The base threshold is $\theta_{2,1}(\mathcal{D}) = w(1,1;\mathcal{D}) - w(1,0;\mathcal{D}) = c(\mathcal{D}) - 0 = c(\mathcal{D})$. By Theorem~\ref{thm:general-invariance}(iii), the spending probability is:
\begin{equation*}
    q_{2,1} = 1 - G_\mathcal{D}(c(\mathcal{D})) = 1 - \psi,
\end{equation*}
confirming that $\psi = G_\mathcal{D}(c(\mathcal{D}))$ depends only on the base shape $\mathcal{D}$ and is invariant to both $a$ and $s$.
\subsection{Verification and behavioral examples}
\label{app:verification}

This appendix verifies the invariance and shape-dependence results from Section~\ref{sec:main_theorem} through analytical closed forms and numerical experiments.

\subsubsection{Analytical example: Exponential improvements.}
Let the base improvements follow an Exponential distribution, $X \sim \mathrm{Exp}(\lambda)$. This is a pure scale family where the scale parameter is $b = 1/\lambda$. 
We invoke the standard result that the difference of two i.i.d. exponential random variables follows a Laplace distribution with location 0 and scale $b$. Thus, the base difference density is $f_{\Delta_0}(x) = \frac{1}{2b}e^{-|x|/b}$. Using this explicit form, we calculate the base threshold $c$:
\begin{equation*}
    c = \E[\Delta_0^+] = \int_0^\infty x \cdot \frac{1}{2b}e^{-x/b}dx = \frac{b}{2}.
\end{equation*}
We then evaluate the Laplace CDF, $F(x) = 1 - \frac{1}{2}e^{-x/b}$ (for $x \ge 0$), at this threshold:
\begin{equation*}
    \psi_{\text{Exp}} = F_{\text{Laplace}}(b/2) = 1 - \frac{1}{2}e^{-(b/2)/b} = 1 - \frac{1}{2}e^{-1/2} \approx 0.697.
\end{equation*}

\subsubsection{Analytical example: Half-Normal improvements.}

Let $X \sim \mathrm{HalfNormal}(\sigma)$. The difference of two Half-Normals does not follow a named textbook distribution, but its CDF can be derived via convolution. For $X_1, X_2 \iid \mathrm{HalfNormal}(\sigma)$, we compute the base threshold $c(\mathcal{D}) = \expect{(\D_0)^+}$ in closed form. Let us first calculate the pdf of the difference between two half normals. The half-normal PDF and CDF are
    \begin{equation*}
        f_X(x)
        = \begin{cases}
            \frac{\sqrt{2}}{\sigma \sqrt{\pi}}
            \exp   \paren{-\frac{x^2}{2\sigma^2}}, & x \ge 0,\\[6pt]
            0, & x < 0,
          \end{cases}
        \qquad
        F_{\mathcal{D}}(x)
        = \Pr(X \le x)
        = \erf  \paren{\frac{x}{\sqrt{2}\sigma}},
        \ x \ge 0,
    \end{equation*}
where $\erf$ is the standard error function. The density of the difference $\D_0 = X_2 - X_1$ is given by the usual convolution:
    \begin{equation*}
        f_{\D_0}(d)
        = \int_{-\infty}^{\infty} f_X(x) f_X(x + d)\,dx
        = \int_{\max\{0,-d\}}^{\infty} f_X(x) f_X(x + d)\,dx.
    \end{equation*}
For $d \ge 0$ we have
    \begin{align*}
        f_{\D_0}(d)
        &= \int_0^{\infty}
            \frac{\sqrt{2}}{\sigma \sqrt{\pi}}
            \exp \paren{-\frac{x^2}{2\sigma^2}}
            \cdot
            \frac{\sqrt{2}}{\sigma \sqrt{\pi}}
            \exp \paren{-\frac{(x + d)^2}{2\sigma^2}}
            dx \\
        &= \frac{2}{\pi \sigma^2}
           \int_0^{\infty}
               \exp \paren{-\frac{x^2 + (x + d)^2}{2\sigma^2}} dx \\
        &= \frac{2}{\pi \sigma^2}
           \exp \paren{-\frac{d^2}{4\sigma^2}}
           \int_0^{\infty}
               \exp \paren{-\frac{(x + d/2)^2}{\sigma^2}} dx.
    \end{align*}
With the change of variables $u = (x + d/2)/\sigma$, $dx = \sigma\,du$ and lower limit $u = d/(2\sigma)$, this becomes
    \begin{equation*}
        f_{\D_0}(d)
        = \frac{2}{\pi \sigma^2}
          \exp \paren{-\frac{d^2}{4\sigma^2}}
          \sigma \int_{d/(2\sigma)}^{\infty} e^{-u^2}\,du.
    \end{equation*}
Using the complementary error function
    \begin{equation*}
        \erfc(z)
        = \frac{2}{\sqrt{\pi}} \int_z^{\infty} e^{-u^2}\,du,
    \end{equation*}
we obtain, for $d \ge 0$,
\begin{equation}
    f_{\D_0}(d)
    = \frac{1}{\sqrt{\pi} \sigma}
      \exp \paren{-\frac{d^2}{4\sigma^2}}
      \erfc\paren{\frac{d}{2\sigma}}.
    \label{eq:half-normal-delta-density}
\end{equation}
Since $X_1$ and $X_2$ are i.i.d., the distribution of $\D_0$ is symmetric around zero, so $f_{\D_0}(-d) = f_{\D_0}(d)$ and $G_{\mathcal{D}}(0) = 1/2$. So, for $z \ge 0$, the CDF can be written as
    \begin{equation*}
        G_{\mathcal{D}}(z)
        = \Pr(\D_0 \le z)
        = \frac{1}{2} + \int_0^{z} f_{\D_0}(d)\,dd.
    \end{equation*}

\noindent Next consider, 
\begin{equation*}
    H(z) = 1 - \frac{1}{2} \erfc^2 \paren{\frac{z}{2\sigma}},
    \quad z \ge 0.
\end{equation*}

\noindent We can use the identity $\frac{d}{dz} \erfc(z) = \frac{-2}{\sqrt{\pi}} \exp(-z^2)$ and the chain rule to differentiate $H(z)$ with respect to $z$ to get:  
\begin{equation*}
    H'(z)
    = \frac{1}{\sqrt{\pi} \sigma} \exp \paren{-\frac{z^2}{4\sigma^2}} \erfc \paren{\frac{z}{2\sigma}}
    = f_{\D_0}(z), \quad z \ge 0,
\end{equation*}
and $H(0) = 1/2 = G_{\mathcal{D}}(0)$. Thus $G_{\mathcal{D}}(z) = H(z)$ for all $z \ge 0$, yielding
\begin{equation*}
    G_{\mathcal{D}}(z)
    = 1 - \frac{1}{2} \erfc^2 \paren{\frac{z}{2\sigma}}
    \quad z \ge 0.
\end{equation*}
By symmetry, $G_{\mathcal{D}}(-z) = 1 - G_{\mathcal{D}}(z)$ for $z > 0$.

Now, we derive a closed form for $c(\mathcal{D})$. Recall that
\begin{equation*}
    c(\mathcal{D})
    = \expect{(\D_0)^+}
    = \int_0^{\infty} d \, f_{\D_0}(d)\,dd.
\end{equation*}
Substituting \eqref{eq:half-normal-delta-density} and changing variables $u = d/(2\sigma)$, so $d = 2\sigma u$ and $dd = 2\sigma\,du$, gives
\begin{align}
    c(\mathcal{D})
    &= \int_0^{\infty}
        d \cdot \frac{1}{\sqrt{\pi}\sigma}
        \exp \paren{-\frac{d^2}{4\sigma^2}}
        \erfc \paren{\frac{d}{2\sigma}}
      \,dd \nonumber\\
    &= \frac{4\sigma}{\sqrt{\pi}}
       \int_0^{\infty} u e^{-u^2}\erfc(u)\,du.
    \label{eq:half-normal-c-integral}
\end{align}
To evaluate the integral, write
\begin{equation*}
    \erfc(u)
    = \frac{2}{\sqrt{\pi}} \int_u^{\infty} e^{-t^2}\,dt,
\end{equation*}
so that
\begin{align*}
    \int_0^{\infty} u e^{-u^2} \erfc(u)\,du
    &= \frac{2}{\sqrt{\pi}}
       \int_0^{\infty} \int_u^{\infty}
           u e^{-u^2} e^{-t^2}\,dt\,du \\
    &= \frac{2}{\sqrt{\pi}}
       \int_0^{\infty} e^{-t^2}
           \left[\int_0^{t} u e^{-u^2}\,du\right] dt,
\end{align*}
where we have interchanged the order of integration. The inner integral is
\begin{equation*}
    \int_0^{t} u e^{-u^2}\,du
    = \frac{1}{2} \paren{1 - e^{-t^2}},
\end{equation*}
so
\begin{align*}
    \int_0^{\infty} u e^{-u^2}\erfc (u)\,du
    &= \frac{1}{\sqrt{\pi}}
       \int_0^{\infty} \paren{e^{-t^2} - e^{-2t^2}} \,dt\\
    &= \frac{1}{\sqrt{\pi}}
       \left[
         \frac{1}{2}\sqrt{\pi}
         - \frac{1}{2}\sqrt{\frac{\pi}{2}}
       \right]
     = \frac{2 - \sqrt{2}}{4}.
\end{align*}
Substituting back into \eqref{eq:half-normal-c-integral} yields
\begin{equation*}
    c(\mathcal{D})
    = \frac{4\sigma}{\sqrt{\pi}} \cdot \frac{2 - \sqrt{2}}{4}
    = \frac{\sigma}{\sqrt{\pi}} \paren{2 - \sqrt{2}}
\end{equation*}

Substituting this into the derived CDF:
\begin{equation*}
    \psi_{\text{HN}} = 1 - \frac{1}{2}\operatorname{erfc}^2\left(\frac{2 - \sqrt{2}}{2\sqrt{\pi}}\right) \approx 0.668.
\end{equation*}

In both cases, the optimal spending probability is not dependent on the parameters of the distribution, confirming Theorem~\ref{thm:shape_dependence}.

\subsubsection{Numerical verification: scale vs.\ shape}
\label{app:numerical-verification}

This appendix reports the Monte Carlo simulations used to verify Theorem~\ref{thm:shape_dependence}, which says that $\psi$ should remain fixed if we stretch the distribution (scale change) but move if we alter its tail profile (shape change). All estimates use $N = 10^6$ independent draws per configuration.

\paragraph{Scale invariance.}
Table~\ref{tab:scale-invariance} empirically validates the scale-invariance prediction: changing the scale of the improvement distribution rescales the optimal cutoff $\hat c$ but leaves the unconditional patience scalar $\hat\psi$ unchanged. For example, within the Lognormal family with fixed shape $\sigma=1.0$, increasing $e^\mu$ from $1$ to $5$ multiplies the underlying improvements by $5$; correspondingly, the estimated cutoff increases from $\hat c\approx 0.856$ to $\hat c\approx 4.288$ (a $5\times$ change), while $\hat\psi$ remains $0.736$ at both endpoints.

We estimate the threshold $\hat{c}$ and patience scalar $\hat{\psi}$. For each distribution parameterization, we generate $N=10^6$ pairs of draws to estimate the threshold $\hat{c} = \frac{1}{N}\sum (X'_i - X_i)^+$. We then do an additional $N=10^6$ independent simulations using this threshold to empirically measure $\hat{\psi}$.

\begin{table}[ht]
    \centering
    \caption{Verification of scale invariance (fixed shape, varying scale). Empirical estimates $(\hat c,\hat\psi)$ are derived from $10^6$ Monte Carlo simulations.}
    \label{tab:scale-invariance}
    \small
    \begin{tabular}{llccc}
    \toprule
    \textbf{Distribution} & \textbf{Shape (Fixed)} & \textbf{Scale Param} & \textbf{Threshold $\hat{c}$} & \textbf{Patience $\hat{\psi}$}  \\
    \midrule
    Exponential & $k=1.0$ & $b=0.1 \to 10.0$ & $0.05 \to 5.002$ & \textbf{0.697} \\
    Half-Normal & N/A & $\sigma=0.2 \to 20.0$ & $0.066 \to 6.62$ & \textbf{0.668}  \\
    Gamma & $k=2.0$ & $\theta=0.5 \to 8.0$ & $0.374 \to 6.001$ & \textbf{0.675}  \\
    Lognormal & $\sigma=1.0$ & $e^\mu=1 \to 5$ & $0.856 \to 4.288$ & \textbf{0.736}  \\
    \bottomrule
\end{tabular}%
\end{table}

\paragraph{Shape dependence.}
For some families, changing the ``shape'' parameter also changes the overall level of $X$ (e.g., Gamma with fixed $\theta$ changes mean $k\theta$), which would move the cutoff $c(\mathcal D)=\E[(X'-X)^+]$ for unit reasons rather than tail behavior. To ensure Table~\ref{tab:shape-dependence} reflects \emph{shape} rather than scale, we normalize each family to have $\E[X]=1$ for every shape value. Specifically, for Gamma we set $\theta=1/k$; for Weibull we set the scale $\lambda=1/\Gamma(1+1/c)$ (numerically $\lambda\approx 0.883$ for $c=0.8$ and $\lambda\approx 1.108$ for $c=1.5$); and for Lognormal we set $\mu=-\sigma^2/2$ (so $\mu=-1.125$ for $\sigma=1.5$ and $\mu=-0.125$ for $\sigma=0.5$).

\begin{table}[ht]
\centering
    \caption{Verification of shape dependence (varying shape). Empirical estimates $(\hat c,\hat\psi)$ are derived from $10^6$ Monte Carlo simulations.}
    \label{tab:shape-dependence}
    \small
    \begin{tabular}{llccc}
    \toprule
    \textbf{Distribution} & \textbf{Shape Param} & \textbf{Shape Description} & \textbf{Threshold $\hat{c}$} & \textbf{Patience $\hat{\psi}$} \\
    \midrule
    Gamma & $k=0.5$ & L-shaped & 0.637 & \textbf{0.732} \\
    Gamma & $k=2.0$ & Unimodal & 0.375 & 0.675 \\
    \midrule
    Weibull & $c=0.8$ & Stretched-exponential tail & 0.581 & \textbf{0.721} \\
    Weibull & $c=1.5$ & Faster-than-exponential tail & 0.371 & 0.668 \\
    \midrule
    Lognormal & $\sigma=1.5$ & Larger log-scale spread & 0.711 & \textbf{0.802} \\
    Lognormal & $\sigma=0.5$ & Smaller log-scale spread & 0.276 & 0.678 \\
    \bottomrule
\end{tabular}%
\end{table}

With mean fixed, heavier tails increase both the cutoff and the agent's patience. For the Weibull family, moving from a lighter tail ($c=1.5$) to a heavier tail ($c=0.8$) raises the cutoff from $\hat c\approx 0.371$ to $\hat c\approx 0.581$, and raises the unconditional patience scalar from $\hat\psi\approx 0.668$ to $\hat\psi\approx 0.721$. This directly changes how often discretion is exercised early: the probability of spending in period~1 drops from $1-\hat\psi\approx 0.332$ to $1-\hat\psi\approx 0.279$ (about a $16\%$ reduction in early overrides). The behavioral implication is visible at the decision level: a realized gain of $\Delta_1=0.5$ triggers an override under the lighter-tailed case ($0.5>0.371$) but is deferred under the heavier-tailed case ($0.5<0.581$), reflecting the option value of waiting for rarer opportunities.

\paragraph{Extreme regimes.}
Table~\ref{tab:extreme-heavy} reports selected extreme configurations. The Lognormal and Pareto examples have high estimated patience. The near-degenerate Gamma example gives $\hat\psi\approx0.655$. Note that this value is specific to that example, not a universal floor for near-degenerate distributions.
 
\begin{table}[ht]
    \centering
    \caption{Extreme regimes: empirical $\hat{\psi}$ from $10^6$ Monte Carlo draws. Excess kurtosis is a finite-sample statistic used as a proxy for tail heaviness (Normal: $\gamma_2=0$; larger values indicate heavier tails when the fourth moment is finite \cite{nistKurtosis}).}
    \label{tab:extreme-heavy}
    \small
        \begin{tabular}{lccc}
        \toprule
        \textbf{Distribution} & \textbf{Tail Behavior} & \textbf{Excess Kurtosis} & \textbf{Patience $\hat{\psi}$} \\
        \midrule
        Lognormal ($\sigma=5$) & Extreme Heavy & $\approx 8 \times 10^{5}$ & \textbf{0.994} \\
        Pareto ($b=1.05$) & Power Law & $\approx 6 \times 10^{5}$ & \textbf{0.962} \\
        Gamma ($k=10^4$) & Near-degenerate Gamma & $\approx 0$ & \textbf{0.655} \\
        \bottomrule
        \end{tabular}%
\end{table}
The extreme-tail Monte Carlo estimates can be unstable, and should be interpreted as such: Pareto with shape $1.05$ has infinite variance, and the Lognormal with $\sigma=5$ is highly variable.

\subsection{Behavioral implications of shape}
\label{app:behavioral_implications}
 
The patience scalar $\psi$ determines the agent's observable behavior relative to two benchmarks: a full-information oracle and the agent's own across-period allocation. The following lemmas make this precise.

\begin{restatable}[DP vs.\ Oracle Benchmark, two-period]{lemma}{LemmaOne}
\label{lem:dp_v_oracle}
Let $\pi^{*}$ be the optimal DP policy and $\tilde{\pi}$ the oracle policy with perfect foresight of $(\Delta_1,\Delta_2)$.
\begin{enumerate}
    \item At $t=1$, the DP policy is more conservative (less likely to spend) than the oracle if and only if $F_{\D \mid \MlignT{t}}(\mathcal{T}_{2,1}) > \tfrac{1-p}{2}.$

    \item At $t=2$, the DP policy is more aggressive (more likely to spend) than the oracle if and only if $F_{\D \mid \MlignT{t}}(\mathcal{T}_{2,1}) > \tfrac{1}{2}.$

\end{enumerate}
\end{restatable}


\begin{proof}[Proof of Lemma \ref{lem:dp_v_oracle}]
The proof proceeds by direct comparison of the probabilities derived in Equations \eqref{eq:oracle_prob}, \eqref{eq:dp_prob1}, and \eqref{eq:dp_prob2}.

\paragraph{Part 1 (Comparison at $t=1$):}
The DP policy is more conservative than the oracle at $t=1$ when:
\begin{align*}
    \PrDP \paren{\text{spend at }1 \mid \MlignT{1}} &< \PrOracle \paren{\text{spend at }1 \mid \MlignT{1}} \\
    1 - F_{\D \mid \MlignT{t}}(\Th_{2,1}) &< \frac{1+p}{2} \\
    \frac{1-p}{2} &< F_{\D \mid \MlignT{t}}(\Th_{2,1}).
\end{align*}
The condition is reversed if the DP policy is more aggressive.

\paragraph{Part 2 (Comparison at $t=2$):}
The DP policy is more aggressive than the oracle at $t=2$ when:
\begin{align*}
    \PrDP \paren{\text{spend at }2 \mid \MlignT{2}} &> \PrOracle \paren{\text{spend at }2 \mid \MlignT{2}} \\
    p + (1-p)F_{\D \mid \MlignT{t}}(\Th_{2,1}) &> \frac{1+p}{2} \\
    (1-p)F_{\D \mid \MlignT{t}}(\Th_{2,1}) &> \frac{1+p}{2} - p \\
    (1-p)F_{\D \mid \MlignT{t}}(\Th_{2,1}) &> \frac{1-p}{2} \\
    F_{\D \mid \MlignT{t}}(\Th_{2,1}) &> \frac{1}{2} \quad (\text{since } 1-p > 0).
\end{align*}
The condition is reversed if the DP policy is more conservative.
\end{proof}

\begin{restatable}[Internal DP Behavior, two-period]{lemma}{LemmaTwo}
\label{lem:internal_dp}
    The optimal DP policy $\pi^{*}$ is more likely to spend its budget at $t=2$ than at $t=1$ (conditional on misalignment in each period) if and only if $F_{\D \mid \MlignT{t}}(\mathcal{T}_{2,1}) > \tfrac{1-p}{2-p}.$

\end{restatable}


\begin{proof}[Proof of Lemma \ref{lem:internal_dp}]
We compare the conditional probability of spending at $t=2$ with that at $t=1$. The DP is more likely to spend at $t=2$ if:
\begin{align*}
    \PrDP \paren{\text{spend at }2 \mid \MlignT{2}} &> \PrDP \paren{\text{spend at }1 \mid \MlignT{1}} \\
    p + (1-p)F_{\D \mid \MlignT{t}}(\Th_{2,1}) &> 1 - F_{\D \mid \MlignT{t}}(\Th_{2,1}) \\
    F_{\D \mid \MlignT{t}}(\Th_{2,1}) + (1-p)F_{\D \mid \MlignT{t}}(\Th_{2,1}) &> 1 - p \\
    (1 + 1 - p) F_{\D \mid \MlignT{t}}(\Th_{2,1}) &> 1 - p \\
    (2-p)F_{\D \mid \MlignT{t}}(\Th_{2,1}) &> 1-p \\
    F_{\D \mid \MlignT{t}}(\Th_{2,1}) &> \frac{1-p}{2-p} \quad (\text{since } 2-p > 0).
\end{align*}
This shows that the DP defers spending to the second period when the cumulative mass of ``small'' gains below the threshold $\Th_{2,1}$ is sufficiently large. The inequality is reversed if the DP is more likely to spend early.
\end{proof}

Since $\psi = p + (1-p)F_{\Delta|\mathsf{A}_t^c}(\mathcal{T}_{2,1})$ (Appendix~\ref{app:psi_derivation}), the condition in Lemma~\ref{lem:internal_dp} is equivalent to $\psi > \frac{1}{2-p}$, and the condition in part 1 of Lemma~\ref{lem:dp_v_oracle} is equivalent to $\psi > \frac{1+p^2}{2}$. The two lemmas combine into a single comparative-statics statement linking tail thickness to patience:
 
\begin{corollary}[Tail thickness and patience]
\label{cor:tails_and_patience}
Consider a family of base shapes $\mathcal{D}_\eta$ where the parameter $\eta$ increases tail thickness. If $\psi(\mathcal{D}_\eta)$ is increasing in $\eta$, then for a fixed alignment probability $p$:

\begin{itemize}
  \item As $\psi$ exceeds $\frac{1}{2-p}$, the DP becomes \textit{internally patient}, spending more often at $t=2$ than at $t=1$ (Lemma~\ref{lem:internal_dp}).
  \item As $\psi$ exceeds $\frac{(1+p^2)}{2}$, the DP becomes more conservative
  than the full-information oracle at $t=1$ (Lemma~\ref{lem:dp_v_oracle}).
\end{itemize}

\end{corollary}

\section{Computing per-period spending profiles}
\label{app:spending_profiles}

This appendix details the computation of the observable spending curve for the general $(T, K)$ problem, as summarized in Section~\ref{subsec:behavioral_implications}.

\subsection{Conditional spending probability}
\label{appen:simpl_six_pnt_one}

For a fixed state $(\tau, k)$, the agent spends if a misalignment occurs \emph{and} the gain clears the threshold:
\begin{equation*}
    q_{\tau,k} = \Pr(\text{spend} \mid \tau, k) = (1-p) \cdot \Pr(\Delta > \Th_{\tau,k} \mid \Mlign).
\end{equation*}
We show this equals $q_{\tau,k} = 1 - F_\D(\Th_{\tau,k})$. The unconditional CDF of $\D$ at $x \geq 0$ is:
\begin{equation*}
    F_\D(x) = \Pr(\D \leq x) = p \cdot \underbrace{\Pr(\D \leq x \mid \Align)}_{= 1 \text{ since } \D = 0 \text{ under } \Align} + (1-p) \cdot \Pr(\D \leq x \mid \Mlign) = p + (1-p) \cdot F_{\D|\Mlign}(x).
\end{equation*}
Therefore $1 - F_\D(x) = (1-p)(1 - F_{\D|\Mlign}(x)) = (1-p) \cdot \Pr(\D > x \mid \Mlign)$, which gives $q_{\tau,k} = 1 - F_\D(\Th_{\tau,k})$. By Theorem~\ref{thm:general-invariance}(iii), this probability is invariant to scale and location: $q_{\tau,k} = 1 - G_\mathcal{D}(\theta_{\tau,k}(\mathcal{D}))$.

\subsection{Evolution of budget distribution as described in Section~\ref{subsec:behavioral_implications}}

Since spending is stochastic, the remaining budget $k$ at time $t$ is a random variable. Let $\pi_t(k)$ be the probability that the agent has exactly $k$ budget units remaining at the start of period $t$. We compute this via a forward pass:

\paragraph{Initialize.} $\pi_1(K) = 1$, and $\pi_1(k) = 0$ for $k < K$.

\paragraph{Forward update.} For $t = 1, \dots, T-1$, let $\tau = T - t + 1$. The budget evolves as:
\begin{equation}
    \label{eq:pi_update_app}
    \pi_{t+1}(k) = \underbrace{\pi_t(k)\big(1 - q_{\tau,k}\big)}_{\text{had } k \text{, didn't spend}} 
    + \underbrace{\pi_t(k+1)  q_{\tau,k+1}}_{\text{had } k{+}1 \text{, spent}}.
\end{equation}

\paragraph{Boundary terms.} $\pi_t(K+1) = 0$, $q_{\tau,0} = 0$ (no budget), and $q_{\tau, k \geq \tau} = 1-p$ (ample budget, $\Th_{\tau,k} = 0$).

\paragraph{Aggregate spending profile.} The unconditional probability of overriding at time $t$ is the expectation over budget states:
\begin{equation}
    \label{eq:aggregate_spending_app}
    \Pr(\text{spend at time } t) = \sum_{k=1}^{K} \pi_t(k) \cdot q_{\tau,k}.
\end{equation}
This yields the spending profile $\big(\Pr(\text{spend at } 1), \dots, \Pr(\text{spend at } T)\big)$, the observable ``spending curve'' of the optimal policy. Since each $q_{\tau,k}$ is scale- and location-free (Theorem~\ref{thm:general-invariance}), and the budget distribution $\pi_t(k)$ is computed from these probabilities via the deterministic update~\eqref{eq:pi_update_app}, the entire spending profile inherits the same invariance. This $T$-vector is the natural generalization of the two-period scalar $\psi$: it captures all observable discretionary behavior and depends only on $\mathcal{D}$, $p$, $T$, and $K$.

\subsection{Threshold visualization}
\label{appen:th_viz}

Figure~\ref{fig:threshold_heatmap} shows the optimal threshold $\mathcal{T}_{\tau,k}$ at every state $(\tau,k)$ for Pareto-distributed gains with $T=20$ and $K=5$.

\begin{figure}[H]
    \centering
    \includegraphics[width=0.85\linewidth]{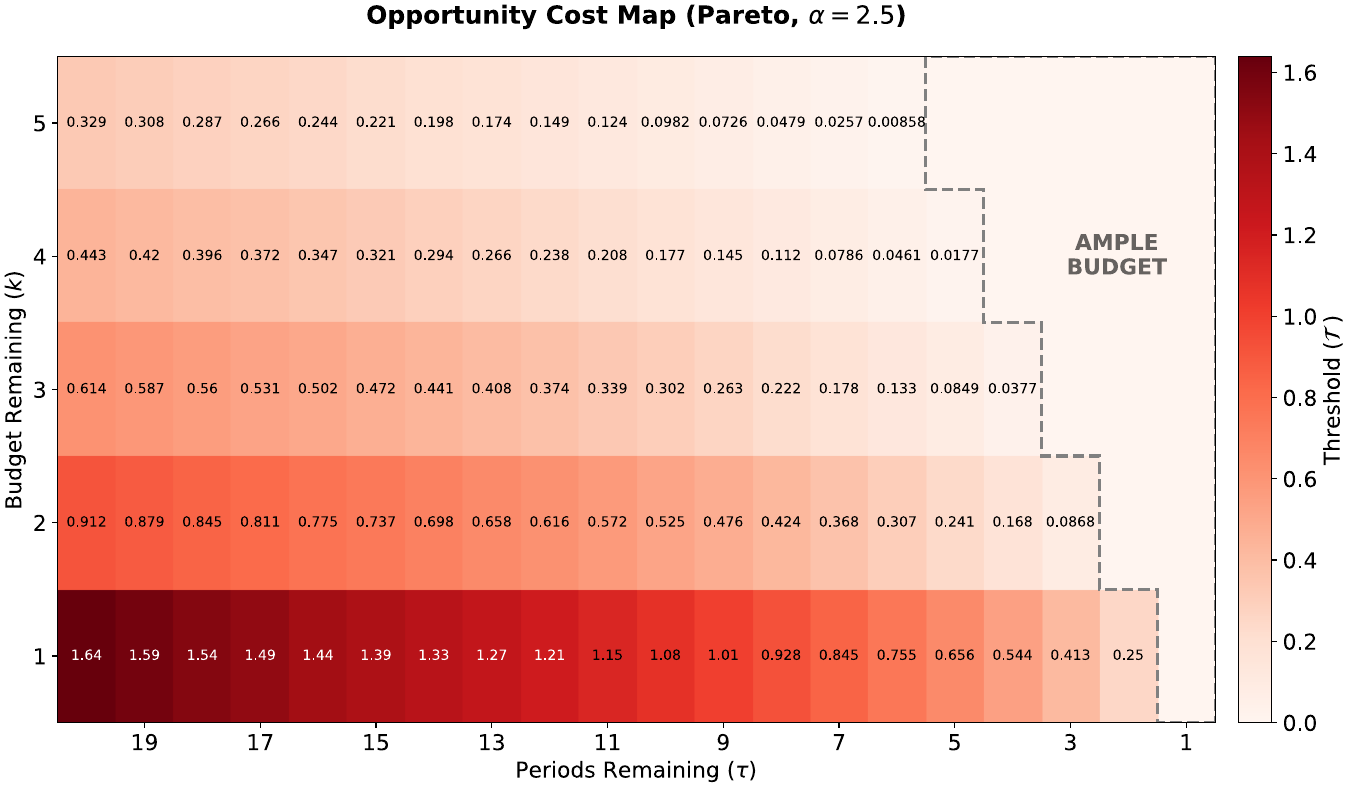}
    \caption{The opportunity cost map: heatmap of optimal thresholds $\mathcal{T}_{\tau,k}$ for ($T=20, K=5$) with a Pareto distribution scaled to mean 1. Darker regions indicate states where the agent requires a very high gain to justify spending. The thresholds naturally decay as time runs out ($\tau \to 0$) or budget becomes abundant relative to time ($k \to \tau$).}
    \label{fig:threshold_heatmap}
\end{figure}
\section{Empirical analysis: full details}
\label{app:empirical-details}

To complement our theoretical characterization of threshold-based override behavior, we examine placement records from a Homeless Management Information System (HMIS) to test whether discretionary deviations from a recovered baseline policy vary systematically with operational constraints. We ask two questions. First, how does discretionary behavior respond to shifts in the operational constraints that bind the caseworker? And second, how is the use of discretion spread over the calendar?

\paragraph{Computing infrastructure used:} All experiments and analysis were conducted on a 16-inch MacBook Pro (2024). The machine was equipped with an Apple M4 Pro system on a chip (SoC), featuring a 14-core CPU, a 20-core GPU, and a 16-core Neural Engine. It included 24 GB of unified memory and a 512 GB solid-state drive (SSD), running macOS 15.

\subsection{Operational setting: agent, policy, and environment}
We map our theoretical framework to the context of homeless service allocation. The caseworker is the decision-maker, and the allocation setting presents the same temporal tradeoff formalized in Section~\ref{sec:the_model}: Transitional Housing (TH) serves as the scarce, resource-intensive intervention, while Emergency Shelter (ES) serves as the default, lower-intensity baseline. Caseworkers possess the discretionary authority to override the default assignment to place a client in scarce transitional housing or, alternatively, move clients from scarce housing to the more readily available option (ES in this case).

To measure this discretion, we must first define the ``Policy'' from which the agent deviates. Drawing on recent work by \citet{pokharel2024discretionary}, we operationalize the default policy not as a rigid bureaucratic manual—which did not exist in this historical period—but as the dominant set of heuristic rules governing the system. Street-level bureaucracy theory posits that agents operate under cognitive constraints by developing simplified heuristics to process routine cases quickly \cite{lipsky1980street}. We model these heuristics using the Short Explainable Rules - Decision Tree (SER-DT) algorithm \cite{SERDT}. SER-DT explicitly optimizes for interpretability by penalizing complexity, producing short, depth-constrained decision trees (we set the maximum depth to 4) that capture systematic variance in assignments. We treat the prediction $P_i$ of these trees as the observable baseline: it represents the decision a ``typical'' bureaucrat in this time-period would make. Thus, using the same definition of discretion as \citet{pokharel2024discretionary}, a discretionary event is defined as any instance in which the caseworker's actual decision differs from this baseline ($A_i \neq P_i$).

\subsection{Data description and filtering}
We analyze administrative records from the St.\ Louis Homeless Management Information System (HMIS), covering households assessed between 2008 and 2014. The data includes records of service entry and exit, linked to detailed household characteristics.

\paragraph{Time context (the role of discretion).}
The study period (2008--2014) represents a unique window for observing natural discretionary behavior. During this time, the homeless service system lacked more explicit, formulaic prioritization mandates (such as the VI-SPDAT \cite{orgcode2015vulnerability} adopted later). Instead, caseworkers relied on professional judgment to balance household vulnerability against expected outcomes.

\paragraph{Feature selection.}
To train the SER-DT policy model described above, we utilize 34 household features available upon request for services, such as age, gender, income, disability status, and family composition. These features represent the information set available to the agent when forming the baseline heuristic rules ($P_i$) (the list of all features used is detailed in Table~\ref{appen:data_features}).

\subsection{Operationalization of discretion}
\label{appen:Operationalization_disc}
For each household $i$ in the dataset, we observe two key variables: (i) the predicted assignment generated by the heuristic policy ($P_i$), and (ii) the caseworker's actual final placement decision ($A_i$).

We restrict our attention to the two primary intervention types: ES and TH. We exclude observations from other project types (i.e., rapid rehousing, prevention) to maintain consistency with the finite-budget resource allocation problem defined in Section \ref{sec:the_model}, as ES and TH provide a clear hierarchy on the intensity of intervention assignment with TH being the more resource-intensive intervention. The restriction applies to the tree's recommendation as well as to the placement: overrides are counted only among households whose recommendation and placement are both ES or TH, and these households are the trials of the overall model. We then define a discretionary event as an instance in which the caseworker overrides the default policy recommendation ($A_i \neq P_i$). We classify these overrides into three indicators to capture the direction of resource usage:

\begin{itemize}
    \item \textbf{Overall Discretion $\big( D_{i}^{\text{all}} \big)$:} An indicator that takes value 1 if the caseworker deviates from the policy in any direction.
        \begin{equation*}
            D_{i}^{\text{all}} = \mathbb{I}\{P_i \neq A_i\}.
        \end{equation*}
        
    \item \textbf{Upward Discretion / ``Upgrading'' $\big( D_{i}^{\uparrow}\big)$:} An indicator for instances where the policy recommends the default option (ES), but the caseworker ``upgrades'' the client to the scarce resource (TH).
    \begin{equation*}
        D_{i}^{\uparrow} = \mathbb{I}\{P_i = \text{ES} \wedge A_i = \text{TH}\}.
    \end{equation*}
    In the language of our dynamic model, these are instances where the caseworker ``spends'' the discretionary budget to capture a perceived utility gain.
    
    \item \textbf{Conservative Discretion / ``Rationing'' $\big( D_{i}^{\downarrow}\big)$:} An indicator for instances where the policy recommends the scarce resource (TH), but the caseworker picks the `default' option (ES).
        \begin{equation*}
            D_{i}^{\downarrow} = \mathbb{I}\{P_i = \text{TH} \wedge A_i = \text{ES}\}.
        \end{equation*}
    These cases represent a refusal to allocate the more intensive resource even when the policy recommends it, effectively ``saving'' the budget for future periods.
\end{itemize}

\paragraph{Errors in the recovered baseline} No record of the routine exists for this period, so some cases counted as overrides may be tree errors, and some real overrides may be missed. The tree sees only the 34 household features in Table~\ref{appen:data_features}, none of which are the calendar or capacity regressors of Appendix~\ref{subsec:feature_engineering}. If its error rates do not move with those regressors, misclassification adds overrides on quiet days and hides them on busy ones, pulling the odds ratios toward one: tree error can weaken a real association but cannot create one. Error rates could move with the regressors only through the case mix, for example if households arriving after a weekend differ in ways that change how often the tree makes a mistake. We assume the case mix is stable in this sense. We cannot test this, because no part of the 2008–2014 data records the routine independently, and for the same reason the share of measured overrides that are tree errors cannot be estimated. Two results suggest the recovered routine is reliable. Ten refits of the tree in \citet{pokharel2024discretionary} give nearly identical predictions (mean pairwise Spearman correlation 0.985 for TH and 0.890 for ES), and in our data XGBoost at its default settings predicts placements far better than the depth-4 trees for both ES and TH (Appendix~\ref{apen:tree_depth_sensitvity}), so deviations from the short rule follow consistent patterns rather than noise.

\paragraph{From household decisions to daily binomial outcomes}

Our theoretical model characterizes optimal overrides through a dynamic threshold rule. Behaviorally, this appears as a time-varying probability of exercising discretion. Although the theory is formulated in abstract decision periods, the operational reality of the homeless services system--staffing schedules, intake workflows, and resource tracking--is organized by \emph{calendar days}. To bridge this gap, we aggregate individual placement decisions to the daily level and treat each day as a binomial experiment, measuring the frequency of discretionary overrides under varying system conditions. Let $t$ index calendar days and let $\mathcal{I}_t$ denote the set of households whose entry date falls on day $t$ and whose realized placement lies in $\{\text{ES},\text{TH}\}$. Define the daily volume $n_t = \lvert \mathcal{I}_t \rvert$, and the daily counts of discretionary overrides as: 

\begin{equation*}
    Y_t^{\text{all}} = \sum_{i \in \mathcal{I}_t} D_i^{\text{all}},
    \qquad
    Y_t^{\uparrow} = \sum_{i \in \mathcal{I}_t} D_i^{\uparrow},
    \qquad
    Y_t^{\downarrow} = \sum_{i \in \mathcal{I}_t} D_i^{\downarrow}.
\end{equation*}

\paragraph{Reference sets for directional discretion.} For directional discretion, the relevant denominator is not total daily volume, but the set of cases for which that direction is \emph{feasible}. We refer to these as eligible \textit{trial sets}. Concretely, Upgrading (ES$\rightarrow$TH) is only possible for households whose baseline recommendation is ES, while Rationing (TH$\rightarrow$ES) is only possible for households recommended for TH. Accordingly, we define daily trial-set sizes:

\begin{equation*}
    n_t^{\text{ES}}=\sum_{i\in\mathcal{I}_t}\mathbb{I}\{P_i=\text{ES}\},
    \qquad
    n_t^{\text{TH}}=\sum_{i\in\mathcal{I}_t}\mathbb{I}\{P_i=\text{TH}\},
\end{equation*}
and we model daily override counts as binomial draws: 
\begin{equation*}
    Y_t^{\text{all}} \sim \text{Binom}\paren{n_t, p_t^{\text{all}}},
    \qquad
    Y_t^{\uparrow} \sim \text{Binom}\paren{n_t^{\text{ES}}, p_t^{\uparrow}},
    \qquad
    Y_t^{\downarrow} \sim \text{Binom}\paren{n_t^{\text{TH}}, p_t^{\downarrow}},
\end{equation*}
restricting to days with active caseloads ($n_t>0$) and, for directional models, non-empty trial sets $\paren{n_t^{\text{ES}}>0 \text{ or }n_t^{\text{TH}}>0}$.

\subsection{Constructing operational state variables}
\label{subsec:feature_engineering}

To test the predictions of the Budgeted Discretion framework, we reconstruct daily proxies for the operational ``state'' visible to the agent. We engineer features intended to capture (i) workload and throughput, (ii) short-run slack in ES and TH, and (iii) temporal cycles. To preserve temporal ordering (and to avoid mechanical dependence on same-day decisions), all operational features are constructed using only information realized \emph{strictly prior} to day $t$.

\subsubsection{Rolling state (immediate operational conditions)}
Our primary specification summarizes recent system conditions over a rolling 7-day window. For each day $t$, we compute total inflow (assignments) and outflow (exits) over $[t-7,t-1]$:
\begin{equation*}
    L_t^{\text{roll}}
    =
    \Big[
    \text{Assign}^{\text{ES}}_{t-7:t-1}, \qquad
    \text{Assign}^{\text{TH}}_{t-7:t-1},\qquad
    \text{Exit}^{\text{ES}}_{t-7:t-1},\qquad
    \text{Exit}^{\text{TH}}_{t-7:t-1}
    \Big]^\top.
\end{equation*}
Assignments proxy recent throughput and workflow burden. Exits proxy \emph{openings} (short-run slack) for each intervention. When exits are high, recent capacity has been freed; when exits are low, the intervention is tighter. In the language of our model, these exit-based measures provide a reduced-form proxy for tightness and, thus, for the effective threshold at which discretion is exercised.

\subsubsection{Weekly block state (administrative cycle robustness)}
As a robustness check, we replace rolling totals with totals from the prior full calendar week (lagged by one full week):
\begin{equation*}
    L_t^{\text{block}}
    =
    \Big[
    \text{Assign}^{\text{ES}}_{\text{prev wk}},\qquad
    \text{Assign}^{\text{TH}}_{\text{prev wk}},\qquad
    \text{Exit}^{\text{ES}}_{\text{prev wk}},\qquad
    \text{Exit}^{\text{TH}}_{\text{prev wk}}
    \Big]^\top.
\end{equation*}
This block specification captures coarser batching and reporting cycles. The results are generally robust to this alternative state definition; see Appendix~\ref{appen:glm_emperics}.

\subsubsection{Lagged TH composition}
We also control for the prevailing ES/TH mix using the lagged (by one day) share of TH placements:
\begin{equation}
    \text{THShare}_{t-1} = \frac{\text{TH}_{t-1}}{n_{t-1}},
    \qquad
    \text{TH}_{t-1}=\sum_{i\in\mathcal{I}_{t-1}}\mathbb{I}\{A_i=\text{TH}\}.
\end{equation}
Using the lagged value avoids mechanical correlation with same-day placements while absorbing slow-moving shifts in the homeless service system's placement regime.

\subsubsection{Temporal structure}
We include calendar controls, capturing systematic cycles in operations:
\begin{itemize}
    \item \textbf{Day-type indicators:} to separate a start-of-week effect from weekend operations, we use a three-level indicator $\text{DayType}(t)\in\{\text{Mon},\text{Tue--Fri},\text{Weekend}\}$ with Tue--Fri as the reference.
    \item \textbf{Month-of-year fixed effects:} indicators for $\text{month}(t)\in\{1,\dots,12\}$ with October as the reference, since October serves as the start of the federal fiscal year.
    \item \textbf{Federal holidays:} an indicator $H_t$ for U.S.~federal holidays.
\end{itemize}

\subsection{Model specification}
\label{subsec:glm_spec}

We estimate the drivers of daily discretion using a generalized linear model (GLM) with a binomial family and logit link. For each outcome $j\in\braces{\text{all},\uparrow,\downarrow}$ and state specification $L_t\in\braces{L_t^{\text{roll}},L_t^{\text{block}}}$, we estimate:
\begin{equation*}
\label{eq:glm_daily}
    \log \paren{\frac{p_t^{j}}{1-p_t^{j}}}
    =
    \alpha^{j}_{\text{DayType}(t)}
    +
    \gamma^{j}_{\text{month}(t)}
    +
    \beta^{j}_{S}\cdot \text{THShare}_{t-1}
    +
    \beta^{j}_{H}\cdot H_t
    +
    (\beta^{j}_{L})^\top L_t.
\end{equation*}
Observations are weighted by their respective trial counts: $n_t$ for Overall Discretion, $n_t^{\text{ES}}$ for Upgrading, and $n_t^{\text{TH}}$ for Rationing. This ensures that days with higher relevant volumes contribute proportionally more to the likelihood.

\subsection{Features used in SER-DT}
\label{appen:features_table}

\begin{table}[H]
\centering
\caption{A comprehensive description of all the features that we used in the creation of the decision trees is shown in this table. All categorical variables were one-hot encoded.}
\label{appen:data_features}
\resizebox{\textwidth}{!}{%
\begin{tabular}{@{}ll@{}}
\toprule
\multicolumn{2}{@{}l}{\textit{Binary features}} \\
\midrule
HUDChronicHomeless & Whether or not the client is chronically homeless. \\
Gender & The original documentation allows 8 values, but the data contain only Male and Female. \\
SpousePresent & Variable indicating whether or not client has a spouse. \\
\midrule
\midrule
\multicolumn{2}{@{}l}{\textit{Categorical features}} \\
\midrule
PrimaryRace & The primary race of client. Can be one of 7 categories. \\
Ethnicity & The ethnicity of client. Can be one of 5 categories. \\
PriorResidence & The type of residence client had before entering system. Can be one of 25 categories. \\
VeteranStatus & Variable indicating whether or not client is a veteran. Can take one of 5 values. \\
DisablingCondition & Variable indicating whether or not client has a disabling condition. Can be one of 3 categories. \\
ReceivePhysicalDisabilityServices & Variable indicating whether or not client received disability services. Can be one of 5 categories. \\
HasDevelopmentalDisability & Variable indicating whether or not client has developmental disability. Can be one of 5 categories. \\
ReceiveDevelopmentalDisabilityServices & Variable indicating whether or not client received services for developmental disability. Can be one of 5 categories. \\
HasChronicHealthCondition & Variable indicating whether or not client has a chronic health condition. Can be one of 5 categories. \\
ReceiveChronicHealthServices & Variable indicating whether or not client received services for a chronic health condition. Can be one of 5 categories. \\
HasHIVAIDS & Variable indicating whether or not client has HIV/AIDS. Can be one of 5 categories. \\
ReceiveHIVAIDSServices & Variable indicating whether or not client received services for HIV/AIDS. Can be one of 5 categories. \\
HasMentalHealthProblem & Variable indicating whether or not client has a mental health condition. Can be one of 5 categories. \\
ReceiveMentalHealthServices & Variable indicating whether or not client received services related to mental health. Can be one of 5 categories. \\
HasSubstanceAbuseProblem & Variable indicating whether or not client has a substance abuse problem. Can be one of 5 categories. \\
ReceiveSubstanceAbuseServices & Variable indicating whether or not client received services related to substance abuse. Can be one of 5 categories. \\
DomesticViolenceSurvivor & Variable indicating whether or not client is a survivor of domestic violence. Can be one of 5 categories. \\
\midrule
\midrule
\multicolumn{2}{@{}l}{\textit{Continuous features}} \\
\midrule
Age & The age of the client. \\
Calls & The number of calls made to hotline before entry. \\
Wait & The time (in days) elapsed since first call to hotline to intervention assignment. \\
RatioOfNumCallstoWaitTime & Ratio of number of calls to wait time. Calculated variable to measure `persistence'. \\
MonthlyAmount & The total income (in dollars) of client per month. \\
numMembers & The number of members in the client's household. \\
Children & The number of children in client's household. \\
Children 0-2 & The number of children aged 0 to 2 in client's household. \\
Children 3-5 & The number of children aged 3 to 5 in client's household. \\
Children 6-10 & The number of children aged 6 to 10 in client's household. \\
Children 11-14 & The number of children aged 11 to 14 in client's household. \\
Children 15-17 & The number of children aged 15 to 17 in client's household. \\
UnrelatedChildren & The number of children in client's household not related to them. \\
UnrelatedAdults & The number of adults in client's household not related to them. \\
\bottomrule
\bottomrule
\end{tabular}%
}
\end{table}

\FloatBarrier
\subsection{Sensitivity to tree depths}
\label{apen:tree_depth_sensitvity}

The complexity of the tree determines which departures from it are counted as discretion. Street-level bureaucracy theory describes routine behavior in terms of a small set of simple rules \citep{lipsky1980street}. A more flexible model can incorporate systematic departures from those rules into the recovered routine, leaving fewer placements classified as overrides. The maximum tree depth therefore forms part of our definition of discretion. Here, we ask whether the empirical conclusions of Section~\ref{sec:empirical} hold across alternative depth limits and fitting methods. We repeat the analysis with the short explainable decision trees (SER-DT) used in the main analysis \citep{SERDT}, fitted at maximum depths 2, 3, 5, and 6. For predictive fit, override rates, and random-group tests, we also fit generalized and scalable optimal sparse decision trees \citep{lin2020generalized} using the guessing strategies of \citet{mctavish2022fast} (GOSDT-guesses) at depths 2 through 6. An XGBoost ensemble \citep{chen2016xgboost} at its default settings, one model per intervention with 100 trees of depth at most 6 and a learning rate of 0.3, provides a more flexible benchmark for these comparisons.

Our reference baseline consists of the depth-4 SER-DT trees from \citet{pokharel2024discretionary}, with one tree for each intervention: emergency shelter (ES), transitional housing (TH), rapid rehousing (RRH), and homelessness prevention. The tables in the main body use predictions from these trees. At depth 4, we report their area under the receiver operating characteristic curve (AUC) and random-group percentiles, and use their predictions in the daily models reported in Section~\ref{sec:empirical}.

In keeping with the protocol of \citet{pokharel2024discretionary}, we randomly split the households ten times, allocating 70\% to fitting and 15\% each to test and validation sets. For each intervention and depth, we retain the tree with the highest test accuracy across the ten splits. GOSDT-guesses uses the same splits, features, and selection criterion, with a penalty of 0.001 per leaf and candidate thresholds guessed by 40 boosted stumps. XGBoost uses five-fold cross-fitting. For the daily analyses, the overall model uses the specification in Table~\ref{tab:discretion_summary_daytype}, which counts a day's overrides among placements predicted ES or TH. The depth-4 upgrading and rationing models are those in Tables~\ref{tab:updrading_discretion_mon} and~\ref{tab:rationing_discretion_mon}.

\begin{figure}[htbp]
\centering
\begin{subfigure}[t]{0.49\linewidth}
\includegraphics[width=\linewidth]{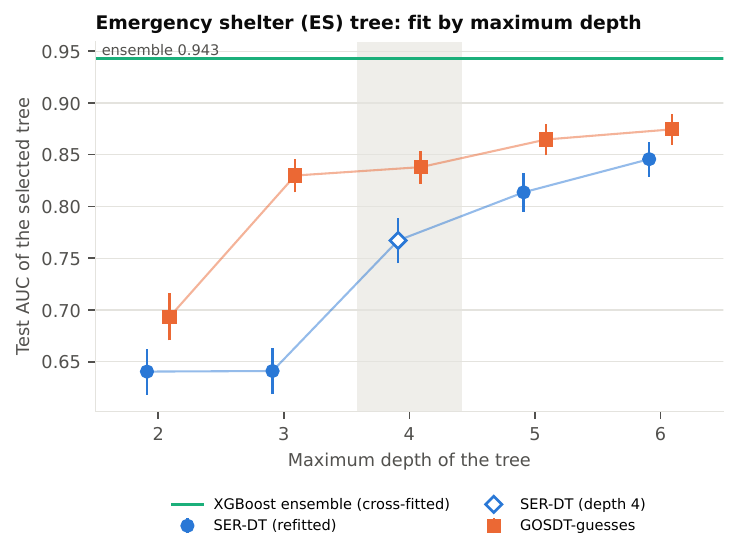}
\caption{}
\end{subfigure}\hfill
\begin{subfigure}[t]{0.49\linewidth}
\includegraphics[width=\linewidth]{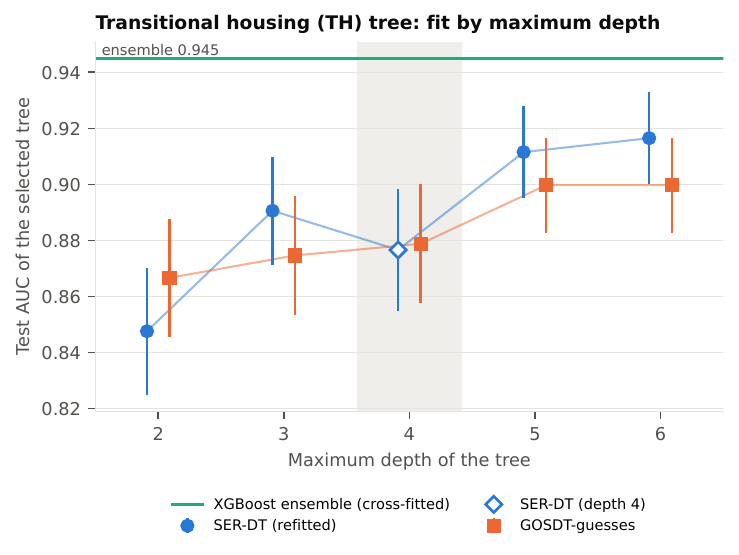}
\caption{}
\end{subfigure}\\[4pt]
\begin{subfigure}[t]{\linewidth}
\includegraphics[width=\linewidth]{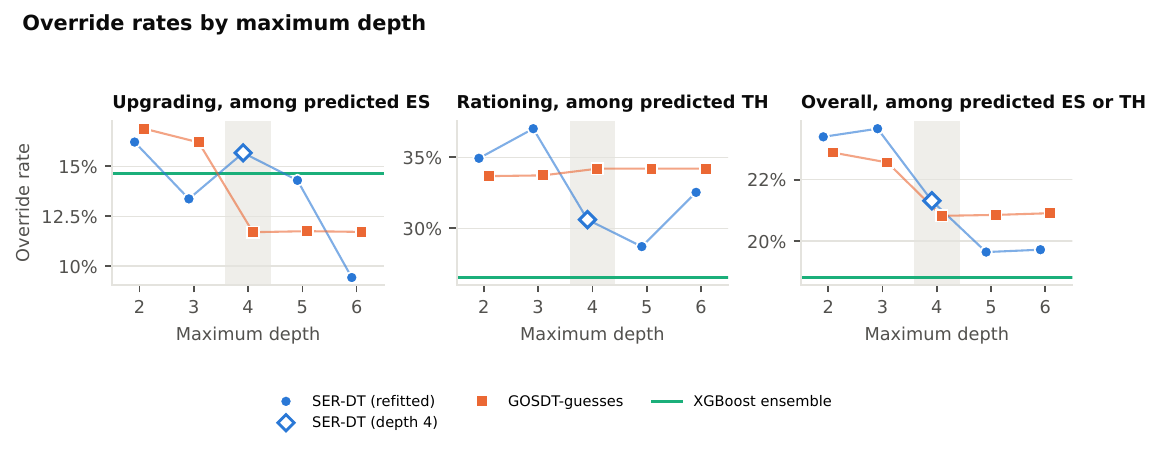}
\caption{}
\label{fig:depth_override_rates}
\end{subfigure}
\caption{Predictive fit and measured discretion by maximum tree depth. Vertical axes are not shared: panels (a) and (b) use different scales, and each of the three panels in (c) has its own scale. (a, b) Test AUCs for the selected ES and TH trees, with DeLong 95\% confidence intervals. Horizontal lines show XGBoost's cross-fitted AUCs.  (c) Upgrading rates among placements predicted ES, rationing rates among placements predicted TH, and overall override rates among placements predicted ES or TH. Horizontal lines show the corresponding rates for XGBoost.  Open diamonds denote the depth-4 SER-DT used in the main analysis. Filled circles denote SER-DT refits at the other depths, and squares denote GOSDT-guesses.}
\label{fig:depth_auc}
\end{figure}

\paragraph{Predictive fit.}
Figure~\ref{fig:depth_auc}a--b reports test AUCs for the selected ES and TH trees, with 95\% confidence intervals computed using DeLong's method. The ensemble's cross-fitted AUC provides a benchmark. The AUCs quoted in this paragraph and the next are means of the test AUC over the ten splits; the figure plots the test AUC of the selected tree, which is higher and varies more across depths. Predictive performance generally improves with depth, although the gains taper at different points across interventions and fitting methods. For ES, mean test AUC across the ten SER-DT splits rises from 0.638 at depth 2 to 0.827 at depth 6, with an increase of 0.030 between depths 5 and 6. Under GOSDT-guesses, ES AUC rises from 0.679 to 0.858, but gains fall below 0.01 per additional level from depth 5 onward. SER-DT depth 4, the tree used in the main analysis, has an AUC of 0.761, compared with 0.943 for XGBoost.

For TH, the gains taper earlier. SER-DT AUC rises from 0.837 at depth 2 to 0.908 at depth 6, with gains below 0.01 beyond depth 5; the depth-4 tree used in the main analysis has an AUC of 0.870. GOSDT-guesses AUC ranges from 0.855 to 0.873, with no increase of 0.01 or more between successive depths, while XGBoost achieves an AUC of 0.945. The improvements in fit must be weighed against the complexity of the resulting rules. The SER-DT trees reach their depth limits at every depth, growing from 3 leaves at depth 2 to as many as 38 at depth 6. GOSDT-guesses produces trees with 4 to 10 leaves, smaller than the SER-DT trees at depths 5 and 6, and reaches the depth limit in nine of the ten ES and TH fits. We use depth 4 as our main specification to strike a reasonable balance between predictive performance and interpretability. Deeper trees can improve fit, particularly for ES, but our aim is to recover a compact set of rules that can represent routine decision-making. The analyses below examine whether the substantive conclusions depend on this choice.

\paragraph{Measured discretion.}
A more flexible baseline should account for more of the systematic variation in placements, reducing the share classified as overrides. Figure~\ref{fig:depth_override_rates} broadly supports this expectation. Among placements predicted ES or TH, the overall override rate falls from 23.4\% at depth 2 to 19.7\% at depth 6 under SER-DT, compared with 21.3\% for the original depth-4 baseline. Under GOSDT-guesses, the rate falls from 22.9\% to 20.9\%. XGBoost yields the lowest rate, at 18.8\%. The decline is not monotonic: the SER-DT override rate increases between depths 2 and 3.

Most of the decline comes from upgrading. Its rate falls from 16.2\% to 9.4\% under SER-DT and from 16.9\% to 11.7\% under GOSDT-guesses. Rationing shows no comparable trend. Its rate ranges from 28.7\% to 37.0\% under SER-DT and remains between 33.7\% and 34.2\% across all GOSDT-guesses depths. XGBoost has the lowest rationing rate, at 26.5\%, but its upgrading rate of 14.7\% exceeds that of the deeper GOSDT-guesses trees. Baseline flexibility thus changes the amount of measured discretion, particularly upgrading.

\begin{figure}[htbp]
\centering
\includegraphics[width=0.85\linewidth]{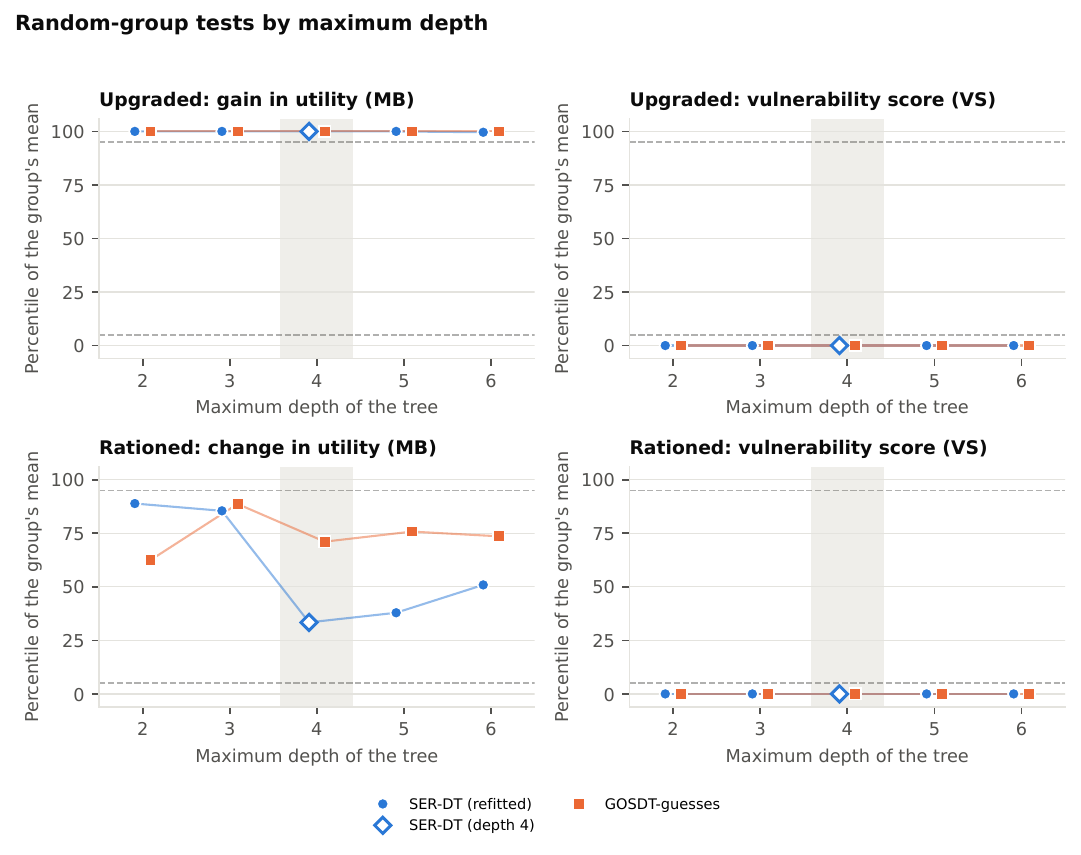}
\caption{Random-group comparisons across tree depths. Each panel reports the percentile of an override group's mean among the means of 1,000 random groups with the same size and predicted intervention. Rows distinguish upgraded households (top) from rationed households (bottom); columns report utility gain (MB, left) and vulnerability score (VS, right). Open diamonds denote the depth-4 SER-DT baseline.}
\label{fig:depth_random_groups}
\end{figure}

\paragraph{Random-group tests.}
We next examine whether changing the models alters the characteristics of households classified as receiving an override. Following \citet{pokharel2024discretionary}, we compare upgraded and rationed households on two measures: the utility gain from the switch they received (MB in the figures) and the vulnerability score (VS). For each group, we draw 1,000 random groups of the same size from households with the same predicted intervention. We then locate the observed group's mean within the distribution of random-group means. The reported percentile is the share of random means below the observed mean. Upgraded households are those predicted ES and placed in TH; rationed households are those predicted TH and placed in ES.

The original analysis found that upgraded households had higher utility gains and lower vulnerability scores than the random groups. Rationed households also had lower vulnerability scores, but their utility gains did not differ significantly. All tested specifications reproduce this pattern (Figure~\ref{fig:depth_random_groups}). For the alternatives---four SER-DT depths, five GOSDT-guesses depths---upgraded households lie between the 99.6th and 100th percentiles for utility gain and at percentile 0.0 for vulnerability. Rationed households also lie at percentile 0.0 for vulnerability. Their utility-gain percentiles range from 37.9 to 88.9, with no statistically significant difference under any specification.

\paragraph{Daily models.}
Finally, we re-estimate the three daily models in Table~\ref{tab:discretion_summary_daytype} using overrides defined by SER-DT at depths 2, 3, 5, and 6. We first consider the ten odds ratios discussed in Section~\ref{sec:empirical}: ES and TH exits in the upgrading and rationing models; Monday, the weekend, November, December, and June in the overall model; and the weekend in the rationing model. At every alternative depth, each estimate lies within the corresponding 95\% confidence interval for the depth-4 tree results reported in Table~\ref{tab:discretion_summary_daytype}, covering all 40 comparisons. This also holds for the full set of 27 odds ratios: all 108 estimates lie within the depth-4 intervals. With the weekly block features in Tables~\ref{tab:discretion_mon}--\ref{tab:rationing_discretion_mon}, 107 of 108 estimates fall within the corresponding depth-4 intervals. The sole exception is the coefficient on ES exits in the upgrading model at depth 6. Its odds ratio is 0.945, just 0.002 below the depth-4 interval's lower bound of 0.947, placing the estimate slightly farther from one.
\section{Generalized linear model results}
\label{appen:glm_emperics}
This appendix reports every coefficient of the daily models of Appendix~\ref{subsec:glm_spec}, for both the weekly block features and the rolling features of Appendix~\ref{subsec:feature_engineering}. Each table gives the odds ratio of each term, its 95\% confidence interval and its $p$-value; $p$-values below 0.05 are in bold, except for the intercept.

\subsection{Weekly cycles}
\label{appen:weekly_glm}
Tables~\ref{tab:discretion_mon}, \ref{tab:updrading_discretion_mon} and~\ref{tab:rationing_discretion_mon} report the models of Table~\ref{tab:discretion_summary_daytype} in full, with the day-type indicator (Monday, Tuesday to Friday, weekend), whose reference is Tuesday to Friday.

\begin{table}[htbp]
\centering
\caption{Results of fitting the GLM for overall discretion, isolating weekly effects. The outcome counts each day's overrides among its placements predicted ES or TH.}
\label{tab:discretion_mon}
\footnotesize
\setlength{\tabcolsep}{4pt}
\renewcommand{\arraystretch}{0.92}
\begin{tabular}{@{}lcccccc@{}}
\toprule
 & \multicolumn{3}{c}{Block features} & \multicolumn{3}{c}{Rolling features} \\
\cmidrule(lr){2-4}\cmidrule(l){5-7}
Variable & Odds ratio & 95\% CI & $p$-value & Odds ratio & 95\% CI & $p$-value \\
\midrule
\multicolumn{7}{@{}l}{\textit{Intercept and resource availability}} \\
Intercept & 0.41708 & [0.30159, 0.57681] & 1.25E-07 & 0.36135 & [0.26132, 0.49966] & 7.46E-10 \\
THShare$_{t-1}$ & 1.0854 & [0.88976, 1.3241] & 0.41895 & 1.06 & [0.859, 1.308] & 0.58698 \\
Holiday & 0.86775 & [0.54688, 1.3769] & 0.54706 & 0.85699 & [0.54011, 1.3598] & 0.51235 \\
ES assignments & 0.99535 & [0.98258, 1.0083] & 0.4793 & 0.99656 & [0.98381, 1.0095] & 0.59987 \\
TH assignments & 0.98944 & [0.97354, 1.0056] & 0.1992 & 1.0059 & [0.98902, 1.023] & 0.49728 \\
ES exits & 0.98813 & [0.9751, 1.0013] & 0.078121 & 0.99468 & [0.98201, 1.0075] & 0.41477 \\
\quad per SD & 0.9362 & [0.86999, 1.0074] & 0.078121 & 0.97004 & [0.90166, 1.0436] & 0.41477 \\
TH exits & 1.0082 & [0.9914, 1.0252] & 0.34206 & 0.99497 & [0.97837, 1.0118] & 0.55656 \\
\quad per SD & 1.0366 & [0.96252, 1.1164] & 0.34206 & 0.97803 & [0.90824, 1.0532] & 0.55656 \\
\addlinespace
\multicolumn{7}{@{}l}{\textit{Weekly cycles (reference: Tue--Fri)}} \\
Mon & 1.2412 & [1.0726, 1.4364] & \textbf{0.0037185} & 1.2437 & [1.0745, 1.4395] & \textbf{0.0034714} \\
Weekend & 0.75872 & [0.60177, 0.9566] & \textbf{0.019531} & 0.76204 & [0.60454, 0.96058] & \textbf{0.021427} \\
\addlinespace
\multicolumn{7}{@{}l}{\textit{Month (reference: October)}} \\
Nov & 0.71122 & [0.51642, 0.9795] & \textbf{0.036903} & 0.70467 & [0.51214, 0.96958] & \textbf{0.031578} \\
Dec & 0.59141 & [0.42812, 0.81698] & \textbf{0.0014417} & 0.59163 & [0.42812, 0.8176] & \textbf{0.0014723} \\
Jan & 0.79854 & [0.59368, 1.0741] & 0.13691 & 0.78886 & [0.58531, 1.0632] & 0.11935 \\
Feb & 0.949 & [0.70558, 1.2764] & 0.7292 & 0.92882 & [0.69123, 1.2481] & 0.62426 \\
Mar & 0.842 & [0.62609, 1.1324] & 0.25527 & 0.81923 & [0.60926, 1.1015] & 0.1869 \\
Apr & 0.87372 & [0.6481, 1.1779] & 0.37575 & 0.87013 & [0.64531, 1.1733] & 0.36165 \\
May & 0.91588 & [0.68211, 1.2298] & 0.55895 & 0.9133 & [0.68053, 1.2257] & 0.5457 \\
Jun & 0.74992 & [0.55261, 1.0177] & 0.064662 & 0.71005 & [0.52386, 0.96241] & \textbf{0.027318} \\
Jul & 0.90081 & [0.67094, 1.2094] & 0.48711 & 0.89886 & [0.67022, 1.2055] & 0.47648 \\
Aug & 0.81364 & [0.60183, 1.1] & 0.18008 & 0.80124 & [0.59337, 1.0819] & 0.14815 \\
Sep & 0.8778 & [0.65028, 1.1849] & 0.3945 & 0.88483 & [0.65584, 1.1938] & 0.42327 \\
\bottomrule
\end{tabular}
\end{table}

\begin{table}[htbp]
\centering
\caption{Results of fitting the GLM for upgrading, isolating weekly effects.}
\label{tab:updrading_discretion_mon}
\footnotesize
\setlength{\tabcolsep}{4pt}
\renewcommand{\arraystretch}{0.92}
\begin{tabular}{@{}lcccccc@{}}
\toprule
 & \multicolumn{3}{c}{Block features} & \multicolumn{3}{c}{Rolling features} \\
\cmidrule(lr){2-4}\cmidrule(l){5-7}
Variable & Odds ratio & 95\% CI & $p$-value & Odds ratio & 95\% CI & $p$-value \\
\midrule
\multicolumn{7}{@{}l}{\textit{Intercept and resource availability}} \\
Intercept & 0.2536 & [0.15612, 0.41192] & 2.97E-08 & 0.21156 & [0.12995, 0.34442] & 4.20E-10 \\
THShare$_{t-1}$ & 1.0524 & [0.7859, 1.4093] & 0.73164 & 0.93123 & [0.68255, 1.2705] & 0.65306 \\
Holiday & 1.0109 & [0.52646, 1.9413] & 0.97393 & 1.0059 & [0.52407, 1.9308] & 0.98587 \\
ES assignments & 0.99977 & [0.98182, 1.018] & 0.97977 & 0.99762 & [0.97962, 1.016] & 0.7977 \\
TH assignments & 0.99885 & [0.97627, 1.0219] & 0.92144 & 1.0208 & [0.99684, 1.0453] & 0.089403 \\
ES exits & 0.96525 & [0.94666, 0.9842] & \textbf{0.00036287} & 0.97053 & [0.95258, 0.98882] & \textbf{0.0016862} \\
\quad per SD & 0.82254 & [0.7388, 0.91579] & \textbf{0.00036287} & 0.84319 & [0.75805, 0.93789] & \textbf{0.0016862} \\
TH exits & 1.0202 & [0.99656, 1.0445] & 0.094423 & 1.0212 & [0.99764, 1.0454] & 0.078135 \\
\quad per SD & 1.0927 & [0.98488, 1.2123] & 0.094423 & 1.0968 & [0.98965, 1.2156] & 0.078135 \\
\addlinespace
\multicolumn{7}{@{}l}{\textit{Weekly cycles (reference: Tue--Fri)}} \\
Mon & 1.1256 & [0.90662, 1.3974] & 0.28381 & 1.1402 & [0.91861, 1.4152] & 0.23406 \\
Weekend & 0.78578 & [0.55433, 1.1139] & 0.17567 & 0.81324 & [0.57366, 1.1529] & 0.24566 \\
\addlinespace
\multicolumn{7}{@{}l}{\textit{Month (reference: October)}} \\
Nov & 0.91139 & [0.56273, 1.4761] & 0.70607 & 0.88028 & [0.54414, 1.4241] & 0.60338 \\
Dec & 0.74647 & [0.45838, 1.2156] & 0.23992 & 0.72972 & [0.44763, 1.1896] & 0.20632 \\
Jan & 1.2291 & [0.7894, 1.9138] & 0.36115 & 1.125 & [0.71977, 1.7583] & 0.60524 \\
Feb & 1.274 & [0.81498, 1.9915] & 0.28808 & 1.1875 & [0.76002, 1.8554] & 0.45042 \\
Mar & 1.1528 & [0.73337, 1.812] & 0.53781 & 1.0625 & [0.67585, 1.6703] & 0.79289 \\
Apr & 1.1349 & [0.72185, 1.7842] & 0.58362 & 1.0611 & [0.67405, 1.6703] & 0.79787 \\
May & 1.1371 & [0.72388, 1.7863] & 0.57705 & 1.1099 & [0.70692, 1.7426] & 0.65053 \\
Jun & 1.1223 & [0.70555, 1.7852] & 0.62611 & 1.0266 & [0.64676, 1.6296] & 0.91127 \\
Jul & 1.3444 & [0.86745, 2.0835] & 0.18556 & 1.2843 & [0.82956, 1.9883] & 0.26185 \\
Aug & 0.97605 & [0.61305, 1.554] & 0.91864 & 0.95199 & [0.59898, 1.513] & 0.83511 \\
Sep & 0.96196 & [0.60527, 1.5288] & 0.86966 & 0.93435 & [0.5879, 1.485] & 0.77392 \\
\bottomrule
\end{tabular}
\end{table}

\begin{table}[htbp]
\centering
\caption{Results of fitting the GLM for rationing, isolating weekly effects.}
\label{tab:rationing_discretion_mon}
\footnotesize
\setlength{\tabcolsep}{4pt}
\renewcommand{\arraystretch}{0.92}
\begin{tabular}{@{}lcccccc@{}}
\toprule
 & \multicolumn{3}{c}{Block features} & \multicolumn{3}{c}{Rolling features} \\
\cmidrule(lr){2-4}\cmidrule(l){5-7}
Variable & Odds ratio & 95\% CI & $p$-value & Odds ratio & 95\% CI & $p$-value \\
\midrule
\multicolumn{7}{@{}l}{\textit{Intercept and resource availability}} \\
Intercept & 0.74218 & [0.46469, 1.1854] & 0.21198 & 0.65816 & [0.41327, 1.0482] & 0.078095 \\
THShare$_{t-1}$ & 1.0518 & [0.79233, 1.3961] & 0.72698 & 1.1504 & [0.85331, 1.551] & 0.35795 \\
Holiday & 0.79096 & [0.40205, 1.5561] & 0.49698 & 0.76989 & [0.38982, 1.5205] & 0.45136 \\
ES assignments & 0.99698 & [0.97782, 1.0165] & 0.75995 & 1.004 & [0.98483, 1.0235] & 0.68639 \\
TH assignments & 0.97263 & [0.94968, 0.99612] & \textbf{0.022679} & 0.98264 & [0.95812, 1.0078] & 0.17449 \\
ES exits & 1.0156 & [0.99651, 1.035] & 0.10985 & 1.0259 & [1.0069, 1.0452] & \textbf{0.0073899} \\
\quad per SD & 1.0891 & [0.9809, 1.2093] & 0.10985 & 1.1569 & [1.0399, 1.287] & \textbf{0.0073899} \\
TH exits & 0.98491 & [0.96048, 1.01] & 0.23544 & 0.95233 & [0.92806, 0.97724] & \textbf{0.00020875} \\
\quad per SD & 0.93493 & [0.83658, 1.0448] & 0.23544 & 0.8066 & [0.71998, 0.90364] & \textbf{0.00020875} \\
\addlinespace
\multicolumn{7}{@{}l}{\textit{Weekly cycles (reference: Tue--Fri)}} \\
Mon & 1.2687 & [1.0305, 1.5621] & \textbf{0.024923} & 1.2414 & [1.0064, 1.5314] & \textbf{0.043472} \\
Weekend & 0.63895 & [0.46477, 0.87839] & \textbf{0.0058081} & 0.63915 & [0.46466, 0.87916] & \textbf{0.0059286} \\
\addlinespace
\multicolumn{7}{@{}l}{\textit{Month (reference: October)}} \\
Nov & 0.66864 & [0.42145, 1.0608] & 0.087393 & 0.67713 & [0.4265, 1.075] & 0.098303 \\
Dec & 0.52615 & [0.33201, 0.8338] & \textbf{0.0062624} & 0.54741 & [0.34478, 0.86912] & \textbf{0.010629} \\
Jan & 0.55675 & [0.36286, 0.85422] & \textbf{0.0073317} & 0.6124 & [0.3975, 0.94348] & \textbf{0.026162} \\
Feb & 0.81131 & [0.52933, 1.2435] & 0.33719 & 0.82134 & [0.53633, 1.2578] & 0.36539 \\
Mar & 0.66791 & [0.43974, 1.0145] & 0.058415 & 0.6953 & [0.45728, 1.0572] & 0.089186 \\
Apr & 0.77213 & [0.50313, 1.1849] & 0.23665 & 0.82588 & [0.5372, 1.2697] & 0.3833 \\
May & 0.82495 & [0.54317, 1.2529] & 0.36677 & 0.84278 & [0.55434, 1.2813] & 0.42358 \\
Jun & 0.52548 & [0.34164, 0.80823] & \textbf{0.0033985} & 0.52213 & [0.33891, 0.80442] & \textbf{0.0032097} \\
Jul & 0.69954 & [0.45494, 1.0756] & 0.10356 & 0.72898 & [0.474, 1.1211] & 0.15005 \\
Aug & 0.72658 & [0.47571, 1.1098] & 0.13939 & 0.73416 & [0.4804, 1.122] & 0.15325 \\
Sep & 0.90581 & [0.59132, 1.3876] & 0.64939 & 0.96224 & [0.62726, 1.4761] & 0.86004 \\
\bottomrule
\end{tabular}
\end{table}

\FloatBarrier
\subsection{Full day-of-week configuration}
\label{appen:full_glm}

Tables~\ref{tab:overall_discretion}, \ref{tab:updrading_discretion} and~\ref{tab:rationing_discretion} replace the day-type indicator with an indicator for each day of the week, whose reference is Monday.

\begin{table}[htbp]
\centering
\caption{Results of fitting the GLM for overall discretion, full day-of-week configuration. The outcome counts each day's overrides among its placements predicted ES or TH.}
\label{tab:overall_discretion}
\footnotesize
\setlength{\tabcolsep}{4pt}
\renewcommand{\arraystretch}{0.92}
\begin{tabular}{@{}lcccccc@{}}
\toprule
 & \multicolumn{3}{c}{Block features} & \multicolumn{3}{c}{Rolling features} \\
\cmidrule(lr){2-4}\cmidrule(l){5-7}
Variable & Odds ratio & 95\% CI & $p$-value & Odds ratio & 95\% CI & $p$-value \\
\midrule
\multicolumn{7}{@{}l}{\textit{Intercept and resource availability}} \\
Intercept & 0.51593 & [0.36515, 0.72898] & 0.00017519 & 0.44808 & [0.31726, 0.63286] & 5.19E-06 \\
THShare$_{t-1}$ & 1.0656 & [0.87294, 1.3008] & 0.5323 & 1.0386 & [0.84097, 1.2826] & 0.72534 \\
Holiday & 0.86017 & [0.54213, 1.3648] & 0.52248 & 0.85026 & [0.53589, 1.349] & 0.49096 \\
ES assignments & 0.99487 & [0.98209, 1.0078] & 0.43526 & 0.99685 & [0.98407, 1.0098] & 0.6322 \\
TH assignments & 0.99075 & [0.97479, 1.007] & 0.26236 & 1.0068 & [0.98993, 1.024] & 0.43105 \\
ES exits & 0.98908 & [0.976, 1.0023] & 0.10564 & 0.99484 & [0.98215, 1.0077] & 0.43003 \\
\quad per SD & 0.94114 & [0.87446, 1.0129] & 0.10564 & 0.97094 & [0.90238, 1.0447] & 0.43003 \\
TH exits & 1.0078 & [0.991, 1.0249] & 0.36539 & 0.99512 & [0.97848, 1.0121] & 0.5701 \\
\quad per SD & 1.0349 & [0.96079, 1.1148] & 0.36539 & 0.97872 & [0.9087, 1.0541] & 0.5701 \\
\addlinespace
\multicolumn{7}{@{}l}{\textit{Day of the week (reference: Monday)}} \\
Tue & 0.91052 & [0.75661, 1.0957] & 0.32108 & 0.91066 & [0.75667, 1.096] & 0.32206 \\
Wed & 0.75151 & [0.61915, 0.91216] & \textbf{0.0038516} & 0.75019 & [0.61799, 0.91068] & \textbf{0.003662} \\
Thu & 0.6826 & [0.55802, 0.835] & \textbf{0.00020411} & 0.67821 & [0.5543, 0.82983] & \textbf{0.00016193} \\
Fri & 0.87231 & [0.71823, 1.0594] & 0.16829 & 0.8712 & [0.71719, 1.0583] & 0.16478 \\
Sat & 0.70326 & [0.53921, 0.91721] & \textbf{0.0093869} & 0.70706 & [0.5422, 0.92206] & \textbf{0.010498} \\
Sun & 0.26723 & [0.13345, 0.53511] & \textbf{0.00019535} & 0.26419 & [0.132, 0.52879] & \textbf{0.00017016} \\
\addlinespace
\multicolumn{7}{@{}l}{\textit{Month (reference: October)}} \\
Nov & 0.70704 & [0.51298, 0.97452] & \textbf{0.034208} & 0.70113 & [0.50917, 0.96546] & \textbf{0.029604} \\
Dec & 0.59249 & [0.42865, 0.81894] & \textbf{0.0015274} & 0.59113 & [0.42753, 0.81734] & \textbf{0.0014723} \\
Jan & 0.79185 & [0.58846, 1.0655] & 0.12335 & 0.78025 & [0.57861, 1.0522] & 0.10381 \\
Feb & 0.94672 & [0.70356, 1.2739] & 0.71772 & 0.92531 & [0.6883, 1.2439] & 0.60716 \\
Mar & 0.84117 & [0.62512, 1.1319] & 0.25346 & 0.81751 & [0.6076, 1.0999] & 0.18325 \\
Apr & 0.87086 & [0.6457, 1.1746] & 0.36499 & 0.86636 & [0.64222, 1.1687] & 0.34765 \\
May & 0.91587 & [0.68173, 1.2304] & 0.5596 & 0.91095 & [0.67843, 1.2232] & 0.53508 \\
Jun & 0.74448 & [0.54845, 1.0106] & 0.058424 & 0.70598 & [0.52072, 0.95714] & \textbf{0.024957} \\
Jul & 0.891 & [0.66316, 1.1971] & 0.44374 & 0.88837 & [0.66189, 1.1923] & 0.43049 \\
Aug & 0.80545 & [0.59543, 1.0895] & 0.16042 & 0.79331 & [0.58717, 1.0718] & 0.13151 \\
Sep & 0.8722 & [0.64561, 1.1783] & 0.37298 & 0.8768 & [0.64937, 1.1839] & 0.39079 \\
\bottomrule
\end{tabular}
\end{table}

\begin{table}[htbp]
\centering
\caption{Results of fitting the GLM for upgrading, full day-of-week configuration.}
\label{tab:updrading_discretion}
\footnotesize
\setlength{\tabcolsep}{4pt}
\renewcommand{\arraystretch}{0.92}
\begin{tabular}{@{}lcccccc@{}}
\toprule
 & \multicolumn{3}{c}{Block features} & \multicolumn{3}{c}{Rolling features} \\
\cmidrule(lr){2-4}\cmidrule(l){5-7}
Variable & Odds ratio & 95\% CI & $p$-value & Odds ratio & 95\% CI & $p$-value \\
\midrule
\multicolumn{7}{@{}l}{\textit{Intercept and resource availability}} \\
Intercept & 0.28191 & [0.16787, 0.47341] & 1.69E-06 & 0.23865 & [0.14182, 0.40162] & 6.84E-08 \\
THShare$_{t-1}$ & 1.0522 & [0.7856, 1.4094] & 0.7327 & 0.92917 & [0.68095, 1.2679] & 0.64314 \\
Holiday & 1.0059 & [0.52395, 1.9312] & 0.98589 & 1.0019 & [0.52211, 1.9225] & 0.99553 \\
ES assignments & 0.99998 & [0.98202, 1.0183] & 0.9986 & 0.99736 & [0.97935, 1.0157] & 0.77631 \\
TH assignments & 0.99794 & [0.97532, 1.0211] & 0.85997 & 1.0207 & [0.99667, 1.0453] & 0.091831 \\
ES exits & 0.96527 & [0.94666, 0.98424] & \textbf{0.0003715} & 0.97106 & [0.95312, 0.98933] & \textbf{0.0020198} \\
\quad per SD & 0.82263 & [0.73877, 0.916] & \textbf{0.0003715} & 0.8458 & [0.76049, 0.94068] & \textbf{0.0020198} \\
TH exits & 1.021 & [0.99726, 1.0453] & 0.083309 & 1.0213 & [0.99764, 1.0455] & 0.077999 \\
\quad per SD & 1.0964 & [0.98795, 1.2167] & 0.083309 & 1.0971 & [0.98967, 1.2162] & 0.077999 \\
\addlinespace
\multicolumn{7}{@{}l}{\textit{Day of the week (reference: Monday)}} \\
Tue & 0.96777 & [0.73978, 1.266] & 0.8111 & 0.96421 & [0.7373, 1.261] & 0.79009 \\
Wed & 0.73086 & [0.54716, 0.97622] & \textbf{0.033763} & 0.72047 & [0.53954, 0.96207] & \textbf{0.026283} \\
Thu & 0.83432 & [0.6269, 1.1104] & 0.21419 & 0.82063 & [0.61664, 1.0921] & 0.1752 \\
Fri & 1.0351 & [0.78323, 1.3679] & 0.80856 & 1.0164 & [0.76888, 1.3436] & 0.90906 \\
Sat & 0.63362 & [0.41744, 0.96175] & \textbf{0.032101} & 0.65366 & [0.43083, 0.99175] & \textbf{0.045619} \\
Sun & 1.1682 & [0.52533, 2.598] & 0.70296 & 1.1188 & [0.5038, 2.4848] & 0.78266 \\
\addlinespace
\multicolumn{7}{@{}l}{\textit{Month (reference: October)}} \\
Nov & 0.92421 & [0.56991, 1.4987] & 0.74931 & 0.88905 & [0.549, 1.4397] & 0.63255 \\
Dec & 0.75408 & [0.46276, 1.2288] & 0.25723 & 0.73494 & [0.45054, 1.1989] & 0.2174 \\
Jan & 1.2419 & [0.79685, 1.9354] & 0.33863 & 1.1316 & [0.72329, 1.7703] & 0.58827 \\
Feb & 1.2952 & [0.82775, 2.0267] & 0.25746 & 1.2044 & [0.77025, 1.8832] & 0.41485 \\
Mar & 1.1812 & [0.75061, 1.8588] & 0.47158 & 1.086 & [0.69014, 1.7089] & 0.72134 \\
Apr & 1.1433 & [0.72673, 1.7987] & 0.5624 & 1.0667 & [0.6772, 1.6802] & 0.78064 \\
May & 1.1386 & [0.72432, 1.7899] & 0.57376 & 1.1082 & [0.70536, 1.7411] & 0.65582 \\
Jun & 1.1323 & [0.71137, 1.8022] & 0.60041 & 1.0328 & [0.6505, 1.6399] & 0.89102 \\
Jul & 1.3544 & [0.87266, 2.1019] & 0.17618 & 1.2901 & [0.83226, 1.9997] & 0.25474 \\
Aug & 0.98789 & [0.61995, 1.5742] & 0.95911 & 0.95976 & [0.6034, 1.5266] & 0.86228 \\
Sep & 0.96974 & [0.60957, 1.5427] & 0.89679 & 0.93984 & [0.59078, 1.4951] & 0.79337 \\
\bottomrule
\end{tabular}
\end{table}

\begin{table}[htbp]
\centering
\caption{Results of fitting the GLM for rationing, full day-of-week configuration.}
\label{tab:rationing_discretion}
\footnotesize
\setlength{\tabcolsep}{4pt}
\renewcommand{\arraystretch}{0.92}
\begin{tabular}{@{}lcccccc@{}}
\toprule
 & \multicolumn{3}{c}{Block features} & \multicolumn{3}{c}{Rolling features} \\
\cmidrule(lr){2-4}\cmidrule(l){5-7}
Variable & Odds ratio & 95\% CI & $p$-value & Odds ratio & 95\% CI & $p$-value \\
\midrule
\multicolumn{7}{@{}l}{\textit{Intercept and resource availability}} \\
Intercept & 0.93726 & [0.56892, 1.5441] & 0.79918 & 0.80186 & [0.4874, 1.3192] & 0.38466 \\
THShare$_{t-1}$ & 1.0063 & [0.75549, 1.3403] & 0.96593 & 1.0967 & [0.81044, 1.484] & 0.5499 \\
Holiday & 0.77958 & [0.39582, 1.5354] & 0.4715 & 0.76008 & [0.38439, 1.503] & 0.43032 \\
ES assignments & 0.99617 & [0.97692, 1.0158] & 0.69997 & 1.0055 & [0.98611, 1.0252] & 0.5819 \\
TH assignments & 0.97677 & [0.95351, 1.0006] & 0.055911 & 0.98536 & [0.96065, 1.0107] & 0.25512 \\
ES exits & 1.0172 & [0.99797, 1.0368] & 0.079868 & 1.0264 & [1.0072, 1.046] & \textbf{0.0068166} \\
\quad per SD & 1.0987 & [0.98885, 1.2208] & 0.079868 & 1.1604 & [1.0419, 1.2925] & \textbf{0.0068166} \\
TH exits & 0.98276 & [0.95818, 1.008] & 0.1783 & 0.95115 & [0.92662, 0.97632] & \textbf{0.00017158} \\
\quad per SD & 0.92592 & [0.82776, 1.0357] & 0.1783 & 0.80218 & [0.71506, 0.89992] & \textbf{0.00017158} \\
\addlinespace
\multicolumn{7}{@{}l}{\textit{Day of the week (reference: Monday)}} \\
Tue & 0.94873 & [0.72376, 1.2436] & 0.70309 & 0.97212 & [0.74027, 1.2766] & 0.83881 \\
Wed & 0.8095 & [0.61511, 1.0653] & 0.13146 & 0.81628 & [0.61899, 1.0765] & 0.15042 \\
Thu & 0.60038 & [0.44575, 0.80863] & \textbf{0.00078508} & 0.61642 & [0.45658, 0.83221] & \textbf{0.0015816} \\
Fri & 0.77587 & [0.58381, 1.0311] & 0.080311 & 0.79535 & [0.59735, 1.059] & 0.11698 \\
Sat & 0.74236 & [0.51675, 1.0665] & 0.10702 & 0.76533 & [0.53192, 1.1012] & 0.14963 \\
Sun & 0.02654 & [0.0036494, 0.19301] & \textbf{0.00033708} & 0.026335 & [0.0036196, 0.19161] & \textbf{0.00032845} \\
\addlinespace
\multicolumn{7}{@{}l}{\textit{Month (reference: October)}} \\
Nov & 0.6548 & [0.41116, 1.0428] & 0.074522 & 0.66835 & [0.41937, 1.0652] & 0.090163 \\
Dec & 0.52942 & [0.33328, 0.84099] & \textbf{0.0070731} & 0.5479 & [0.34428, 0.87196] & \textbf{0.011152} \\
Jan & 0.55142 & [0.35876, 0.84754] & \textbf{0.0066427} & 0.60382 & [0.3911, 0.93224] & \textbf{0.02281} \\
Feb & 0.81319 & [0.52885, 1.2504] & 0.34621 & 0.82308 & [0.53581, 1.2644] & 0.37401 \\
Mar & 0.66739 & [0.43831, 1.0162] & 0.059432 & 0.69384 & [0.45511, 1.0578] & 0.089357 \\
Apr & 0.76344 & [0.49625, 1.1745] & 0.21939 & 0.81723 & [0.53023, 1.2596] & 0.3605 \\
May & 0.84075 & [0.55147, 1.2818] & 0.42013 & 0.85311 & [0.55915, 1.3016] & 0.4611 \\
Jun & 0.51989 & [0.33735, 0.80121] & \textbf{0.0030332} & 0.52079 & [0.33747, 0.80371] & \textbf{0.0032082} \\
Jul & 0.68562 & [0.44465, 1.0572] & 0.087569 & 0.71611 & [0.4642, 1.1047] & 0.13114 \\
Aug & 0.70711 & [0.46152, 1.0834] & 0.11136 & 0.71716 & [0.46774, 1.0996] & 0.12735 \\
Sep & 0.90026 & [0.58582, 1.3835] & 0.63173 & 0.9493 & [0.61692, 1.4608] & 0.81296 \\
\bottomrule
\end{tabular}
\end{table}

\FloatBarrier
\subsection{Remaining slack over the fiscal year}
\label{appen:remaining_slack}

The model allows sensitivity to remaining budget to strengthen over part of the horizon, but not necessarily through its end: once budget is ample, the spending probability is $1-p$ regardless of the positive budget left.
We therefore treat a remaining-slack interaction with fiscal position as a descriptive check for strengthening capacity associations. We test this by interacting a measure of remaining slack with the position in the fiscal year. Remaining slack on day $t$ is the cumulative number of transitional housing (TH) exits minus TH entries on earlier days of the same fiscal year (October 1 to September 30), minus its average for the same fiscal week across years, standardized to mean zero and standard deviation (SD) one. Fiscal position is the number of days since October 1 divided by 365. We add slack to the rolling day-type model of Table~\ref{tab:discretion_summary_daytype} (the slack-only model) and then add its product with fiscal position (the interaction model). The sample is the 1,329 days with ES or TH placements in fiscal years 2009--2013; fiscal year 2008 is excluded because the daily series begins partway through it. Because 26.1\% of TH entries are upgrades, net slack partly reflects past overrides, so we also report a version built from TH exits alone.

\begin{table}[htbp]
\centering
\caption{Slack terms of the remaining-slack test, fiscal years 2009 to 2013. Both models are the rolling day-type model of Table~\ref{tab:discretion_summary_daytype}: the slack-only model adds slack, and the interaction model adds slack and the product of fiscal position and slack; the overall model counts overrides among placements predicted ES or TH. Slack is standardized, so its odds ratios are per SD, and the interaction is the factor by which the effect of slack changes over a full fiscal year. The slack term of the interaction model, the effect of slack on October 1, is not shown, because slack is close to 0 then by construction. $p$-values below 0.05 are in bold.}
\label{tab:slack_test}
\footnotesize
\setlength{\tabcolsep}{4pt}
\begin{tabular}{@{}lccccccc@{}}
\toprule
 & & \multicolumn{3}{c}{Slack} & \multicolumn{3}{c}{Fiscal position $\times$ slack} \\
 & & \multicolumn{3}{c}{(slack-only model)} & \multicolumn{3}{c}{(interaction model)} \\
\cmidrule(lr){3-5}\cmidrule(l){6-8}
Model & Days & Odds ratio & 95\% CI & $p$-value & Odds ratio & 95\% CI & $p$-value \\
\midrule
\multicolumn{8}{@{}l}{\textit{Net slack (TH exits minus TH entries)}} \\
Overall & 1,322 & 0.953 & [0.877, 1.035] & 0.250 & 1.083 & [0.826, 1.420] & 0.562 \\
Upgrading & 1,191 & 0.879 & [0.779, 0.991] & \textbf{0.036} & 1.284 & [0.862, 1.912] & 0.219 \\
Rationing & 979 & 1.112 & [0.983, 1.258] & 0.092 & 0.806 & [0.540, 1.202] & 0.289 \\
\addlinespace
\multicolumn{8}{@{}l}{\textit{Exits only (TH exits)}} \\
Overall & 1,322 & 0.988 & [0.913, 1.069] & 0.757 & 1.013 & [0.675, 1.520] & 0.952 \\
Upgrading & 1,191 & 1.012 & [0.905, 1.133] & 0.830 & 0.775 & [0.431, 1.393] & 0.394 \\
Rationing & 979 & 0.908 & [0.804, 1.025] & 0.119 & 1.792 & [0.971, 3.305] & 0.062 \\
\bottomrule
\end{tabular}
\end{table}

Table~\ref{tab:slack_test} reports the results. The interaction is above one for overall discretion (1.083, $p=0.562$) and upgrading (1.284, $p=0.219$) and below one for rationing (0.806, $p=0.289$), and no interaction is significant at the 5\% level under either measure of slack. Adding slack also leaves the calendar pattern largely unchanged: the month coefficients of the overall model shrink by 3.7\% in total, and no month odds ratio moves by more than 0.02. The one significant slack term, less upgrading with more net slack (0.879, $p=0.036$), disappears with exits alone (1.012, $p=0.830$), consistent with net slack partly reflecting past upgrades.

We find no evidence that the association between remaining slack and overrides strengthens over the fiscal year. These data therefore cannot separate budget rationing from calendar-based explanations, and we report the calendar patterns in Section~\ref{sec:empirical} as descriptive rather than as evidence about the shape of the gain distribution.

\FloatBarrier


\newpage
\section*{NeurIPS Paper Checklist}

\begin{enumerate}

\item {\bf Claims}
    \item[] Question: Do the main claims made in the abstract and introduction accurately reflect the paper's contributions and scope?
    \item[] Answer: \answerYes{} 
    \item[] Justification: The abstract and introduction (Section \ref{sec:introduction}) clearly state the three main contributions and the main body of the paper (in conjunction with the appendix) expand on each of those contributions. 

\item {\bf Limitations}
    \item[] Question: Does the paper discuss the limitations of the work performed by the authors?
    \item[] Answer: \answerYes{} 
    \item[] Justification: The paper includes discussion of limitations in Section \ref{sec:discussion-conclusion}. All assumptions that the theory rests on are clearly pointed out throughout the paper. 

\item {\bf Theory assumptions and proofs}
    \item[] Question: For each theoretical result, does the paper provide the full set of assumptions and a complete (and correct) proof?
    \item[] Answer:\answerYes{} 
    \item[] Justification: Both main theorems (Theorem \ref{thm:general-invariance} and \ref{thm:shape_dependence}) state their assumptions explicitly. Proof sketches appear in the main text and complete proofs are provided in Appendix \ref{app:proofs} and Appendix  \ref{app:two_period_case}. Supporting results are numbered, cross-referenced, and proved in the appendix.

    \item {\bf Experimental result reproducibility}
    \item[] Question: Does the paper fully disclose all the information needed to reproduce the main experimental results of the paper to the extent that it affects the main claims and/or conclusions of the paper (regardless of whether the code and data are provided or not)?
    \item[] Answer: \answerYes{} 
    \item[] Justification: For the theoretical simulations, all distribution families, parameterizations, and sample sizes ($N=10^6$) are reported in Appendix \ref{app:verification}. For the empirical analysis, the GLM specification, feature construction (rolling 7-day windows), aggregation procedure (daily binomial), and all covariates are described in Section \ref{sec:empirical} and Appendix \ref{app:empirical-details}. The decision-tree baseline follows the SER-DT method as used by \citet{pokharel2024discretionary} with max depth 4 on 34 features (listed in Table \ref{appen:data_features}). The HMIS dataset is administrative and not publicly releasable, but all methodology is fully specified.

\item {\bf Open access to data and code}
    \item[] Question: Does the paper provide open access to the data and code, with sufficient instructions to faithfully reproduce the main experimental results, as described in supplemental material?
    \item[] Answer: \answerNo{} 
    \item[] Justification: The HMIS data is protected administrative data that we use under a specific data sharing agreement. It cannot be publicly released due to privacy restrictions. Simulation for the theoretical results (scale invariance, shape dependence) can be fully reproduced from the specifications in Appendix \ref{app:verification}.

\item {\bf Experimental setting/details}
    \item[] Question: Does the paper specify all the training and test details (e.g., data splits, hyperparameters, how they were chosen, type of optimizer) necessary to understand the results?
    \item[] Answer: \answerYes{} 
    \item[] Justification: Simulation details (distributions, parameterizations, sample sizes) are in Appendix \ref{app:verification}. The empirical GLM specification is fully described in Appendix \ref{subsec:glm_spec}, including the link function (logit), weighting scheme (trial counts), covariates, and reference categories. Feature engineering is detailed in Appendix \ref{subsec:feature_engineering}. Computing infrastructure is reported in the beginning of Appendix \ref{app:empirical-details}.

\item {\bf Experiment statistical significance}
    \item[] Question: Does the paper report error bars suitably and correctly defined or other appropriate information about the statistical significance of the experiments?
    \item[] Answer: \answerYes{} 
    \item[] Justification: The empirical results in Table \ref{tab:discretion_summary_daytype} report odds ratios with 95\% confidence intervals and significance levels ($p<0.001$, $p<0.01$, $p<0.05$, $p<0.1$). Full coefficient tables with confidence intervals and $p$-values appear in Appendix \ref{appen:glm_emperics}. 

\item {\bf Experiments compute resources}
    \item[] Question: For each experiment, does the paper provide sufficient information on the computer resources (type of compute workers, memory, time of execution) needed to reproduce the experiments?
    \item[] Answer: \answerYes{} 
    \item[] Justification: Computing infrastructure is reported in the beginning of Appendix \ref{app:empirical-details}.
    
\item {\bf Code of ethics}
    \item[] Question: Does the research conducted in the paper conform, in every respect, with the NeurIPS Code of Ethics \url{https://neurips.cc/public/EthicsGuidelines}?
    \item[] Answer: \answerYes{} 
    \item[] Justification: The research uses de-identified administrative HMIS records and does not involve direct interaction with human subjects. The work aims to improve understanding of discretionary decision-making in social services, with potential benefits for more equitable resource allocation.

\item {\bf Broader impacts}
    \item[] Question: Does the paper discuss both potential positive societal impacts and negative societal impacts of the work performed?
    \item[] Answer: \answerYes{} 
    \item[] Justification: Section \ref{sec:discussion-conclusion} discusses implications for decision-support system design, noting that systems should help manage not only which cases merit discretion but also when to deploy it. The paper frames discretion as potentially welfare-improving when treated as an explicitly budgeted resource.
    
\item {\bf Safeguards}
    \item[] Question: Does the paper describe safeguards that have been put in place for responsible release of data or models that have a high risk for misuse (e.g., pre-trained language models, image generators, or scraped datasets)?
    \item[] Answer: \answerNA{} 
    \item[] Justification: The paper does not release models, datasets, or tools with high risk for misuse. The contribution is a theoretical framework and empirical analysis of existing administrative data.

\item {\bf Licenses for existing assets}
    \item[] Question: Are the creators or original owners of assets (e.g., code, data, models), used in the paper, properly credited and are the license and terms of use explicitly mentioned and properly respected?
    \item[] Answer: \answerYes{} 
    \item[] Justification: The HMIS data is cited via relevant prior work \citep{kube2023fair, pokharel2024discretionary}. We also cite the code we use: the SER-DT reference implementation \citep{SERDT} (\url{https://github.com/victorfcsouza/explainable-dts}, commit 86d017e, run unchanged; the repository does not state a license), GOSDT-guesses \citep{mctavish2022fast} (\texttt{gosdt} 1.0.4, BSD 3-Clause license), and XGBoost \citep{chen2016xgboost} (version 3.4.1, Apache 2.0 license). No pre-trained models are used.

\item {\bf New assets}
    \item[] Question: Are new assets introduced in the paper well documented and is the documentation provided alongside the assets?
    \item[] Answer: \answerNA{} 
    \item[] Justification: The paper does not release new datasets, models, or code packages.

\item {\bf Crowdsourcing and research with human subjects}
    \item[] Question: For crowdsourcing experiments and research with human subjects, does the paper include the full text of instructions given to participants and screenshots, if applicable, as well as details about compensation (if any)? 
    \item[] Answer: \answerNA{} 
    \item[] Justification: The paper does not involve crowdsourcing or direct research with human subjects.

\item {\bf Institutional review board (IRB) approvals or equivalent for research with human subjects}
    \item[] Question: Does the paper describe potential risks incurred by study participants, whether such risks were disclosed to the subjects, and whether Institutional Review Board (IRB) approvals (or an equivalent approval/review based on the requirements of your country or institution) were obtained?
    \item[] Answer:\answerNA{} 
    \item[] Justification: The study analyzes pre-existing de-identified administrative data and does not constitute human subjects research requiring IRB review. No interventions were performed on individuals.

\item {\bf Declaration of LLM usage}
    \item[] Question: Does the paper describe the usage of LLMs if it is an important, original, or non-standard component of the core methods in this research? Note that if the LLM is used only for writing, editing, or formatting purposes and does \emph{not} impact the core methodology, scientific rigor, or originality of the research, declaration is not required.
    \item[] Answer: \answerNA{} 
    \item[] Justification: LLMs were not used as a component of the core methodology, theoretical development, or empirical analysis in this research.

\end{enumerate}

\end{document}